\documentclass[final,11pt]{article}

\usepackage{comment,amsmath,amssymb,amsfonts} 
\usepackage{epsfig} \usepackage{latexsym,nicefrac,bbm}
\usepackage{xspace}
\usepackage{color,fancybox,graphicx,url,subfigure,fullpage}
\usepackage[top=1.1in, bottom=1.1in, left=1in, right=1in]{geometry}
\usepackage{tabularx} 
\usepackage{mathpazo}

\usepackage{enumitem}
\usepackage{booktabs}

\newcommand{\Span}{{\ensuremath{\mathsf{Span}}\xspace}}
\newcommand{\defeq}{\stackrel{\text{\tiny def}}{=}}
\newcommand{\calK}{\mathcal{K}} \newcommand{\nfrac}{\nicefrac}
 \newcommand{\BS}{{\sc Balanced
Separator}\xspace} 
\newcommand{\SDP}{{\sc SDP}\xspace} \newcommand{\eps}{\epsilon}
\newcommand{\alg}{{\sc BalSep}\xspace}
\newcommand{\avg}{{\ensuremath{\mathsf{avg}}\xspace}}
\newcommand{\vol}{{\ensuremath{\mathsf{vol}}\xspace}}

\newcommand{\expRational}{{\sc ExpRational}}
\newcommand{\lanczos}{{\sc Lanczos}}
\renewcommand{\epsilon}{\varepsilon}
\newcommand{\calI}{\ensuremath{\mathcal{I}}}
\newcommand{\Invert}{{\ensuremath{\mathsf{Invert}}\xspace}}

\newcommand{\Vol}[1]{\text{vol}(#1)}

\def\showauthornotes{1} 
 
\def\showdraftbox{1}

\newtheorem{theorem}{Theorem}[section]

\newtheorem{definition}[theorem]{Definition}
\newtheorem{lemma}[theorem]{Lemma}
\newtheorem{remark}[theorem]{Remark}
\newtheorem{proposition}[theorem]{Proposition}
\newtheorem{corollary}[theorem]{Corollary}

\newtheorem{fact}[theorem]{Fact}

\def\FullBox{\hbox{\vrule width 6pt height 6pt depth 0pt}}

\def\qed{\ifmmode\qquad\FullBox\else{\unskip\nobreak\hfil
\penalty50\hskip1em\null\nobreak\hfil\FullBox
\parfillskip=0pt\finalhyphendemerits=0\endgraf}\fi}

\def\qedsketch{\ifmmode\Box\else{\unskip\nobreak\hfil
\penalty50\hskip1em\null\nobreak\hfil$\Box$
\parfillskip=0pt\finalhyphendemerits=0\endgraf}\fi}

\newenvironment{proof}{\begin{trivlist} \item {\bf Proof:~~}}
  {\qed\end{trivlist}}

\newcommand\rea{\mathbb R}

\newcommand\R{\mathbb R}

\newcommand{\marginlabel}[1]%
{\mbox{}\marginpar{\it{\raggedleft\hspace{0pt}#1}}}
\newcommand\card[1]{\left| #1 \right|} 
\newcommand{\poly}{\mathrm{poly}}

\newcommand{\ceil}[1]{\left\lceil\, {#1}\,\right\rceil}

\definecolor{Mygray}{gray}{0.8}

 \ifcsname ifcommentflag\endcsname\else
  \expandafter\let\csname ifcommentflag\expandafter\endcsname
                  \csname iffalse\endcsname
\fi

\ifnum\showauthornotes=1

\else

\fi

\ifnum\showauthornotes=1
\newcommand{\Authornote}[2]{{\sf\small\color{red}{[#1: #2]}}}
\newcommand{\Authoredit}[2]{{\sf\small\color{red}{[#1]}\color{blue}{#2}}}
\newcommand{\Authorcomment}[2]{{\sf \small\color{gray}{[#1: #2]}}}
\newcommand{\Authorfnote}[2]{\footnote{\color{red}{#1: #2}}}
\newcommand{\Authorfixme}[1]{\Authornote{#1}{\textbf{??}}}
\newcommand{\Authormarginmark}[1]{\marginpar{\textcolor{red}{\fbox{
#1:!}}}}
\else
\newcommand{\Authornote}[2]{}
\newcommand{\Authoredit}[2]{}
\newcommand{\Authorcomment}[2]{}
\newcommand{\Authorfnote}[2]{}
\newcommand{\Authorfixme}[1]{}
\newcommand{\Authormarginmark}[1]{}
\fi

\newcommand{\from}{:}

\newcommand{\norm}[1]{\ensuremath{\left\lVert #1 \right\rVert}}

\newlength{\pgmtab}  
 {
	\begin{enumerate}}{\end{enumerate}}

\newcounter{lecnum}

\newlength{\tpush}
\ifnum\showdraftbox=1

\else

\fi

\newcommand{\Tr}{{\rm Tr}}
\newcommand{\cH}{\mathcal H}

\newcommand{\mwu}{\textsc{MWU}\xspace}
\newcommand{\mmwu}{\textsc{MMWU}\xspace}

\newcommand{\Degin}{D^{-\nfrac{1}{2}}}

\newcommand{\supp}{\mathrm{supp}}

\title{\bf Approximating the Exponential, 
 the Lanczos Method and an $\tilde{O}(m)$-Time Spectral Algorithm for Balanced Separator  }
\author{Lorenzo Orecchia \\
{\small MIT}\\
{\small Cambridge, MA, USA.} \\
{\small orecchia@mit.edu}
\and Sushant Sachdeva\thanks{This work was done while this author was interning at Microsoft Research India, Bangalore.}   \\ 
 {\small Princeton University} \\
{\small Princeton, NJ, USA.}\\
{\small  sachdeva@cs.princeton.edu}
 \and Nisheeth K. Vishnoi
     \\ {\small Microsoft Research}  \\
{\small Bangalore, India}\\
{\small nisheeth.vishnoi@gmail.com}}

\date{}
\begin{document}

\maketitle

\begin{abstract}
We give a novel spectral approximation algorithm for the balanced separator problem that, given a graph $G$, a constant balance $b \in (0,1/2],$ and a parameter $\gamma,$ either finds an $\Omega(b)$-balanced cut of conductance $O(\sqrt{\gamma})$ in $G,$ or outputs a certificate that all $b$-balanced cuts in $G$ have conductance at least $\gamma,$ and runs in time $\tilde{O}(m).$ This settles the question of designing asymptotically optimal spectral algorithms for balanced separator. 
Our algorithm relies on a variant of the heat kernel random walk and
requires, as a subroutine, an algorithm to compute $\exp(-L)v$
where $L$ is the Laplacian of a graph related to $G$ and $v$ is a
vector. Algorithms for computing the matrix-exponential-vector product efficiently comprise our next set of results. 
Our main result here is a new algorithm which computes a good
approximation to $\exp(-A)v$ for a class of symmetric positive
semidefinite (PSD) matrices $A$ and a given vector $u,$ in time roughly $\tilde{O}(m_A),$ where $m_A$ is the number of non-zero entries of $A.$ This uses, in a non-trivial way, the breakthrough result of Spielman and Teng on inverting symmetric and diagonally-dominant matrices in $\tilde{O}(m_A)$ time.
Finally, we prove that $e^{-x}$ can be uniformly approximated up to a small additive error, in a non-negative interval $[a,b]$ with a polynomial of degree roughly $\sqrt{b-a}.$ 
 While this result is of independent interest in approximation theory, we show that, via the Lanczos method from numerical analysis, it yields a simple algorithm to compute $\exp(-A)v$ for symmetric PSD matrices that runs in time roughly $O(t_A\cdot \sqrt{\norm{A}}),$ where $t_A$ is time required for the
computation of the vector $Aw$ for given vector $w.$ As an application, we obtain a simple and practical algorithm, with output conductance $O(\sqrt{\gamma}),$ for balanced separator that runs in time $\tilde{O}(\nfrac{m}{\sqrt{\gamma}}).$ This latter algorithm matches the running time, but improves on the approximation guarantee of the Evolving-Sets-based algorithm by Andersen and Peres for balanced separator.

\vspace{5mm}

\end{abstract}

\newpage

\tableofcontents

\newpage

\section{Introduction and Our Results}
 
\subsection{Balanced Separator}

The {\BS} problem ({\sc BS}) asks the following decision question: 
given an unweighted graph $G=(V,E),$ $V =[n], \card{E} = m,$ 
a constant balance parameter $b \in (0,\nfrac{1}{2}],$ and a target conductance value $\gamma \in (0,1),$
 does $G$ have a $b$-balanced cut $S$ such that $\phi(S) \leq \gamma$?
\renewcommand{\BS}{{\sc BS}\xspace}
{Here, the conductance of
a cut $(S,\bar{S})$ is defined to be $\phi(S) \defeq \nfrac{|E(S,\bar{S})|}{\min
\{\vol(S),\vol (\overline{S})\}},$ where $\vol(S)$ is the sum of the degrees of the vertices in the set $S$. 
Moreover, a cut $(S, \bar{S})$ is $b$-balanced if $\min \{\vol(S), \vol(\bar{S})\} \geq b \cdot \Vol{V}.$}
This is a classic NP-hard problem and a central object of study for the development of
approximation algorithms, both in theory and in practice. 
On the theoretical side, \BS has far reaching connections to spectral graph theory,
 the study of random walks and metric embeddings.
In practice, algorithms for \BS play a crucial role in the design of recursive algorithms~\cite{Shmoys}, 
clustering~\cite{Kannan} and scientific computation~\cite{Schloegel00}.

Spectral methods are an important set of techniques in the design of graph-partitioning algorithms and are fundamentally based on the study of the behavior of random walks over the instance graph.
Spectral algorithms tend to be conceptually appealing, because of the intuition based on 
the underlying diffusion process, and easy to implement, as many of the primitives required,
such as eigenvector computation, already appear in highly-optimized software packages.
The most important spectral algorithm for graph partitioning is the Laplacian Eigenvector ({\sc LE}) algorithm of Alon and Milman~\cite{Alon}, which, given a graph of conductance at most $\gamma,$ outputs a cut of conductance at most $O(\sqrt{\gamma}),$ an approximation guarantee that is asymptotically optimal for spectral algorithms.
A consequence of the seminal work of Spielman and Teng \cite{ST3} is that the {\sc LE} algorithm can run in time $\tilde{O}(m)$ 
using the Spielman-Teng solver.
Hence, {\sc LE} is an asymptotically optimal spectral algorithm for the minimum-conductance problem, both for running time (up to polylog factors) and approximation quality.
In this paper, we present a simple random-walk-based algorithm that is the first such asymptotically optimal spectral algorithm for \BS. Our algorithm can be seen as an analogue to the {\sc LE} algorithm for the balanced version of the minimum-conductance problem and settles the question of designing spectral algorithms for \BS. The following is our main theorem on graph partitioning.
\begin{theorem}[Spectral Algorithm for Balanced Separator]\label{thm:bsmain}
Given an unweighted graph $G=(V,E)$, a balance parameter $b \in
(0,\nfrac{1}{2}], \; b = \Omega(1)$ and a conductance value $\gamma
\in (0,1),$ we give an algorithm called {\alg}$(G,b, \gamma),$ that either
outputs an $\Omega(b)$-balanced cut $S \subset V$ such that $\phi(S)
\leq O(\sqrt{\gamma}),$ or outputs a certificate that no $b$-balanced
cut of conductance $\gamma$ exists. \alg runs in time ${O}(m \;\poly(\log
n)).$ 
\end{theorem}
The algorithm for Theorem \ref{thm:bsmain} relies on our ability to
compute the product of the matrix-exponential of a matrix and an
arbitrary vector in time essentially proportional to the sparsity of
the matrix. 
Our contribution to the problem of computing the matrix-exponential-vector product appear in detail in Section \ref{sec:exp-intro}. 
The algorithm required for Theorem \ref{thm:bsmain}
runs in time $\tilde{O}(m)$ and, notably, makes use of the Spielman-Teng solver in a non-trivial way. We also prove an
alternative novel result on how to perform this matrix-exponential
computation, which relies just on matrix-vector products. This result,
when combined with our {\BS} algorithm based on random walks, yields a
theorem identical to Theorem \ref{thm:bsmain} except that the running time now
increases to $\tilde{O}(\nfrac{m}{\sqrt{\gamma}}),$ see Theorem
\ref{thm:bsmain2}. However, this latter algorithm not only turns out to be 
 almost as simple and practical as the LE algorithm, but it also improves
in the approximation factor upon the result of Andersen and Peres
\cite{AP} who obtain the same running time using Evolving-Sets-based random walk.

\subsubsection{Comparison to Previous Work on Balanced Separator}
The best known approximation for \BS is $O(\sqrt{\log n})$
achieved by the seminal work of Arora, Rao and Vazirani~\cite{ARV} 
that combines semidefinite programming (\SDP) and flow ideas. 
A rich line of research has centered on reducing the running time of this algorithm using \SDP and flow ideas~\cite{KRV, AK, OSVV}.
This effort culminated in Sherman's work~\cite{Sherman}, which brings down the required running time to $O(n^\eps)$ $s$-$t$ maximum-flow computations.
\footnote{Even though the results of~\cite{ARV} and ~\cite{Sherman} are stated for the Sparsest Cut problem, the same techniques apply to the conductance problem, e.g. by modifying the underlying flow problems. See for example~\cite{AL}.}
However, these algorithms are based on advanced theoretical ideas that
are not easy to implement or even capture in a principled
heuristic. Moreover, they fail to achieve a nearly-linear
\footnote{Following the convention of~\cite{ST2}, we denote by nearly-linear a running time of $\tilde{O}(\nfrac{m}{\poly(\gamma)}).$}
running time, which is crucial in many of today's applications that involve very large graphs.
To address these issues, researchers have focused on the design of simple, nearly-linear-time algorithms for \BS based on spectral techniques. The simplest spectral algorithm for \BS is the Recursive Laplacian Eigenvector ({\sc RLE}) algorithm (see, for example, \cite{Kannan}). 
This algorithm iteratively uses {\sc LE} to remove low-conductance unbalanced cuts from $G,$
until a balanced cut or an induced $\gamma$-expander is found.
The running time of the {\sc RLE} algorithm is quadratic in the worst case, as $\Omega(n)$ unbalanced cuts may be found, each requiring a global computation of the eigenvector.
Spielman and Teng~\cite{ST1} were the first to design nearly-linear-time algorithms outputting an $\Omega(b)$-balanced cut of conductance $O(\sqrt{\gamma \; \textrm{polylog} n}),$ if a $b$-balanced cut of conductance less than $\gamma$ exists.
Their algorithmic approach is based on local random walks, which are
used to remove unbalanced cuts in time proportional to the size of the
cut removed, hence avoiding the quadratic dependence of {\sc RLE}. 
Using similar ideas, Andersen, Chung and Lang~\cite{ACL}, and Andersen and Peres~\cite{AP} improved the approximation guarantee to $O(\sqrt{\gamma \; \log n})$ and the running time to $\tilde{O}(\nfrac{m}{\sqrt{\gamma}}).$
More recently, Orecchia and Vishnoi (OV) \cite{OV} employed an \SDP formulation of the problem, together with the Matrix Multiplicative Weight Update (\mmwu) of~\cite{AK} and a new \SDP rounding, to obtain an output conductance of $O(\sqrt{\gamma})$ with running time $\tilde{O}(\nfrac{m}{\gamma^2})$, effectively removing unbalanced cuts in $O(\nfrac{\log n}{\gamma})$ iterations. In Section~\ref{sec:sdpinter}, we give a more detailed comparison with OV and discussion of our novel width-reduction techniques from an optimization point of view.
Finally, our algorithm should also be compared to the remarkable
results of Madry~\cite{Madry10} for BS, which build up on R\"{a}cke's work \cite{Racke} and on the low-stretch spanning trees from Abraham {\em et
al.}~\cite{ABN}, to achieve a trade-off between running time and
approximation. For every integer $k \ge 1,$ he achieves roughly
$O((\log n)^k)$ approximation in time $\tilde{O}(m+2^k\cdot
n^{1+2^{-k}}).$
Calculations show that for $\gamma \geq 2^{-(\log \log n)^2},$ our
algorithm achieves strictly better running time and approximation than
Madry's for sparse graphs.\footnote{In the table on Page 4 of the full
  version of Madry's paper, it has been erroneously claimed that using
  the Spielman-Teng solver, Alon-Milman algorithm runs in time
  $\tilde{O}(m)$ for BS.}
More importantly, we believe that our algorithm is significantly simpler, especially in its second form mentioned above, and likely to find applications in practical settings. 

\subsection{The Matrix Exponential, the Lanczos Method and 
Approximations to $e^{-x}$}\label{sec:exp-intro}
We first state a few definitions used in this section. We will work
with $n \times n,$ symmetric and positive semi-definite (PSD) matrices
over $\mathbb{R}.$ For a matrix $M,$ abusing notation, we denote its
exponential by $\exp(-M),$ or by $e^{-M},$ and define it as
$\sum_{i\geq 0} \frac{(-1)^i}{i!}M^{i}.$ $M$ is said to be
\emph{Symmetric and Diagonally Dominant} (SDD) if, $M_{ij} = M_{ji},$
for all $i,j$ and $M_{ii} \ge \sum_j |M_{ij}|,$ for all $i$. Let
$m_M$ denote the number of non-zero entries in $M$ and let $t_M$
denote the time required to multiply the matrix $M$ with a given
vector $v.$ In general, $t_M$ depends on how $M$ is given as an input
and can be $\Theta(n^2)$. However, it is possible to exploit the
special structure of $M$ if given as an input appropriately: It is
possible to just multiply the non-zero entries of $M,$ giving $t_M =
O(m_M).$ Also, if $M$ is a rank one matrix $ww^\top,$ where $w$ is
known, we can multiply with $M$ in $O(n)$ time. We move on to our results.

At the core of our algorithm for {\BS}, and more generally of most
\mmwu based algorithms, lies an algorithm to quickly compute
$\exp(-A)v$ for a PSD matrix $A$ and a unit vector $v.$ It is
sufficient to compute an approximation $u,$ to $\exp(-A)v,$ in time
which is as close as possible to $t_{A}.$ It can be shown that using
about $\|A\|$ terms in the Taylor series expansion of $\exp(-A),$ one
can find a vector $u$ that approximates $\exp(-A)v.$ Hence, this
method runs in time roughly $O(t_{A} \cdot \|A\|).$ In our
application, and certain others \cite{AK,Kthesis,IPS11,IPS}, this
dependence on the norm is prohibitively large. The following
remarkable result was cited in Kale \cite{Kthesis}.
\vspace{2pt}
\paragraph{\bf Hypothesis.}
{\em Let $A \succeq 0$ and $\eps >0.$ There is an algorithm that requires $O\left(\log^2 \nfrac{1}{\epsilon}\right)$ iterations to find a vector $u$ such that $\norm{\exp(-A)v-u} \le \norm{\exp(-A)}\epsilon,$
for any unit vector $v$. The time for every iteration is ${O}(t_A).$ }

This hypothesis would suffice to prove Theorem \ref{thm:bsmain}. But, to
 the best of our knowledge, there is no known proof of this result. 
 In fact, the source of this unproved hypothesis can be traced to a
 paper of Eshof and Hochbruck (EH) \cite{EH}. EH suggest that one may
 use the Lanczos method (described later), and combine it with a
 rational approximation for $e^{-x}$ due to Saff, Schonhage and Varga
 \cite{SSV}, to reduce the computation of $\exp(-A)v$ to a number of
 $(I+\alpha A)^{-1}v$ computations for some $\alpha>0.$ Note that this is
 insufficient to prove the hypothesis above as there is no known way
 to compute $(I+ \alpha A)^{-1}v$ in time $O(t_{A}).$ They note this
 and propose the use of iterative methods to do this computation. They
 also point out that this will only result in an approximate solution
 to $(I+ \alpha A)^{-1}v$ and make no attempt to analyze the running time or the error of their method
 when the inverse computation is approximate. We believe that we are quite distant from proving
 the hypothesis for all PSD matrices and, moreover, that proving such a result may provide valuable insights into
 a fast (approximate) inversion method for symmetric PSD
 matrices, an extremely important open problem.
 
 A significant part of this paper is devoted to a proof of the above
 hypothesis for a class of PSD matrices that turns out to be
 sufficient for the {\BS} application. For the norm-independent,
 fast-approximate inverse computation, we appeal to the result of
 Spielman and Teng \cite{ST3} (also see improvements by Koutis,
 Miller and Peng~\cite{KMP}). The theorem we prove is the following.
\begin{theorem}[SDD Matrix Exponential Computation]
\label{thm:expRational}
Given an $n \times n$ SDD matrix $A$, a vector $v$
and a parameter $\delta \le 1$, there is an algorithm that can compute
a vector $u$ such that $\norm{\exp(-A)v-u} \le \delta\norm{v}$ in time
$\tilde{O}((m_A+n)\log (2+\norm{A})).$ Here the tilde hides
$\poly(\log n)$ and $\poly(\log \nfrac{1}{\delta})$ factors.
\end{theorem}
\noindent
First, we note that for our application, the dependence of the running
time on the $\log (2+\|A\|)$ turns out to just contribute an extra
$\log n$ factor. Also, for our application $\delta=\nfrac{1}{\poly
 (n)}.$ Secondly, for our {\BS} application, the matrix we need to
invert is not SDD or sparse. Fortunately, we can combine
Spielman-Teng solver with the Sherman-Morrison formula to invert our
matrices; see Theorem \ref{thm:expRational2}. A significant effort
goes into analyzing the effect of the error introduced due to
approximate matrix inversion. This error can cascade due to the
iterative nature of our algorithm that proves this theorem.

Towards proving the hypothesis above, when the only guarantee we know
on the matrix is that it is symmetric and PSD, we prove the following
theorem, which is the best known algorithm to compute $\exp(-A)v$ for
an arbitrary symmetric PSD matrix $A,$ when
$\norm{A}=\omega(\poly(\log n)).$
\begin{theorem}[PSD Matrix Exponential Computation]
\label{thm:expRational-psd}
Given an $n \times n$ symmetric PSD matrix $A$, a vector $v$ and a
parameter $\delta \le 1$, there is an algorithm that computes a
vector $u$ such that $\norm{\exp(-A)v-u} \le \delta\norm{v}$ in time
$\tilde{O}\left((t_A+n) \sqrt{1+ \norm{A}} \log (2+\norm{A})\right).$
Here the tilde hides $\poly(\log n)$ and $\poly(\log
\nfrac{1}{\delta})$ factors.
\end{theorem}
 
\noindent In the symmetric PSD setting we also prove the following
theorem which, for our application, gives a result comparable to
Theorem~\ref{thm:expRational-psd}.

\begin{theorem}[Simple PSD Matrix Exponential Computation]
\label{thm:expPoly}
Given an $n \times n$ symmetric PSD matrix $A$, a vector $v$ and a
parameter $\delta \le 1$, there is an algorithm that computes a vector
$u$ such that $\norm{\exp(-A)v-u} \le \delta\norm{v},$ in time
$O((t_A+n)\cdot k+k^2),$ where $k \defeq
\tilde{O}(\sqrt{1+\norm{A}}).$ Here the tilde hides $\poly(\log
\nfrac{1}{\delta})$ factors.
\end{theorem}
As noted above, $t_A$ can be significantly smaller than $m_A.$ Moreover, it only uses
multiplication of a vector with the matrix $A$ as a primitive and does
not require matrix inversion. Consequently, it does not need 
tools like the SDD solver or conjugate gradient, thus obviating the
error analysis required for the previous algorithms. Furthermore,
this algorithm is very simple and when combined with our random
walk-based {\sc BalSep} algorithm, results in a very simple and
practical $O(\sqrt{\gamma})$ approximation algorithm for BS that runs in time $\tilde{O}(\nfrac{m}{\sqrt{\gamma}}).$

Theorem \ref{thm:expPoly} relies on the Lanczos method which can be
used to convert guarantees about polynomial approximation from scalars
to matrices. In particular, it uses the following structural result
(the upper bound) on the best degree $k$ polynomial $\delta$-uniformly
approximating $e^{-x}$ in an interval $[a,b].$ We also prove a lower
bound which establishes that the degree cannot be improved beyond
lower order terms. This suggests that improving on the 
$\tilde{O}(\nfrac{m}{\sqrt{\gamma}})$ running time in Theorem
\ref{thm:expPoly} requires more advanced techniques. 
To the best of our knowledge, this theorem is new
and is of independent interest in approximation theory.
\begin{theorem}[Uniform Approximation to $e^{-x}$]
\label{thm:exp-poly-approx} $\\$
\begin{itemize}
\item{\bf Upper Bound.} For every $0 \le a < b,$ and $0<\delta \le 1$,
there exists a polynomial ${p}$ that satisfies,
$\sup_{x \in [a,b]} |e^{-x}-{p}(x)| \le
\delta\cdot e^{-a},$ and has degree $O\left(\sqrt{\max\{\log^2
\nfrac{1}{\delta},(b-a) \cdot \log \nfrac{1}{\delta} \}}\cdot
\left(\log \nicefrac{1}{\delta}\right)^2 \right)$. 

\item{\bf Lower Bound.} For every $0 \leq a < b $ such that $a + \log_e 4 \le b,$ and $\delta
\in (0,\nfrac{1}{8}],$ any polynomial $p(x)$ that approximates
$e^{-x}$ uniformly over the interval $[a,b]$ up to an error of
$\delta\cdot e^{-a},$ must have degree at least
$\frac{1}{2}\cdot\sqrt{b-a}\ .$
\end{itemize}
\end{theorem}

\section{Organization of the Main Body of the Paper}
\noindent 
In Section \ref{sec:overview} we present a technical overview of our results and in Section \ref{sec:open} we discuss the  open problems arising from our work. 
The main body of the paper follows after it and is divided into three sections, each of which
have been written so that they can be read independently. Section
\ref{sec:balsep-main} contains a complete description and all the
proofs related to Theorem \ref{thm:bsmain} and Theorem
\ref{thm:bsmain2}. Section \ref{sec:exp} contains our results on
computing the matrix exponential; in particular the proofs of Theorems
\ref{thm:expRational}, \ref{thm:expPoly} and
\ref{thm:expRational2}. Section \ref{sec:poly} contains the proof of
our structural results on approximating $e^{-x}$ and the proof of
Theorem \ref{thm:exp-poly-approx}.

\section{Technical Overview of Our Results}
\label{sec:overview}

\subsection{Our Spectral Algorithm for Balanced Separator}
In this section, we provide an overview of Theorem \ref{thm:bsmain}. As pointed out in the introduction, our algorithm, {\alg}, when combined with the matrix-exponential-vector algorithm in Theorem \ref{thm:expPoly} results in a very simple and practical algorithm for {\BS}. 
We record the theorem here for completeness and then move on to the overview of {\alg} and its proof. The proof of this theorem appears in Section \ref{sec:gpanalysis}.
\begin{theorem}[Simple Spectral Algorithm for Balanced Separator]\label{thm:bsmain2}
Given an unweighted graph $G=(V,E)$, a balance parameter $b \in (0,\nfrac{1}{2}], \; b = \Omega(1)$ and a conductance value $\gamma \in (0,1),$ we give an algorithm, which runs in time $\tilde{O}(\nfrac{m}{\sqrt{\gamma}}),$ 
 that either outputs an $\Omega(b)$-balanced cut $S \subset V$ such that $\phi(S) \leq O(\sqrt{\gamma})$ or outputs a certificate that no $b$-balanced cut of conductance $\gamma$ exists. 
 \end{theorem}

\subsubsection{Comparison with the {\sc RLE} Algorithm}

Before we explain our algorithm, it is useful to review the {\sc RLE} algorithm. Recall that given $G,\gamma$ and $b,$ the goal of the {\BS} problem is to either certify that every $b$-balanced cut in $G$ has conductance at least $\gamma,$ or produce a $\Omega(b)$ balance cut in $G$ of conductance $O(\sqrt{\gamma}).$ 
{\sc RLE} does this by applying {\sc LE} iteratively to remove unbalanced cuts of conductance $O(\sqrt{\gamma})$ from $G.$
The iterations stop and the algorithm outputs a cut, when it either finds a
$(b/2)$-balanced cut of conductance $O(\sqrt{\gamma})$ or the union of
all unbalanced cuts found so far is $(b/2)$-balanced. Otherwise, the
algorithm terminates when the residual graph has spectral gap at least
$2 \gamma.$ In the latter case, any $b$-balanced cut must have at
least half of its volume lie within the final residual graph, and hence,
has conductance at least $\gamma$ in the original graph.
Unfortunately, this algorithm may require $\Omega(n)$ iterations in the worst case. For instance, this is true if the graph $G$ consists of $\Omega(n)$ components loosely connected to an expander-like core through cuts of low conductance. This example highlights the weakness of the {\sc RLE} approach: the second eigenvector of the Laplacian may only be correlated with one low-conductance cut and fail to capture at all even cuts of slightly larger conductance. This limitation makes it impossible for {\sc RLE} to make significant progress at any iteration. We now proceed to show how to fix RLE and present our algorithm at a high level.
\subsubsection{High-Level Idea of Our Algorithm}
Rather than working with the vertex embedding given by the
eigenvector, at iteration $t,$ we will consider the multi-dimensional
vector embedding represented by the transition probability matrix
$P^{(t)}$ of a certain random walk over the graph. We refer to this kind of walk as
an {\em Accelerated Heat Kernel Walk} ({\sc AHK}) and we describe it
formally in Section \ref{sec:overview-alg}.
 At each iteration $t=1,2,\ldots,$ the current {\sc AHK} walk is
 simulated for $\tau = \nfrac{\log n}{\gamma}$ time to obtain
 $P^{(t)}.$ For any $t,$ this choice of $\tau$ ensures that the walk must mix
 across all cuts of conductance much larger than $\gamma,$ hence
 emphasizing cuts of the desired conductance in the embedding
 $P^{(t)}.$
The embedding obtained in this way, can be seen as a weighted combination of multiple eigenvectors, with eigenvectors of low eigenvalue contributing more weight. 
Hence, the resulting embedding captures not only the cut corresponding to the second eigenvector, 
but also cuts associated with other eigenvectors of eigenvalue close to $\gamma.$
This enables our algorithm to potentially find many different low-conductance unbalanced cuts at once.
Moreover, the random walk matrix is more stable than the eigenvector under small perturbations of the graph, making it possible to precisely quantify our 
progress from one iteration to the next as a function of the mixing of the current random walk.
For technical reasons, we are unable to show that we make sufficient
progress if we just remove the unbalanced cuts found, as in {\sc
 RLE}. Instead, if we find a low-conductance unbalanced cut $S^{(t)}$
at iteration $t,$ we perform a {\em soft} removal, by modifying the
current walk $P^{(t)}$ to accelerate the convergence to stationarity on the set
$S^{(t)}.$ This ensures that a different cut is found using
$P^{(t+1)}$ in the next iteration. 
In particular, the {\sc AHK} walks we consider throughout the execution of the algorithm will behave like the standard heat kernel on most of the graph, except on a small unbalanced subset of vertices, where their convergence will be accelerated. We now present our algorithm in more detail. We first recall some definitions.

\subsubsection{Definitions}
 $G=(V,E)$ is the unweighted
instance graph, where $V = [n]$ and $\card{E}=m.$ We let $d \in \R^n$
be the degree vector of $G,$ \emph{i.e.}, $d_i$ is the degree of vertex $i.$
For a subset $ S \subseteq V,$ we define the edge volume as $\vol(S)
\defeq \sum_{i \in S} d_i.$ The total volume of $G$ is $2m.$ We denote
by $K_V$ the complete graph with weight $ \nfrac{d_id_j}{2m}$ between
every pair $i,j \in V.$ For $i \in V,$ $S_i$ is the star graph rooted
at $i$, with edge weight of $ \nfrac{d_id_j}{2m}$ between $i$
and $j,$ for all $j \in V.$ For an undirected graph $H=(V, E_H)$, let
$A(H)$ denote the adjacency matrix of $H$, and $D(H)$ the diagonal
matrix of degrees of $H$. The (combinatorial) Laplacian of $H$ is
defined as $L(H) \defeq D(H) - A(H)$.
By $D$ and $L$, we denote $ D(G)$ and $L(G)$ respectively for the input graph $G.$ 
For two matrices $A,B$ of equal dimensions, let $A \bullet B \defeq \Tr(A^\top B) = \sum_{ij} A_{ij}\cdot B_{ij}.$ ${\bf 0}$ denotes the all $0$s vector.

\subsubsection{The AHK Random Walk and its Mixing}
\label{sec:overview-alg}
We will be interested in continuous-time random walk processes over
$V$ that take into account the edge structure of $G$. The simplest
such process is the {\it heat kernel} process, which has already found
many applications in graph partitioning, particularly in the work of
Chung \cite{ChungExp}, and in many of the combinatorial algorithms for
solving graph partitioning {\SDP}s \cite{Lthesis}. The heat kernel is
defined as the continuous-time process having transition rate matrix
$-LD^{-1}$ and, hence, at time $\tau$ the probability-transition
matrix becomes $\exp({-\tau LD^{-1}}).$
The AHK random walk process has an additional parameter $\beta,$ which is a non-negative vector in $\R^n.$ The
transition rate matrix is then defined to be $- (L + \sum_{i \in V}
\beta_i L(S_i))D^{-1}.$ The effect of adding the star terms to the transition rate
matrix is that of accelerating the convergence of the process to
stationarity at vertices $i$ with large value of $\beta_i$, since a large
fraction of the probability mass that leaves these vertices is
distributed uniformly over the edges.
We denote by $P_{\tau}(\beta)$ the probability-transition $
\exp({-\tau(L + \sum_{i \in V} \beta_i L(S_i)) D^{-{1} }}).$ As is
useful in the study of spectral properties of non-regular graphs, we
will study $D^{-1} P_{\tau} (\beta).$ This matrix describes the
probability distribution over the edges of $G$ and has the advantage
of being symmetric and PSD: $D^{-1} P_{\tau}(\beta) =
D^{-\nfrac{1}{2}} \exp({-\tau D^{-\nfrac{1}{2}} (L + \sum_{i \in V}
 \beta_i L(S_i)) D^{-\nfrac{1}{2} }}) D^{-\nfrac{1}{2}}.$ In particular,
$D^{-1}P_{\tau}(\beta)$ can be seen as the Gram
matrix of the embedding given by the columns of its square root, which
in this case is just $D^{-\nfrac{1}{2}}P_{\nfrac{\tau}{2}}(\beta).$
This property will enable us to use geometric \SDP-rounding techniques
to analyze AHK walks.
Throughout the algorithm, we keep track of the mixing of $P^{(t)}$ by considering the {\it total deviation} of $P^{(t)}$ from stationarity, \emph{i.e.}, the sum over all vertices $i \in V$ of the $\ell_2^2$-distance from the stationary distribution of $P^{(t)} e_i.$
Here, $e_i$ denotes the vector in $\mathbb{R}^n$ which is $1$ at the
$i^\text{th}$ coordinate and $0$ elsewhere. We denote the contribution of
vertex $i$ to this distance by $\Psi(P_{\tau}(\beta), i).$
Similarly, the total deviation from stationarity over a subset $S
\subseteq V,$ is given by, $\Psi(P_{\tau}(\beta), S) \defeq \sum_{i \in S} \Psi(P_{\tau}(\beta), i).$
$\Psi(P_{\tau}(\beta), V )$ will play the role of potential function
in our algorithm. Moreover, it follows from the definition of
$\Psi$ that $\Psi(P_{\tau}(\beta), V )= L(K_V) \bullet
D^{-1}P_{\tau}(\beta).$ Finally, to connect to the high-level idea described earlier, for each iteration $t,$
we use $P^{(t)} \defeq P_\tau(\beta^{(t)})$ for $\tau = \nfrac{\log
 n}{\gamma}$ with $\beta^{(t)}\approx \nfrac{1}{\tau} \sum_{j=1}^{t-1} \sum_{i \in S^{(j)}} e_i$ and starting with $\beta^{(1)} = {\bf 0}.$
We now provide a slightly more detailed description of our algorithm from Theorem \ref{thm:bsmain} (called {\alg}) and its analysis. For reference, {\alg} appears in Figure \ref{fig:algorithm}.

\subsubsection{Our Algorithm and its Analysis}
The algorithm proceeds as follows: At iteration $t,$ it checks if the
total deviation of $P^{(t)},$ \emph{i.e.}, $\Psi(P^{(t)},V),$ is sufficiently
small (\emph{i.e.}, $P^{(t)}$ is mixing). In this case, we can guarantee that
no balanced cut of conductance less than $\gamma$ exists in $G.$ 
In more formal language, it appears below.

\vspace{2mm}
{\em {\bf (A)} (see Lemma \ref{lem:pot})
Let $S = \cup_{i=1}^t S^{(i)}.$
For any $t \geq 1,$ if $\; \Psi(P^{(t)}, V) \leq \nfrac{1}{\poly(n)},$ and $\vol(S) \le \nfrac{b}{100} \cdot 2m$, then
$
L + \nfrac{1}{\tau}\sum_{i \in S}  L(S_i) \succeq \Omega(\gamma) \cdot L(K_V).
$
Moreover, no $b$-balanced cut of conductance less than $\gamma$
exists in $G.$}
\vspace{2mm}

\noindent
This result has a simple explanation in terms of the {\sc AHK} random walk $P^{(t)}.$ Notice that $P^{(t)}$ is accelerated only on a small unbalanced set $S.$ Hence, if a balanced cut of conductance less than $\gamma$ existed, its convergence could not be greatly helped by the acceleration over $S.$ Thus, if $P^{(t)}$ is still mixing very well, no such balanced cut can exist. 
On the other hand, if $P^{(t)}$ has high total deviation (\emph{i.e.}, the
walk has not yet mixed), then, intuitively, some cut of low conductance exists in
$G.$ Formally, we show that, the embedding $\{v^{(t)}_i\}_{i \in V}$ has low quadratic form with respect to the Laplacian of $G.$ 

\vspace{2mm}
{\em {\bf (B)} (see Lemma \ref{lem:obj}) If $\Psi(P^{(t)}, V) \geq
 \frac{1}{\poly(n)},$ then $L \bullet D^{-1}P^{(t)} \leq O(\gamma)
 L(K_V) \bullet D^{-1}P^{(t)}.$ }
\vspace{2mm}

\noindent
From an \SDP-rounding perspective, this means that the embedding $P^{(t)}$ can
be used to recover a cut $S^{(t)}$ of conductance $O(\sqrt{\gamma}),$
using the \SDP-rounding techniques from OV. If $S^{(t)}$ or $\cup_{i=1}^t S^{(i)}$ is 
$\Omega(b)$-balanced, then we output that cut and terminate. Otherwise,
$S^{(t)}$ is unbalanced. In this case, we accelerate the convergence
from $S^{(t)}$ in the current {\sc AHK} walk by increasing
$(\beta^{(t)})_i$ for every $i \in S^{(t)}$ to give $\beta^{(t+1)},$
and using $\beta^{(t+1)}$ to produce $P^{(t+1)}$ and move on to the
next iteration.

The analysis of our algorithm bounds the number of iterations by using
the total deviation of $P^{(t)}$ from stationarity as a potential
function.
Using the techniques of OV, it is possible to show that, whenever an
unbalanced cut $S^{(t)}$ is found, most of the deviation of $P^{(t)}$
can be attributed to $S^{(t)}.$ In words, we can think of $ S^{(t)}$
as the main reason why $P^{(t)}$ is not mixing. Formally,

\vspace{1mm}
{\em {\bf (C)} (see Corollary \ref{cor:unbalanced}) At iteration $t$, if
 $\Psi(P^{(t)}, V) \geq \nfrac{1}{\poly(n)}$ and $S^{(t)}$ is not
 $\nfrac{b}{100}$-balanced, then w.h.p. $\Psi(P^{(t)}, S^{(t)}) \geq
 \nfrac{1}{2} \cdot \Psi(P^{(t)}, V).$}
\vspace{1mm}

\noindent
Moreover, we can show that accelerating the convergence of the walk from $S^{(t)}$ has the effect of removing from $P^{(t+1)}$ 
a large fraction of the deviation due to $S^{(t)}.$
The proof is a simple application of the Golden-Thompson inequality~\cite{Bhatia} and mirrors the main step in the \mmwu analysis.
Hence, we can show the total deviation of $P^{(t+1)}$ is just a constant fraction of that of $P^{(t)}.$

\vspace{1mm}
{\em {\bf (D)} (see Theorem \ref{thm:reduction})
If $\; \Psi(P^{(t)}, V) \geq \nfrac{1}{\poly(n)}$ and $S^{(t)}$ in not $\nfrac{b}{100}$ balanced, then w.h.p
\[
\Psi(P^{(t+1)}, V) \leq \Psi(P^{(t)}, V) - \nfrac{1}{3} \cdot \Psi(P^{(t)}, S^{(t)}) \leq \nfrac{5}{6} \cdot \Psi(P^{(t)}, V).
\]}
This potential reduction at every iteration allows us to argue that after $T=O(\log n)$ iterations, $P^{(T+1)}$ must have a small deviation from stationarity and yields a certificate that no balanced cut of conductance less than $\gamma$ exists in $G.$
Finally, to ensure that each iteration requires only $\tilde{O}(m)$ time,
we use the Johnson-Lindenstrauss Lemma to compute a $O(\log n)$-dimensional approximation to the embedding $P^{(t)}.$ 
To compute this approximation, we rely on the results on approximating
the matrix exponential discussed in Section \ref{sec:overview-exp}.

\subsubsection{Exponential Embeddings of Graph and Proof Ideas}
Now, we illustrate the usefulness of the exponential embedding
obtained from the AHK random walk, which is key to the proofs for (A)
and (B) above. We suppress some details pertinent to our
algorithm. Consider the AHK walk in the first iteration, with
$\beta^{(1)} = 0.$ Letting $C \defeq D^{-\nfrac{1}{2}} L
D^{-\nfrac{1}{2} },$ $P \defeq P^{(1)}=\exp(-\tau C).$ For the proof
of (A), it follows that if $\Psi(P,V) \leq \nfrac{1}{\poly (n)},$ then
by definition $ L(K_V) \bullet D^{-1}P= \Tr(\exp(-\tau C))-1 \leq
\nfrac{1}{\poly (n)}.$ Hence, $\lambda_2(C) \gtrapprox \nfrac{\log
 n}{\tau} = \gamma,$ by the choice of $\tau.$ This lower bound on the
second eigenvalue certifies that $G$ has no cut of conductance at most
$\gamma.$ For our algorithm, at iteration $t$, when $\beta^{(t)} \neq
0,$ we will ensure that the Laplacians of the stars have small enough
weight ($\approx \nfrac{1}{\tau}$) and small support
($\vol(\cup_{i=1}^t S^{(i)}) \leq \nfrac{b}{100}\cdot 2m$) for the
argument above to still yield a lower bound of $\gamma$ on the
conductance of any $b$-balanced cut.

For (B), observe that $L \bullet D^{-1}P = C \bullet \exp(-\tau C) =
\sum_i \lambda_i e^{-\tau \lambda_i},$ where $\lambda_i,$ for
$i=1,\ldots,n,$ are the eigenvalues of $C.$
For eigenvalues larger than
$2\gamma,$ $e^{-\tau \lambda_i}$ is bounded by $\nfrac{1}{\poly(n)}$
for $\tau = \nfrac{\log n}{\gamma}$. Since eigenvalues of a normalized
Laplacian are $O(1),$ the contribution to the sum above by eigenvalues
$\lambda_i > 2 \gamma$ is at most $\nfrac{1}{\poly(n)}$ overall.
This can be shown to be a small fraction of the total sum.
Hence, the quantity $L \bullet D^{-1}P$ is mostly determined by the
eigenvalues of value less than $2\gamma$ and we have $L \bullet
D^{-1}P \leq O(\gamma) \cdot L(K_V) \bullet D^{-1} P.$
The same analysis goes through when $t > 1.$

\subsection{Our Algorithms to Compute an Approximation to $\exp(-A)v$}
\label{sec:overview-exp}
In this section, we give an overview of the algorithms in Theorem
\ref{thm:expRational} and Theorem \ref{thm:expPoly} and their
proofs. The algorithm for Theorem~\ref{thm:expRational-psd} is very
similar to the one for Theorem~\ref{thm:expRational} and we give
the details in Section~\ref{sec:expRational-psd:proof}. A few
quick definitions: A matrix $M$ is called \emph{Upper Hessenberg} if, $(M)_{ij}
= 0$ for $i > j+1.$ $M$ is called \emph{tridiagonal} if $M_{ij} = 0$
for $i > j+1$ and for $j > i+1.$ Let $\lambda_1(M)$ and $\lambda_n(M)$
denote the largest and smallest eigenvalues of $M$ respectively.

As we mention in the introduction, the matrices that we need to
exponentiate for the {\BS} algorithm are no longer sparse or SDD.
Thus, Theorem \ref{thm:expRational} is insufficient for our
application. Fortunately, the following theorem suffices and its proof
is not very different from that of Theorem \ref{thm:expRational}, which
is explained below. Its proof appears in Section~\ref{sec:expRational2:proof}.
\begin{theorem}[Matrix Exponential Computation Beyond SDD]
\label{thm:expRational2}
Given a vector $v,$ a parameter $\delta \le 1$ and an $n \times n$
symmetric matrix $A=\Pi H M H \Pi$ where $M$ is SDD, $H$ is a diagonal
matrix with strictly positive entries and $\Pi$ is a rank $(n-1)$
projection matrix, $\Pi \defeq I-ww^\top$ ($w$ is explicitly known and
$\norm{w}=1$), there is an algorithm that computes a vector $u$ such
that $\norm{\exp(-A)v-u} \le \delta\norm{v}$ in time
$\tilde{O}((m_M+n)\log (2+\norm{H M H})).$ The tilde hides $\poly(\log
n)$ and $\poly(\log \nfrac{1}{\delta})$ factors.
\end{theorem}

\noindent
Recall from Section \ref{sec:overview-alg} that our algorithm for
{\BS} requires us to compute $\exp(-A)v$ for a matrix $A$ of the form
$D^{-\nicefrac{1}{2}}(L + \sum_i \beta_i L(S_i))D^{-\nicefrac{1}{2}},$
where $\beta_i \ge 0.$ We first note that if we let $\Pi \defeq I -
\nfrac{1}{2m}\cdot (D^{\nicefrac{1}{2}}{ 1})(D^{\nicefrac{1}{2}}{
 1})^\top,$ the projection onto the space orthogonal to
$\nfrac{1}{\sqrt{2m}}\cdot D^{\nicefrac{1}{2}}{1},$ then, for each
$i,$ $D^{-\nicefrac{1}{2}} L(S_i) D^{-\nicefrac{1}{2}} = \Pi
(\nicefrac{d_i}{2m}\cdot I + e_i e_i^\top) \Pi.$ Since
$D^{\nicefrac{1}{2}}{1}$ is an eigenvector of
$D^{-\nicefrac{1}{2}}LD^{-\nicefrac{1}{2}},$ we have, $\Pi
D^{-\nicefrac{1}{2}}LD^{-\nicefrac{1}{2}} \Pi =
D^{-\nicefrac{1}{2}}LD^{-\nicefrac{1}{2}}.$ Thus, \[A = \Pi
D^{-\nicefrac{1}{2}}L D^{-\nicefrac{1}{2}} \Pi + \sum_i \beta_i \Pi
(\nicefrac{d_i}{2m}\cdot I + e_i e_i^\top)\Pi =\Pi
D^{-\nicefrac{1}{2}}(L + \sum_i \beta_i \nicefrac{d_i}{2m}\cdot D +
\sum_i \beta_i d_i\cdot e_i e_i^\top) D^{-\nicefrac{1}{2}}\Pi.\] This
is of the form $\Pi H MH \Pi,$ where $H \defeq D^{-\nicefrac{1}{2}}$
is diagonal and $M$ is SDD.
It is worth noting that since $A$ itself may be neither sparse nor
SDD, we cannot apply the Spielman-Teng SDD solver to approximate
$(I+\alpha A)^{-1}.$ The proof of the above theorem uses the Sherman-Morrison formula to extend the SDD solver
to fit our requirement. Moreover, to obtain a version of Theorem \ref{thm:expPoly} for such matrices, we do not have to do anything additional since multiplication by $H$ and $\Pi$
take $O(n)$ steps and hence, $t_A$ is still $O(m_M + n).$ The details appear in Section \ref{sec:expRational2:proof}. Finally, note that in our application, $\|HMH\|$ is ${\rm poly}(n).$ 

We now give an overview of the proofs of Theorem \ref{thm:expRational}
and Theorem \ref{thm:expPoly}. First, we explain a general method known as the Lanczos method, which is pervasive in numerical linear algebra. We then show how suitable adaptations of this can be combined with (old and new) structural results in approximation theory to obtain our results.

\subsubsection{Lanczos Method} 
Given an $n \times n$ symmetric PSD matrix $B$ and a function
$f:\mathbb{R} \mapsto \mathbb{R},$ we can define $f(B)$ as follows:
Let $u_1,\ldots,u_n$ be eigenvectors of $B$ with eigenvalues
$\lambda_1,\ldots,\lambda_n.$ Define $f(B) \defeq \sum_{i}
f(\lambda_i)u_iu_i^\top.$ We will reduce both our algorithms to
computing $f(B)v$ for a given vector $v,$ albeit with different $f$'s
and $B$'s. We point out the $f$'s and $B$'s required for Theorems
\ref{thm:expRational} and \ref{thm:expPoly} in Sections
\ref{sec:approx-rat} and \ref{sec:approx-poly} respectively.

Since exact computation of $f(B)$ usually requires diagonalization of
$B,$ which could take as much as $O(n^3)$ time (see~\cite{eigendecomp}),
we seek an approximation to $f(B)v$. The Lanczos method allows us to
do exactly that: It looks for an approximation to $f(B)v$ of the form
$p(B)v$, where $p$ is a polynomial of small degree, say $k$. Before we
describe how, we note that it computes this approximation in roughly
$O((t_B + n)k)$ time plus the time it takes to compute $f(\cdot)$ on a
$(k+1) \times (k+1) $ tridiagonal matrix, which can often be upper bounded by
$O(k^2)$ (see~\cite{eigendecomp}). Hence, the time is reduced to
$O((t_B+n)k+k^2).$ What one has lost in this process is accuracy: The
candidate vector $u$ output by the Lanczos method, is now only an
approximation to $f(B)v.$ The quality of approximation, or $\|f(B)v
-u\|,$ can be upper bounded by the {\em uniform error} of the best
degree $k$ polynomial approximating $f$ in the interval
$[\lambda_n(B),\lambda_1(B)].$ Roughly, $ \|f(B)v-u\| \approx (\min_{p_k
 \in \Sigma_k} \sup_{x \in [\lambda_1(B),\lambda_n(B)]} |f(x) -
p_k(x)|).$ Here $\Sigma_k$ is the collection of all real polynomials
of degree at most $k.$ Surprisingly, one does not need to know the
best polynomial and proving {\em existence} of good polynomials is
sufficient. By increasing $k,$ one can reduce this error and, indeed,
if one lets $k=n,$ there is no error. Thus, the task is reduced to
proving {existence} of low degree polynomials that approximate $f$
within the error tolerable for the applications.

\subsubsection*{\it Computing the Best Polynomial Approximation.} 
Now, we describe in detail, the Lanczos method and how it achieves the
error guarantee claimed above. Notice that for any
polynomial $p$ of degree at most $k,$ the vector $p(B)v$ lies in
$\calK \defeq \Span\{v,Bv,\ldots,B^kv\}$ -- called the \emph{Krylov
 subspace}.
The
Lanczos method iteratively creates an orthonormal basis
$\{v_i\}_{i=0}^{k}$ for $\calK$, such that $\forall\ i \le k,\
\Span\{v_0,\ldots,v_i\} = \Span\{v,\ldots,B^iv\}.$ Let $V_k$ be the $n
\times (k+1)$ matrix with $\{v_i\}_{i=0}^k$ as its columns. Thus,
$V_kV_k^\top $ denotes the projection onto
the Krylov subspace. We let $T_k$ be the $(k+1)\times (k+1)$ matrix
expressing $B$ as an operator restricted to $\calK$ in the basis
$\{v_i\}_{i=0}^k$, \emph{i.e.}, $T_k \defeq V_k^\top B V_k.$ Note that this is
not just a change of basis, since vectors in $\calK$ can be mapped by
$B$ to vectors outside $\calK$. Now, since $v,Bv \in \calK$, we must
have $Bv = (V_k V^\top_k) B (V_k V^\top_k) v = V_k (V^\top_k B V_k
)V_k^\top v = V_kT_kV_k^\top v.$ Iterating this argument, we get that
for all $i \le k$, $B^iv = V_kT_k^iV_k^\top v,$ and hence, by
linearity, $p(B)v = V_kp(T_k)V_k^\top v,$ for any polynomial $p$ of degree
at most $k.$

Now, a natural approximation for $f(B)v$ is $V_kf(T_k)V_k^\top v$. 
Writing $r_k(x) \defeq f(x)-p_k(x),$ where $p_k$ is any degree $k$
approximation to $f(x)$, the error in the approximation is
$f(B)v-V_{k}f(T_{k})V_{k}^{\top}v =r_k(B)v-V_kr_k(T_k)V_k^\top v,$
{\em for any choice of $p_{k}.$} Hence, the norm of the
error vector is at most $(\norm{r_k(B)} + \norm{r_k(T_k)})\norm{v},$
which is bounded by the value of $r_k$ on the eigenvalues of $B$
(eigenvalues of $T_k$ are a subset of eigenvalues of $B$). More
precisely, the norm of the error is bounded by $2\norm{v} \cdot
\max_{\lambda \in \text{Spectrum}(B)}| f(\lambda)-p_{k}(\lambda)|.$
Minimizing over $p_k$ gives the error bound claimed above.
Note that we do not explicitly need the approximating polynomial. It
suffices to prove that there exists a degree $k$ polynomial that
uniformly approximates $f$ well on an interval containing the spectrum
of $B$ and $T_k.$

If we construct the basis iteratively as above, $Bv_j \in
\Span\{v_0,\ldots,v_{j+1}\}$ by construction, and if $i > j+1,$ $v_i$
is orthogonal to this subspace and hence $v_i^\top (Bv_j)=0$. Thus,
$T_k$ is Upper Hessenberg. Moreover, if $B$ is symmetric, $v_j^\top
(Bv_i) = v_i^\top (Bv_j),$ and hence $T_k$ is symmetric and
tridiagonal. This means that while constructing the basis, at step
$i+1$, it needs to orthonormalize $Bv_i$ only w.r.t. $v_{i-1}$ and
$v_i$. Thus the total time required is $O((t_B + n)k)$, plus the
time required for the computation of $f(T_k)$, which can typically be
bounded by $O(k^2)$ for a tridiagonal matrix (using \cite{eigendecomp}).
This completes an overview of the Lanczos method. 
 The
{\lanczos} procedure described in Figure~\ref{fig:lanczos} in the main
body, implements the Lanczos method.
We now move on to describing how we apply it to obtain our two algorithms.

\subsubsection{Approximating $\exp(-A)v$ Using a Rational Approximation to $e^{-x}$}\label{sec:approx-rat}
\subsubsection*{\it Our Algorithm.} The starting point of the
algorithm that underlies Theorem \ref{thm:expRational} is a rather
surprising result by Saff, Sch\"{o}nhage and Varga (SSV)~\cite{SSV},
which says that for any integer $k,$ there exists a degree $k$
polynomial $p_k^\star$ such that, $p_k^\star((1+\nfrac{x}{k})^{-1})$
approximates $e^{-x}$ up to an error of $O(k\cdot 2^{-k})$
over the interval $[0,\infty)$
(Theorem~\ref{thm:rational-approximations},
Corollary~\ref{cor:pk-star}). Then, to compute $\exp(-A)v,$ one could
apply the Lanczos method with $B \defeq (I + \nfrac{A}{k})^{-1}$ and
$f(x) \defeq e^{k(1-\nfrac{1}{x})}.$ Essentially, this was the method
suggested by Eshof and Hochbruck \cite{EH}. The strong approximation
guarantee of the SSV result along with the guarantee of the Lanczos
method from the previous section, would imply that the order of the
Krylov subspace for $B$ required would be roughly $\log
\nfrac{1}{\delta},$ and hence, independent of $\|A\|.$ The running
time is then dominated by the computation $Bv=(I +
\nfrac{A}{k})^{-1}v.$

EH note that the computation of exact matrix inverse is a costly
operation ($O(n^3)$ time in general) and all known faster methods for
inverse computation incur some error. They suggest using the Lanczos
method with faster iterative methods, e.g. Conjugate Gradient, for
computing the inverse (or rather the product of the inverse with a
given vector) as a {\em heuristic}. They make no attempt to give a
theoretical justification of why approximate computation
suffices. Also note that, even if the computation was error-free, a
method such as Conjugate Gradient will have running time which varies
with $\sqrt{\nfrac{\lambda_1(A)}{\lambda_n(A)}}$ in general. Thus, the
EH method falls substantially short of resolving the hypothesis
mentioned in the introduction.

To be able to prove Theorem \ref{thm:expRational} using the SSV
guarantee, we have to adapt the Lanczos method in several ways, and
hence, deviate from the method suggested by EH: 1) EH construct $T_k$
as a tridiagonal matrix as Lanczos method suggests, but since the
computation is no longer exact, the basis $\{v_i\}_{i=0}^k$ is no
longer guaranteed to be orthonormal. As a result, the proofs of the
Lanczos method break down. Our algorithm, instead, builds an
orthonormal basis, which means that $T_k$ becomes an Upper Hessenberg
matrix instead of tridiagonal and we need to compute $k^2$ dot
products in order to compute $T_k.$ 2) With $T_k$ being asymmetric,
several nice spectral properties are lost, \emph{e.g.} real
eigenvalues and an orthogonal set of eigenvectors. We overcome this
fact by symmetrizing $T_k$ to construct $\widehat{T}_k = \frac{T_k +
 T_k^\top}{2}$ and computing our approximation with $\widehat{T}_k.$
This permits us to bound the quality of a polynomial approximation
applied to $\widehat{T}_k$ by the behavior of the polynomial on the eigenvalues
of $\widehat{T}_k$. 3) Our analysis is based on the SSV approximation
result, which is better than the variant proved and
used by EH. Moreover, for their \emph{shifting} technique, which is
the source of the $\norm{\exp(-A)}$ factor in the hypothesis, the
given proof in EH is incorrect and it is not clear if the given bound
could be achieved even under exact computation\footnote{EH show the
 existence of degree $k$ polynomials in $(1+\nu x)^{-1}$ for any
 \emph{constant} $\nu \in (0,1),$ that approximate $e^{-x}$ up to an
 error of $\exp(\nfrac{1}{2\nu}-\Theta(\sqrt{k(\nu^{-1}-1)})).$ In
 order to deduce the claimed hypothesis, it needs to be used for $\nu
 \approx \nfrac{1}{\lambda_{n}(A)},$ in which case, there is a factor of
 $e^{\lambda_n(A)}$ in the error, which could be huge.}. 4) Most
importantly, since $A$ is SDD, we are able to employ the 
Spielman-Teng solver (Theorem \ref{thm:Spielman-Teng}) to approximate
$(I+\nfrac{A}{k})^{-1}v$. This procedure, called {\expRational}, has
been described in Figure~\ref{fig:expRational} in the main body.

\subsubsection*{\it Error Analysis.}
To complete the proof of Theorem \ref{thm:expRational}, we need to
analyze the role of the error that creeps in due to approximate matrix
inversion. The problem is that this error, generated in each iteration
of the Krylov basis computation, propagates to the later steps. Thus,
small errors in the inverse computation may lead to the basis $V_k$
computed by our algorithm to be quite far from the $k$-th order Krylov
basis for $B,v.$
We first show that, assuming the error in computing the inverse is
small, $\widehat{T}_k$ can be used to approximate degree $k$
polynomials of $B= (I+\nfrac{A}{k})^{-1}$ when restricted to the
Krylov subspace, {\em i.e.} $\|p(B)v-V_kp(\widehat{T}_k)V_k^\top v\|
\lessapprox \norm{p}_1.$ Here, if $p \defeq \sum_{i=0}^k a_i \cdot
x^i,$ $\norm{p}_1 = \sum_{i \ge 0}^k |a_i|.$
This is the most technical part of the error analysis and
unfortunately, the only way we know of proving the error bound above
is by {\em tour de force}. A part of this proof is to show that the
spectrum of $\widehat{T}_k$ cannot shift far from the spectrum of $B.$

To bound the error in the candidate vector output by the algorithm,
{\em i.e.} $\|f(B)v-V_kf(\widehat{T}_k)V_k^\top v\|,$ we start by
expressing $e^{-x}$ as the sum of a degree $k$-polynomial $p_k$ in
$(1+\nfrac{x}{k})^{-1}$ and a remainder function $r_k.$ We use the
analysis from the previous paragraph to upper bound the error in the
polynomial part by $\approx \norm{p}_1.$ We bound the contribution of
the remainder term to the error by bounding $\norm{r_k(B)}$ and
$\|{r_k(\widehat{T}_k)}\|.$ This step uses the fact that eigenvalues
of $r_k(\widehat{T}_k)$ are $\{r_k(\lambda_i)\}_i,$ where
$\{\lambda_i\}_i$ are eigenvalues of $\widehat{T_k}.$ This is the
reason our algorithm symmetrizes $T_k$ to $\widehat{T}_k.$ To complete
the error analysis, we use the polynomials $p_k^\star$ from SSV and
bound $\norm{p_k^\star}_1.$ Even though we do not know $p_k^\star$
explicitly, we can bound its coefficients indirectly by writing it as
an interpolation polynomial. All these issues make the error analysis
highly technical. However, since the error analysis is crucial for our
algorithms, a more illuminating proof is highly desirable.
 
\subsubsection{Approximation Using Our Polynomial Approximation to $e^{-x}$}\label{sec:approx-poly}
More straightforwardly, combining the Lanczos method with the setting $B\defeq A$ and $f(x) \defeq e^{-x}$ along with the 
polynomial approximation to $e^{-x}$ that we prove in Theorem \ref{thm:exp-poly-approx}, we get that setting 
$k \approx \sqrt{\lambda_1(A)-\lambda_n(A)}\cdot\poly(\log
\nfrac{1}{\delta})$ suffices to obtain a vector $u$ that satisfies
$\norm{\exp(-A)v-u}\le \delta\norm{v}\norm{\exp(-A)}.$ This gives us
our second method for approximating $\exp(-A)v.$ 
Note that this algorithm avoids any inverse
computation and, as a result, the procedure and the proofs are simpler and the algorithm 
practical.

\subsection{Our Uniform Approximation for $e^{-x}$}
In this section, we give a brief overview of the proof of Theorem
\ref{thm:exp-poly-approx}. The details appear in Section
\ref{sec:poly} and can be read independently of the rest of the paper.
 
A straightforward approach to approximate $e^{-x}$ over $[a,b]$ is by
truncating its series expansion around $\frac{a+b}{2}.$ With a degree
of the order of $(b-a) + \log \nfrac{1}{\delta},$ these polynomials
achieve an error of $\delta \cdot e^{-\nfrac{(b+a)}{2}}$, for any
constant $\delta > 0.$ This approach is equivalent to approximating
$e^{\lambda}$ over $[-1,1],$ for $\lambda \defeq \nfrac{(b-a)}{2},$ by
polynomials of degree $O(\lambda + \log \nfrac{1}{\delta}).$ On the
flip side, it is known that if $\lambda$ is constant, the above result
is optimal (see e.g.~\cite{saff-rational}).
Instead of polynomials, one could consider approximations by rational
functions, as in \cite{rational1,rational2}. However, the author in
\cite{saff-rational} shows that, if both $\lambda$ and the degree of
the denominator of the rational function are constant, the required
degree of the numerator is only an additive constant better than that
for the polynomials. It might seem that the question of approximating
the exponential has been settled and one cannot do much
better. However, the result by SSV mentioned before, seems surprising
in this light. The lower bound does not apply to their result, since
the denominator of their rational function is unbounded.
In a similar vein, we ask the following question: If we are looking
for weaker error bounds, \emph{e.g.} $\delta\cdot e^{-a}$ instead of
$\delta\cdot e^{-\nfrac{(b+a)}{2}}$ (recall $b > a$), can we improve
on the degree bound of $O((b-a)+\log \nfrac{1}{\delta})$? Theorem
\ref{thm:exp-poly-approx} answers this question in the affirmative and
gives a new upper bound and an almost matching lower bound. We give an
overview of the proofs of both these results next.

\subsubsection*{\it Upper Bound.} We wish to show that there exists a
polynomial of degree of the order of $\sqrt{b-a}\cdot \poly(\log
\nfrac{1}{\delta})$ that approximates $e^{-x}$ on the interval
$[a,b],$ up to an error of $\delta\cdot e^{-a}$ for any $\delta > 0.$
Our approach is to approximate $(1+\nfrac{x}{k})^{-1}$ on the interval
$[a,b],$ by a polynomial $q$ of degree $l,$ and then compose the
polynomial $p_k^\star$ from the SSV result with $q$, to obtain
$p_k^\star(q(x))$ which is a polynomial of degree $k\cdot l$
approximating $e^{-x}$ over $[a,b]$. Thus, we are looking for
polynomials $q$ that minimize $|q(x)-\nfrac{1}{x}|$ over
$[1+\nfrac{a}{k},1+\nfrac{b}{k}]$. Slightly modifying the
optimization, we consider polynomials $q$ that minimize $|x\cdot q(x)
- 1|$ over $[1+\nfrac{a}{k},1+\nfrac{b}{k}]$. In
Section~\ref{sec:poly}, we show that the solution to this modified
optimization can be derived from the well-known Chebyshev
polynomials. For the right choice of $k$ and $l$, the composition of
the two polynomials approximates $e^{-x}$ to within an error of
$\delta \cdot e^{-a}$ over $[a,b],$ and has degree $\sqrt{b-a}\cdot
\poly(\log \nfrac{1}{\delta})$ . To bound the error in the composition
step, we need to bound the sum of absolute values of coefficients of
$p_k^\star,$ which we achieve by rewriting $p_k^\star$ as an
interpolation polynomial. The details appear in Section
\ref{sec:poly}.

\subsubsection*{\it Lower Bound.}
As already noted, since we consider a weaker error bound $\delta\cdot
e^{-a}$ and $\lambda \defeq \nfrac{(b-a)}{2}$ isn't a constant for our
requirements, the lower bounds mentioned above no longer
hold. Nevertheless, we prove that the square-root dependence on $b-a$
of the required degree is optimal. The proof is simple and we give the
details here: Using a theorem of Markov from approximation theory
(see~\cite{cheney-book}), we show that, any polynomial approximating
$e^{-x}$ over the interval $[a,b]$ up to an error of $\delta\cdot
e^{-a},$ for some constant $\delta$ small enough, must have degree of
the order of $\sqrt{b-a}.$ Markov's theorem says that the absolute
value of the derivative of a univariate polynomial $p$ of degree $k,$
which lives in a box of height $h$ over an interval of width $w,$ is
upper bounded by $\nfrac{d^2h}{w}.$ Let $p_k$ be a polynomial of
degree $k$ that $\delta\cdot e^{-a}$-approximates $e^{-x}$ in the
interval $[a,b].$ If $b$ is large enough and, $\delta$ a small enough
constant, then one can get a lower bound of $\Omega(e^{-a})$ on the
derivative of $p_k$ using the Mean Value Theorem. Also, one can
obtain an upper bound of $O(e^{-a})$ on the height of the box in which
$p_k$ lives. Both these bounds use the fact that $p_k$ approximates
$e^{-x}$ and is $\delta \cdot e^{-a}$ close to it. Since the width of
the box is $b-a,$ these two facts, along with Markov's theorem,
immediately imply a lower bound of $\Omega(\sqrt{b-a})$ on $k.$ This
shows that our upper bound is tight up to a factor of $\poly(\log
\nfrac{1}{\delta}).$

\section{Discussion and Open Problems}
\label{sec:open}
In this paper, using techniques from disparate areas such as random walks, SDPs, numerical linear algebra and approximation theory, we have settled the question of designing an asymptotically optimal $\tilde{O}(m)$ spectral algorithm for {BS} (Theorem \ref{thm:bsmain}) and alongwith provided a simple and practical algorithm (Theorem \ref{thm:bsmain2}). 
However, there are several outstanding problems that emerge from our work. 

The
main remaining open question regarding the design of spectral
algorithms for \BS is whether it is possible to obtain stronger
certificates that no sparse balanced cuts exist, in nearly-linear time.
This question is of practical importance in the construction of
decompositions of the graph into induced graphs that are
near-expanders, in nearly-linear time~\cite{ST2}. OV show that their
certificate, which is of the same form as that of \alg, is stronger
than the certificate of Spielman and Teng~\cite{ST2}. In particular,
our certificate can be used to produce decompositions into components
that are guaranteed to be subsets of induced expanders in $G.$
However, this form of certificate is still much weaker than that given
by {\sc RLE}, which actually outputs an induced expander of large volume.

With regards to approximating the Matrix exponential, a computation which plays an important role in SDP-based algorithms, random walks, numerical linear algebra and quantum computing, settling the
hypothesis remains the main open question. Further, as noted earlier, the error analysis plays a crucial role in making Theorem \ref{thm:expRational} and, hence, Theorem \ref{thm:bsmain} work, but its proof is rather long and difficult. A more illuminating proof of this would be highly desirable.

Another question is to close the gap between the upper and
lower bounds on polynomial approximations to $e^{-x}$ over an interval
$[a,b]$  in Theorem \ref{thm:exp-poly-approx}.

\section{The Algorithm for Balanced Separator}
\label{sec:balsep-main}
In this section we provide our spectral algorithm {\alg} and prove Theorem \ref{thm:bsmain}. We also mention how Theorem \ref{thm:bsmain2} follows easily from the proof of Theorem \ref{thm:bsmain} and Theorem \ref{thm:expPoly}. We first present the preliminaries for this section. 

\subsection{Basic Preliminaries}

\newcommand{\1}{{\bf 1}}
\newcommand{\0}{{\bf 0}}

\subsubsection*{\it Instance Graph and Edge Volume.} We denote by $G=(V,E)$ the unweighted instance graph, where $V = [n]$ and $\card{E}=m.$ We assume $G$ is connected. We let $d \in \R^n$ be the degree vector of $G,$ i.e. $d_i$ is the degree of vertex $i.$ For a subset $ S \subseteq V,$ we define the edge volume as $\vol(S) \defeq \sum_{i \in S} d_i.$ The total volume of $G$ is $2m.$ 
The conductance of
a cut $(S,\bar{S})$ is defined to be $\phi(S) \defeq \nfrac{|E(S,\bar{S})|}{\min
\{\vol(S),\vol (\overline{S})\}},$ where $\vol(S)$ is the sum of the degrees of the vertices in the set $S$. 
Moreover, a cut $(S, \bar{S})$ is $b$-balanced if $\min \{\vol(S), \vol(\bar{S})\} \geq b \cdot \Vol{V}.$

\subsubsection*{\it Special Graphs}
We denote by $K_V$ the complete graph with weight $\nfrac{d_id_j}{2m}$ between every pair $i,j \in V.$
For $i \in V,$ $S_i$ is the star graph rooted at $i$, with edge weight of $ \nfrac{d_id_j}{2m}$ between $i$ and $j,$ for all $j \in V.$

\subsubsection*{\it Graph matrices.} For an undirected graph $H=(V, E_H)$, let
$A(H)$ denote the adjacency matrix of $H$, and $D(H)$ the diagonal matrix of
degrees of $H$.
The (combinatorial) Laplacian of $H$ is defined as $L(H) \defeq D(H) - A(H)$.
Note that for all $x \in \R^{V}$, $x^\top L(H) x = \sum_{\{i,j\} \in E_H} (x_i - x_j)^2$. 
By $D$ and $L$, we denote $ D(G)$ and $L(G)$ respectively for the input graph $G.$ Finally, the natural random walk over $G$ has transition matrix $W \defeq AD^{-1}.$

\subsubsection*{\it Vector and Matrix Notation.}
We are working within the vector space $\mathbb{R}^n.$ We will denote
by $I$ the identity matrix over this space.  For a symmetric matrix
$A,$ we will use $A \succeq 0$ to indicate that $A$ is positive
semi-definite.
The expression $A \succeq B$ is equivalent to $A - B \succeq 0$. For two matrices $A,B$ of equal dimensions, let $A \bullet B \defeq \Tr(A^\top B) = \sum_{ij} A_{ij}\cdot B_{ij}.$ We denote by $\{e_i\}_{i=1}^n$ the standard basis for $\R^n.$ ${\bf 0}$ and $\1$ will denote the all $0$s and all $1$s vectors respectively.

\begin{fact}\label{fct:decomp}
$L(K_V) = D - \nfrac{1}{2m} \cdot D\1\1^\top D =  D^{\nfrac{1}{2}} (I - \nfrac{1}{2m} \cdot D^{\nfrac{1}{2}} \1 \1 D^{\nfrac{1}{2}})  D^{\nfrac{1}{2}}.$
\end{fact}

\subsubsection*{\it Embedding Notation.} We will deal with vector embeddings of $G$, where each vertex $i \in V$ is mapped to a vector $v_i \in \R^d$, for some $d \leq n.$ For such an embedding $\{v_i\}_{i \in V},$ we denote by $v_\avg$ the mean vector, i.e. $v_\avg \defeq \sum_{i \in V} \nfrac{d_i}{2m} \cdot v_i.$
Given a vector embedding $\{v_i \in \R^d\}_{i \in V},$ recall that $X$
is the Gram matrix of the embedding if $X_{ij} = v_i^\top v_j.$ A Gram
matrix $X$ is always PSD, \emph{i.e.}, $X \succeq 0$. For any $X \in \R^{n \times n}, X\succeq 0,$ we call $\{v_i\}_{i \in V}$ the {\it embedding corresponding to $X$} if $X$ is the Gram matrix of $\{v_i\}_{i \in V}.$ 
For $i \in V,$ we denote by $R_i$ the matrix such that $R_i \bullet X = \norm{v_i - v_\avg}^2.$

\begin{fact}\label{fct:degkv}
$ \sum_{i \in V} d_i R_i \bullet X= \sum_{i \in V} d_i \norm{v_i - v_\avg}^2 = \nfrac{1}{2m} \cdot \sum_{i<j } d_j d_i \norm{v_i - v_j}^2 =   L(K_V) \bullet X.$
\end{fact}

\subsection{{\sc AHK} Random Walks} \label{subsec:ahk}

The random-walk processes used by our algorithm are continuous-time Markov processes~\cite{parzen1999} over $V.$ In these processes, state transitions do not take place at specified discrete intervals, but follow exponential distributions described by a {\it transition rate matrix} $Q \in R^{n \times n}$, where $Q_{ij}$ specifies the rate of transition from vertex $j$ to $i$.
More formally, letting $p(\tau) \in \R^n$ be the probability distribution of the process at time $t \geq 0$, we have that
$\nfrac{\partial p(\tau)}{\partial \tau} = Q p(\tau)$
Given a transition rate matrix
\footnote{A matrix $Q$ is a valid transition rate matrix if its diagonal entries are non-positive and its off-diagonal entries are non-negative. 
Moreover, it must be that ${\bf 1} Q = 0,$ to ensure that probability mass is conserved.
} 
$Q$, the differential equation for $p(\tau)$ implies that
$
p(\tau) = e^{\tau Q} p(0).
$
In this paper, we will be interested in a class of continuous-time Markov processes over $V$ that take into account the edge structure of $G$. The simplest such process is the {\it heat kernel} process, which is defined as having transition rate matrix $Q =- (I - W) = -LD^{-1}.$ 
The heat kernel can also be interpreted as the probability transition matrix of the following discrete-time random walk: 
sample a number of steps $i$ from a Poisson distribution with mean $\tau$ and perform $i$ steps of the natural random walk over $G:$
$$
p(\tau) = e^{-\tau LD^{-1}} p(0) = e^{-\tau(I-W)} p(0)= e^{-\tau} \sum_{i=0}^\infty \frac{\tau^i}{i!} W^i p(0). 
$$

For the construction of our algorithm, we generalize the concept of heat kernel to a larger class of continuous-time Markov processes, which we name {\it Accelerated Heat Kernel} (\textsc{AHK}) processes.
A process $\cH(\beta)$ in this class is defined by a non-negative vector $\beta \in \R^n$ and the transition rate matrix of $\cH(\beta)$ is $Q(\beta) \defeq - (L + \sum_{i \in V} \beta_i L(S_i))D^{-1}.$ As this is the negative of a sum of Laplacian matrices, it is easy to verify that it is a valid transition rate matrix.  The effect of adding the star terms to the transition rate matrix is that of accelerating the convergence of the process to stationary at vertices $i$ with large value of $\beta_i$, as a large fraction of the probability mass that leaves these vertices is distributed uniformly over the edges. We denote by $P_\tau(\beta)$ the probability-transition matrix of $\cH(\beta)$ between time $0$ and $\tau,$ i.e. $P_\tau(0) = e^{\tau Q(\beta)}.$

\subsubsection*{\it Embedding View.}
A useful matrix to study $\cH(\beta)$ will be $D^{-1} P_{2\tau} (\beta).$ This matrix describes the probability distribution over the edges of $G$ and has the advantage of being symmetric and positive semidefinite: 
$$ 
D^{-1} P_{2\tau}(\beta) = D^{-\nfrac{1}{2}} e^{-(2\tau) D^{-\nfrac{1}{2}} (L + \sum_{i \in V} \beta_i L(S_i)) D^{-\nfrac{1}{2} }} D^{-\nfrac{1}{2}},
$$
Moreover, we have the following fact:
\begin{fact} \label{fct:sqroot}
$D^{-\nfrac{1}{2}} P_{\tau}(\beta)$ is a square root of $D^{-1} P_{2\tau}(\beta).$ 
\end{fact}
\begin{proof}
$$
\left(D^{-\nfrac{1}{2}} P_{\tau}(\beta)\right)^\top D^{-\nfrac{1}{2}} P_{\tau}(\beta) = e^{\tau(Q(\beta))^\top} D^{-1} e^{\tau Q(\beta)} =
D^{-1} e^{\tau Q(\beta)} e^{\tau Q(\beta)} = D^{-1} e^{2\tau Q(\beta)}. 
$$
\end{proof}

\noindent
Hence, $D^{-1} P_{2\tau}(\beta)$ is the Gram matrix of the embedding given by the columns of its square root $D^{-\nfrac{1}{2}} P_{\tau}(\beta).$ This property will enable us to use geometric \SDP techniques to analyze $\cH(\beta).$

\subsubsection*{\it Mixing.}
Spectral methods for finding low-conductance cuts are based on the idea that random walk processes mix slowly across sparse cuts, so that it is possible to detect such cuts by considering the starting vertices for which the probability distribution of the process strongly deviates from stationary.
We measure this deviation for vertex $i$ at time $t$ by the $\ell_2^2$-norm of the distance between $P_\tau(\beta) e_i$ and the uniform distribution {\it over the edges of }$G.$ We denote it by $\Psi(P_\tau(\beta), i):$
$$
\Psi(P_\tau(\beta), i) \defeq d_i \sum_{j \in V}  d_j \left(\frac{e_j^\top P_\tau(\beta)e_i}{d_j} - \frac{1}{2m}\right)^2 
$$

A fundamental quantity for our algorithm will be the total deviation from stationarity over a subset $S \subseteq V.$ We will denote $\Psi(P_t(\beta), S) \defeq \sum_{i \in S} \Psi(P_t(\beta), i).$
In particular, $\Psi(P_{\tau}(\beta), V )$  will play the role of potential function in our algorithm. The following facts express these mixing quantities in the geometric language of the embedding corresponding to $D^{-1} P_{2\tau}(\beta).$
\begin{fact}
$\Psi(P_\tau(\beta), i) = d_i R_i \bullet D^{-1} P_{2\tau}(\beta).$
\end{fact}
\begin{proof}
By Fact~\ref{fct:sqroot} and the definition of $R_i:$
\begin{align*}
d_i R_i \bullet D^{-1} P_{2\tau}(\beta) = d_i \norm{D^{-\nfrac{1}{2}} P_{\tau}(\beta) e_i - \sum_{j \in V} \frac{d_j}{2m} D^{-\nfrac{1}{2}} P_{\tau}(\beta) e_j}^2 = d_i \norm{D^{-\nfrac{1}{2}} P_{\tau}(\beta) e_i - \frac{D^{\nfrac{1}{2}}\1}{2m}}^2\\
=d_i \norm{D^{\nfrac{1}{2}} \left(D^{-1} P_{\tau}(\beta) e_i - \frac{\1}{2m}\right)}^2 = \Psi(P_\tau(\beta), i).
\end{align*}
\end{proof}

The following is a consequence of Fact~\ref{fct:degkv}:
\begin{fact} \label{fct:totdev}
$\Psi(P_\tau(\beta), V) =\sum_{i \in V} d_i R_i \bullet D^{-1} P_{2\tau}(\beta) = L(K_V) \bullet D^{-1} P_{2\tau}(\beta). $
\end{fact}

\subsection{Algorithm Description}\label{sec:balsepalgo}

\subsubsection*{\it Preliminaries}
All the random walks in our algorithm will be run for time $\tau \defeq \nfrac{O(\log n)}{\gamma}.$ 
We will consider embeddings given by the columns of $\Degin P_{\tau}(\beta)$ for some choice of $\beta.$
Because we want our algorithm to run in time $\tilde{O}(m)$ and we are only interested in Euclidean distances between vectors in the embedding, 
we will use the Johnson-Lindenstrauss Lemma (see Lemma~\ref{lem:jl} in Section~\ref{sec:proofs}) to obtain an $O(\log n)$-dimensional embedding approximately preserving distances between columns of $\Degin P_\tau(\beta)$ up to a factor of $(1+\eps),$ where $\eps$ is a constant such that $\nfrac{1+\eps}{1-\eps} \leq \nfrac{4}{3}.$

Our algorithm \alg will call two subroutines {\sc FindCut} and {\sc ExpV}. 
{\sc FindCut} is an \SDP-rounding algorithm that uses random projections and radial sweeps to find a low-conductance cut, that is either $c$-balanced, for some constant $c=\Omega(b) \leq \nfrac{b}{100}$ defined in OV, or obeys a strong guarantee stated in Theorem~\ref{thm:findcut}. 
Such algorithm is implicit in \cite{OV} and is described precisely in Section~\ref{sec:findcut}.
{\sc ExpV} is a generic algorithm that approximately computes products of the form $P_\tau(\beta) u$ for unit vectors $u.$ {\sc Expv} can be chosen to be either the algorithm implied by Thereom~\ref{thm:expRational2}, which makes use of the Spielman-Teng solver, or that in Theorem~\ref{thm:expPoly}, which just applies the Lanczos method.

We are now ready to describe \alg, which will output a $c$-balanced cut of conductance $O(\sqrt{\gamma})$ or the string {\sc NO}, if it finds a certificate that no $b$-balanced cut of conductance less than $\gamma$ exists. \alg can also fail and output the string {\sc Fail}. We will show that this only happens with small probability. The algorithm \alg is defined in Figure~\ref{fig:algorithm}. The constants in this presentation are not optimized and are likely to be higher than what is necessary in practice. They can also be modified to obtain different trade-offs between the approximation guarantee and the output balance.

\begin{figure*}[htb]
  	\begin{tabularx}{\textwidth}{|X|}
	\hline
		
	{\bf Input:} An unweighted connected instance graph $G=(V,E),$ a constant balance value $b \in (0, \nfrac{1}{2}],$ 
	a conductance value $\gamma \in [\nfrac{1}{n^2},1).$
	\vspace{1mm}

	Let $S=0, \bar{S} = V.$ Set $\tau = \nfrac{\log n}{12 \gamma}$ and $\beta^{(1)} = \0.$
	\\
	At iteration $t=1, \ldots, T  = 12 \log n:$
	\begin{enumerate}[label=\arabic*.]
	 
	\item Denote $P^{(t)} \defeq P_{\tau}(\beta^{(t)}).$ Pick $k = O(\nfrac{\log n}{\eps^2})$ random unit vectors $\{u^{(t)}_1, u^{(t)}_2, \ldots, u^{(t)}_k \in \R^{n}\}$ and use the subroutine {\sc ExpV} to
	compute the embedding $\{v^{(t)}_i \in \R^k\}_{i \in V}$ defined as 
	\begin{equation*}
\left(v^{(t)}_i\right)_j = \sqrt{\frac{n}{k}} u_j^\top D^{-\nfrac{1}{2}} P^{(t)} e_i.
	\end{equation*} 
	Let $X^{(t)}$ be the Gram matrix corresponding to this embedding.
	
	\item If $L(K_V) \bullet X^{(t)} = \sum_{i \in V} d_i ||v^{(t)}_i - v^{(t)}_\avg||^2 \leq \frac{1+\eps}{n},$ 
	output {\sc NO} and terminate.
	
	\item Otherwise, run {\sc FindCut}$(G, b, \gamma, \{v^{(t)}_i\}_{i \in V}).$ {\sc FindCut} outputs a cut $S^{(t)}$ with $\phi(S^{(t)}) \leq O(\sqrt{\gamma})$ or fails,
	in which case we also output {\sc Fail} and terminate.
	
	\item If $S^{(t)}$ is $c$-balanced, output $S^{(t)}$ and terminate. If not, update $S \defeq S \cup S^{(t)}$. 
	If $S$ is $c$-balanced, output $S$ and terminate.
	
	\item Otherwise, update $\beta^{(t+1)} = \beta^{(t)} + \frac{ 72 \gamma}{T} \sum_{i \in S^{(t)}} e_i$  and proceed to the next iteration.
	\end{enumerate}
	
	\vspace{1mm}
	
	Output {\sc NO} and terminate.
\\
\hline 
\end{tabularx}
  \caption{The \alg Algorithm}
  \label{fig:algorithm}
\end{figure*}

At iteration $t=1,$ we have $\beta^{(1)} = \0,$ so that $P^{(1)}$ is just the probability transition matrix
of the heat kernel on $G$ for time $\tau.$ In general at iteration $t,$ \alg runs {\sc ExpV} to compute $O(\log n)$ random projections of $P^{(t)}$ and constructs an approximation $\{v^{(t)}_i\}_{i \in V}$ to the embedding given by the columns of $D^{-\nfrac{1}{2}} P^{(t)}.$
This approximate embedding has Gram matrix $X^{(t)}.$

In Step $2,$ \alg computes $L(K_V) \bullet X^{(t)},$ which is an estimate of the total deviation $\Psi(P^{(t)}, V)$ by Fact~\ref{fct:totdev}. If this deviation is small, the {\sc AHK} walk $P^{(t)}$ has mixed sufficiently over $G$ to yield a certificate that $G$ cannot have any $b$-balanced cut of conductance less than $\gamma.$ This is shown in Lemma~\ref{lem:pot}.
If the {\sc AHK} walk $P^{(t)}$ has not mixed sufficiently, we can use {\sc FindCut} to find a cut $S^{(t)}$ of low conductance $O(\sqrt{\gamma}),$ which is an obstacle for mixing. If $S^{(t)}$ is $c$-balanced , we output it and terminate. Similarly, if $S \cup S^{(t)}$ is $c$-balanced, as $\phi(S \cup S^{(t)}) \leq O(\sqrt{\gamma}),$ we can also output $S \cup S^{(t)}$ and exit.
Otherwise, $S^{(t)}$ is unbalanced and is potentially preventing \alg from detecting balanced cuts in $G.$ We then proceed to modify the {\sc AHK} walk, by increasing the values of $\beta^{(t+1)}$ for the vertices in $S^{(t)}.$ This change ensures that $P^{(t+1)}$ mixes faster from the vertices in $S^{(t)}$ and in particular mixes across $S^{(t)}.$
In particular, this means that, at any given iteration $t,$ the support of $\beta^{(t)}$ is $\cup_{r=1}^{t-1} S^{(r)},$ which is an unbalanced set.

The \alg algorithm exactly parallels the {\sc RLE} algorithm, introducing only two fundamental changes. 
First, we use the embedding given by the {\sc AHK} random walk $P^{(t)}$ in place of the eigenvector to find cuts in $G$ or in a residual graph. 
Secondly, rather than fully removing unbalanced low-conductance cuts from the graph, we modify $\beta^{(t)}$ at every iteration $t,$ so $P^{(t+1)}$ at the next iteration mixes across the unbalanced cuts found so far.

\subsection{Analysis}\label{sec:gpanalysis}

\renewcommand{\mwu}{MMWU\xspace}

The analysis of \alg is at heart a modification of the \mwu argument in OV, stated in a random-walk language. 
This modification allows us to deal with the different embedding used by \alg at every iteration with respect to OV.

In this analysis, the quantity $\Psi(P^{(t)}, V)$ plays the role of potential function.
Recall that, from a random-walk point of view, $\Psi(P^{(t)}, V)$ is the total deviation from stationarity of $P_{\tau}(\beta^{(t)})$ over all vertices as starting points. We start by showing that if the potential function is small enough, we obtain a certificate that no $b$-balanced cut of conductance at most $\gamma$ exists. In the second step, we show that, if an unbalanced cut $S^{(t)}$ of low conductance is found, the potential decreases by a constant fraction. Unless explicitly stated otherwise, all proofs are found in Section~\ref{sec:proofs}.

\subsubsection*{\it Potential Guarantee.} We argue that, if $\Psi(P^{(t)}, V)$ is sufficiently small, it must be the case that $G$ has no $b$-balanced cut of conductance less than $\gamma.$ A similar result is implicit in OV. This theorem has a simple explanation in terms of the {\sc AHK} random walk $P^{(t)}.$ Notice that $P^{(t)}$ is accelerated only on a small unbalanced set $S.$ Hence, if a balanced cut of conductance less than $\gamma$ existed, its convergence could not be greatly helped by the acceleration over $S.$ Then, if $P^{(t)}$ is still mixing very well, no such balanced cut can exist. 
\begin{lemma}\label{lem:pot}
Let $S = \cup_{i=1}^t S^{(i)}.$
If $\Psi(P^{(t)}, V) \leq \frac{4}{3n},$ and $\vol(S) \leq c \cdot 2m \leq \nfrac{b}{100} \cdot 2m$, then
$$
L + \sum_{i \in V} \beta_i^{(t)}  L(S_i) \succeq 3\gamma \cdot L(K_V).
$$
Moreover, this implies that no $b$-balanced cut of conductance less than $\gamma$ exists in $G.$
\end{lemma}

\subsubsection*{\it The Deviation of an Unbalanced Cut.}
In the next step, we show that, if the walk has not mixed sufficiently,  w.h.p.  the embedding $\{v^{(t)}_i\}_{i \in V},$ computed by \alg,  has low quadratic form with respect to the Laplacian of $G.$ From a \SDP-rounding perspective, this means that the embedding can be used to recover cuts of value close to $\gamma.$ This part of the analysis departs from that of OV, as we use our modified definition of the embedding.
\begin{lemma}\label{lem:obj}
If $\Psi(P^{(t)}, V) \geq \frac{1}{n},$ then w.h.p. $L \bullet X^{(t)} \leq O(\gamma) \cdot L(K_V) \bullet X^{(t)}.$ 
\end{lemma}
\noindent
This guarantee on the embedding allows us to apply \SDP-rounding techniques in the subroutine {\sc FindCut}. The following result is implicit in~\cite{OV}. Its proof appears in Section~\ref{sec:findcut} for completeness.
\begin{theorem}\label{thm:findcut}
Consider an embedding $\{v_i \in \R^d\}_{i \in V}$ with Gram matrix $X$ such that $L \bullet X^{(t)} \leq \alpha L(K_V) \bullet X^{(t)},$ for $\alpha > 0.$ On input $(G,b,\alpha, \{v_i\}_{i \in V}),$ {\sc FindCut} runs in time $\tilde{O}(md)$ and w.h.p. outputs a cut $C$ with $\phi(C) \leq O(\sqrt{\alpha}).$ Moreover, there is a constant $c = \Omega(b) \leq \nfrac{b}{100}$ such that either $C$ is $c$-balanced or 
$$
\sum_{i \in C} d_i R_i \bullet X \geq \nfrac{2}{3} \cdot L(K_V) \bullet X. 
$$
\end{theorem}
\noindent
The following corollary is a simple consequence of Lemma~\ref{lem:obj} and Theorem~\ref{thm:findcut}:
\begin{corollary}\label{cor:unbalanced}
At iteration $t$ of \alg, if $\Psi(P^{(t)}, V) \geq \frac{1}{n}$ and $S^{(t)}$ is not $c$-balanced, then w.h.p. $\Psi(P^{(t)}, S) \geq \nfrac{1}{2} \cdot \Psi(P^{(t)}, V).$
\end{corollary}
In words, at the iteration $t$ of \alg, the cut $S^{(t)}$ must either be $c$-balanced or be an unbalanced cut that contributes a large constant fraction of the total deviation of $P^{(t)}$ from the stationary distribution.
In this sense, $S^{(t)}$ is the main reason for the failure of $P^{(t)}$ to achieve better mixing. To eliminate this obstacle and drive the potential further down, $P^{(t)}$ is updated to $P^{(t+1)}$ by accelerating the convergence to stationary from all vertices in $S^{(t)}.$ Formally, this is achieved by adding weighted stars rooted at all vertices over $S^{(t)}$ to the transition-rate matrix of the {\sc AHK} random walk $P^{(t)}$. 

\subsubsection*{\it Potential Reduction.}
The next theorem crucially exploits the stability of the process $\cH(\beta^{(t)})$ and Corollary~\ref{cor:unbalanced} to show that the potential decreases by a constant fraction at every iteration in which an unbalanced cut is found.
More precisely, the theorem shows that accelerating the convergence from $S^{(t)}$ at iteration $t$ of \alg has the effect of eliminating at least a constant fraction of the total deviation due to $S^{(t)}.$
The proof is a simple application of the Golden-Thompson inequality~\cite{Bhatia} and mirrors the main step in the \mwu analysis.
\begin{theorem}\label{thm:reduction}
At iteration $t$ of \alg, if $\Psi(P^{(t)}, V) \geq \frac{1}{n}$ and $S^{(t)}$ is not $c$-balanced, then w.h.p.
$$
\Psi(P^{(t+1)}, V) \leq \Psi(P^{(t)}, V) - \nfrac{1}{3} \cdot \Psi(P^{(t)}, S^{(t)}) \leq \nfrac{5}{6} \cdot \Psi(P^{(t)}, V).
$$
\end{theorem}

We are now ready to prove Theorem~\ref{thm:bsmain} and Theorem~\ref{thm:bsmain2} by applying Lemma~\ref{thm:reduction} to show that after $O(\log n)$ iterations, the potential must be sufficiently low to yield the required certificate according to Lemma~\ref{lem:pot}.

\begin{proof}[Proof of Theorem~\ref{thm:bsmain}]
If \alg outputs a cut $S$ in Step $4$, by construction, we have that $\phi(S) \leq O(\sqrt{\gamma})$
and $S$ is $\Omega(b)$-balanced. 
Alternatively, at iteration $t,$ if $L(K_V) \bullet X^{(t)} \leq \nfrac{1+\eps}{n},$ we have by Lemma~\ref{lem:jl} that
$$
\Psi(P^{(t)}, V) = L(K_V) \bullet D^{-1} P_{2\tau}(\beta^{(t)}) \leq \frac{1}{1-\eps} L(K_V) \bullet X^{(t)} \leq \frac{1+\eps}{1-\eps} \cdot \frac{1}{1n}\leq \frac{4}{3n}.
$$
Therefore, by Lemma~\ref{lem:pot}, we have a certificate that no $b$-balanced cut of conductance less than $\gamma$ exists in $G.$
Otherwise, we must have $L(K_V) \bullet X^{(t)} \geq \nfrac{1+\eps}{n},$
which, by Lemma~\ref{lem:jl}, implies that
$$
\Psi(P^{(t)}, V) \geq \nfrac{1}{n}.
$$
Then, by Lemma~\ref{lem:obj} and Theorem~\ref{thm:findcut}, we have w.h.p. that {\sc FindCut} does not fail and outputs a cut $S^{(t)}$ with $\phi(S^{(t)}) \leq O(\sqrt{\gamma}).$ As \alg has not terminated in Step $4,$ it must be the case that $S^{(t)}$ is not $c$-balanced and, by Theorem~\ref{thm:reduction}, we obtain that w.h.p. $ \Psi(P^{(t+1)},V) \leq \nfrac{5}{6} \cdot \Psi(P^{(t)}, V).$ 
Now, 
$$
\Psi(P^{(1)}, V) = L(K_V) \bullet D^{-1} P_{2\tau}(\0) \leq I \bullet P_{2\tau}(\0)  \leq n.
$$
Hence, after $\nfrac{2 \log n}{\log(\nfrac{6}{5})} \leq 12 \log n= T$ iterations, w.h.p.
we have that $\Psi(P^{(T)}, V) \leq \nfrac{1}{n}$ and, by Lemma~\ref{lem:pot}, no $b$-balanced cut of conductance less than $\gamma$ exists.

We now consider the running time required by the algorithm at every iteration.
In Step $1,$ we compute $k = O(\log n),$ products of the form $D^{-\nfrac{1}{2}} P^{(t)}u,$ where $u$ is an unit vector,
using the {\sc ExpV} algorithm based on the Spielman-Teng solver, given in Theorem~\ref{thm:expRational2}. 
This application of Theorem~\ref{thm:expRational2} is explained in Section~\ref{sec:overview-exp}.
By the definition of $\beta^{(t)},$ at iteration $t$ we have:
\begin{align*}
\norm{HMH} = \norm{\tau \Degin (L + \sum_{i \in V} \nfrac{d_i}{2m} \cdot  \beta_i D + \sum_{i \in V}  \nfrac{d_i}{2m} \cdot  \beta_i e_i e_i^\top) \Degin} \leq \norm{\tau \Degin (L + 2\cdot 72 \cdot \gamma D) \Degin} \\
\leq O(\tau) = \poly(n). 
\end{align*}
Moreover, it is easy to see that our argument is robust up to an error $\delta = \nfrac{1}{\poly(n)}$ in this computation 
and the sparsity of $M$ is $O(m)$ so that the running time of a single matrix-exponential-vector product is $\tilde{O}(m).$
Given the embedding produced by Step $1,$ $L(K_V) \bullet X^{(t)}$ can be computed in time $\tilde{O}(nk)= \tilde{O}(n)$ by computing the distances $||v^{(t)}_i - v^{(t)}_\avg||^2$ for all $i \in V.$
By Theorem~\ref{thm:findcut}, Step $3$ runs in time $\tilde{O}(mk) = \tilde{O}(m).$ Finally, both Steps $4$ and $5$ can be performed in time $\tilde{O}(m).$ As there are at most $O(\log n)$ iterations, the theorem follows.
\end{proof}

\noindent
Theorem~\ref{thm:bsmain2} is proved similarly. It suffices to show that a single matrix-exponential-vector product requires time $\tilde{O}(\nfrac{m}{\sqrt{\gamma}}).$
\begin{proof}[Proof of Theorem~\ref{thm:bsmain2}]
Using the algorithm of Theorem~\ref{thm:expPoly}, we obtain that 
$\norm{A} \leq O(\tau),$ so that $k = \tilde{O}(\sqrt{\tau})=\tilde{O}(\nfrac{1}{\sqrt{\gamma}}) \leq \tilde{O}(n).$
Hence, the running time of a single computation for this method is $\tilde{O}(m\sqrt{\tau}) =\tilde{O}(\nfrac{m}{\sqrt{\gamma}}).$  
\end{proof}

\subsection{Proofs}\label{sec:proofs}
In this section we provide the proofs from the Section \ref{sec:balsep-main}. We start with some preliminaries.

\subsubsection{Preliminaries}

\subsubsection*{\it Vector and Matrix Notation.}
For a symmetric matrix $A,$ denote by $\lambda_i(A),$ the $i^\text{th}$ smallest eigenvalue of $A.$ For a vector $x \in \mathbb{R}^{n},$ let ${\rm supp}(x)$ be the set of vertices where $x$ is not zero.

\begin{fact} \label{fct:eigen}
$L \preceq 2 \cdot D$ and $L(S_i) \preceq 2 \cdot D.$ 
\end{fact}
\begin{fact}\label{fct:stardomination}
For all $i \in V,$ $L(S_i) = \nfrac{d_i}{2m} \cdot L(K_V) + d_i R_i.$
In particular, $L(S_i) \succeq d_i R_i.$
\end{fact}

\subsubsection*{\it Notation for \alg.}
At iteration $t,$ we denote 
$$
C^{(t)} \defeq D^{-\nfrac{1}{2}} Q(\beta^{(t)}) D^{\nfrac{1}{2}} = 
D^{-\nfrac{1}{2}} (L + \sum_{i \in V} \beta^{(t)}_i L(S_i)) D^{-\nfrac{1}{2}}.
$$
The following are useful facts to record about $C^{(t)}:$

\begin{fact}\label{fct:eigenvalue}
The vector $D^{\nfrac{1}{2}}$ is the eigenvector of $C^{(t)}$ with smallest eigenvalue $0.$
\end{fact}
\begin{fact}\label{lem:upbound}
$C^{(t)} \preceq O(1) \cdot I.$
\end{fact}

\subsubsection{Useful Lemmata}

\begin{lemma} \label{lem:tracedecomp}
$
\Psi(P^{(t)}, V) = L(K_V) \bullet D^{-1} P_{2\tau}(\beta^{(t)}) = \Tr(e^{-2\tau C^{(t)}}) - 1 .
$
\end{lemma}
\begin{proof}
By definition, we have
$$
\Psi(P^{(t)}, V) = L(K_V) \bullet D^{-1} P_{2\tau}(\beta^{(t)}) = 
 L(K_V) \bullet D^{-1} e^{-2\tau Q(\beta^{(t)}) } =
L(K_V) \bullet \Degin e^{-2\tau C^{(t)}} \Degin.
$$
Using Fact~\ref{fct:decomp} and the cyclic property of the trace function,
we obtain
$$
L(K_V) \bullet \Degin e^{-2\tau C^{(t)}} \Degin = 
(I - \nfrac{1}{2m} D^{\nfrac{1}{2}} \1 \1 D^{\nfrac{1}{2}}) \bullet e^{-2\tau C^{(t)}}. 
$$
Finally, by Fact~\ref{fct:eigenvalue}, we must have that the right-hand side equals $\Tr(e^{-2\tau C^{(t)}}) - 1,$ as required.
\end{proof}

\noindent
The following lemma is a simple consequence of the convexity of $e^{-x}.$ It is proved in~\cite{Lthesis}.
\begin{lemma}\label{lem:mexptolin2}
For a symmetric matrix $A \in \mathbb{R}^{n \times n}$ such that $\rho I \succeq A \succeq 0$ and $\tau > 0,$ we have
$$
e^{-\tau A} \preceq \left(I - \frac{(1-e^{-\tau \rho})}{\rho}A\right).
$$
\end{lemma}

\noindent
The following are standard lemmata.
\begin{lemma}[Golden-Thompson inequality~\cite{Bhatia}]\label{lem:gt}
Let $X,Y \in \mathbb{R}^{n \times n}$ be symmetric matrices. Then, $$\Tr\left({e^{X+Y}}\right) \leq \Tr\left({e^X e^Y}\right).$$
\end{lemma}

\begin{lemma}[Johnson-Lindenstrauss]\label{lem:jl}
Given an embedding $\{v_i \in \mathbb{R}^d\}_{i \in V}$, $V=[n],$ let $u_1, u_2, \ldots, u_k$, be vectors sampled independently uniformly from the $n-1$-dimensional sphere of radius $\sqrt{\nfrac{n}{k}}.$ Let $U$ be the $k \times t$ matrix having the vector $u_i$ as $i^\text{th}$ row and let $\tilde{v}_i \defeq  U v_i$. Then, for $k_\eps \defeq O(\nfrac{\log n}{\delta^2}),$ for all $i, j \in V$ 
$$
(1 - \eps) \cdot \norm{v_i - v_j}^2 \leq \norm{\tilde{v}_i - \tilde{v}_j}^2 \leq (1 + \eps)\cdot  \norm{v_i - v_j}^2.
$$
\end{lemma}

\subsubsection{Proof of Lemma~\ref{lem:pot}}

\begin{proof}

Let $S = \cup_{i=1}^t S^{(i)}$ and set $\beta \defeq \beta^{(t)}.$
By Lemma~\ref{lem:tracedecomp}, we have
$
\Tr(e^{-2\tau C^{(t)}}) - 1  \leq \nfrac{4}{3n}.
$
Hence,
$
\lambda_{n-1}(e^{-2\tau C^{(t)}}) \leq \nfrac{4}{3n},
$
which implies that, by taking logs, 
$$
\lambda_2(C^{(t)}) \geq \frac{\log n}{4\tau} \geq 3 \gamma.
$$
This can be rewritten in matrix terms, by Fact~\ref{fct:decomp} and Fact~\ref{fct:eigenvalue},and because $\supp(\beta) = S$ by the construction of \alg: 
\begin{equation}\label{eq:sdpcertificate}
L + \sum_{i \in S} \beta^{(t)}_i L(S_i) \succeq 3 \gamma \cdot L(K_V).
\end{equation}
which proves the first part of the Lemma.

For the second part, we start by noticing  that, for $i \in S,$  $\beta^{(t)}_i \leq 72\gamma \cdot  \nfrac{t}{T} \leq 72\gamma.$
Now for any $b$-balanced cut $U,$ with $\vol(U) \leq \vol(\bar{U}),$ consider the vector
$x_U$ defined as
$$
(x_U)_i \defeq  \left\{\!\!\!
\begin{array}{ll}
\sqrt{\frac{1}{2m} \cdot \frac{\vol(\bar{U})}{\vol(U)}} \;\; &\textrm{for} \; i \in U\\
\\
- \sqrt{\frac{1}{2m} \cdot \frac{\vol(U)}{\vol(\bar{U})}} \;\; &\textrm{for} \; i \in \bar{U}\\
\end{array}\right.
$$ 
Applying the guarantee of Equation~\ref{eq:sdpcertificate}, we obtain
\begin{align*}
x_U^\top L x_U + 72\gamma \cdot \sum_{i \in S} x_U^\top L(S_i) x_U \succeq 3 \gamma \cdot x_U^\top L(K_V) x_U.
\end{align*}
Notice that
\begin{align*}
x_U^\top L x_U &= \sum_{\{i,j\} \in E} ((x_U)_i - (x_U)_j)^2 = \frac{2m}{\vol(\bar{U})} \cdot \frac{\card{E(U, \bar{U})}}{\vol(U)} \leq 2 \cdot \phi(U),\\
x_U^\top L(S_i) x_U &= \sum_{j \in V} \frac{d_j}{2m} ((x_U)_i - (x_u)_j)^2) \leq \frac{2m}{\vol(U)} \cdot  \frac{d_i}{2m} = \frac{d_i}{\vol(U)},\\
x_U^\top L(K_V) x_U &= \sum_{i < j \in V} \frac{d_j d_i}{2m} ((x_U)_i - (x_u)_j)^2) = 1.
\end{align*}
Hence, our guarantee becomes
$$
2 \phi(U) + 72 \gamma \cdot \frac{\vol(S)}{\vol(U)} \geq 3 \gamma. 
$$
As $\vol(S) \leq \nfrac{b}{100} \cdot 2m \leq \nfrac{\vol(U)}{100},$ we have $\phi(U) \geq \gamma.$

\end{proof}

\subsubsection{Proof of Lemma~\ref{lem:obj}}
 
\begin{proof}
Consider
$
L \bullet D^{-1} P_{2\tau}(\beta{^{(t)}}) = L \bullet \Degin e^{-2\tau C^{(t)} } \Degin.
$
Using the cyclic property of trace and the definition of $C^{(t)},$ we have that
$$
L \bullet D^{-1} P_{2\tau}(\beta^{(t)}) \leq C^{(t)} \bullet e^{-2\tau C^{(t)} }.
$$
We now consider the spectrum of $C^{(t)}.$ By Fact~\ref{fct:eigenvalue}, the smallest eigenvalue is $0.$ 
Let the remaining eigenvalues be $\lambda_2 \leq \lambda_3 \leq \cdots \leq \lambda_n.$
Then, 
$
C^{(t)} \bullet e^{-2\tau C^{(t)}} = \sum_{i=2}^n \lambda_i e^{-2\tau \lambda_i}. 
$
We will analyze these eigenvalues in two groups. For the first group, we consider eigenvalues smaller than $24 \gamma$ and use Lemma~\ref{lem:tracedecomp}, together with the fact that $\gamma \geq \nfrac{1}{n^2}:$
$$
\sum_{i: \lambda_i \leq 24\gamma} \lambda_i e^{-2\tau \lambda_i} \geq \frac{1}{n^2} \cdot (\Tr(e^{-2\tau C^{(t)}}) - 1) = \frac{1}{n^3}.
$$
For the remaining eigenvalues, we have, by Lemma~\ref{lem:upbound}::
$$
\sum_{i: \lambda_i \geq 24 \gamma} \lambda_i e^{-2\tau \lambda_i} \leq  O(1)\cdot n \cdot e^{-2 \tau 24\gamma} \leq \frac{O(1)}{n^3}.  
$$
Combining these two parts, we have:
$$
L \bullet D^{-1} P_{2\tau}(\beta^{(t)}) \leq C^{(t)} \bullet e^{-2\tau C^{(t)}} \leq O(1) \cdot \sum_{i: \lambda_i \leq 24\gamma} \lambda_i e^{-2\tau \lambda_i} \leq O(\gamma) \cdot   L(K_V) \bullet D^{-1} P_{2\tau}(\beta^{(t)}).
$$ 
Now, we apply the Johnson-Lindenstrauss Lemma (Lemma~\ref{lem:jl}) to both sides of this inequality to obtain:
$$
L \bullet X^{(t)} \leq O(\gamma) \cdot L(K_V) \bullet X^{(t)} .
$$
\end{proof}

\subsubsection{Proof of Corollary~\ref{cor:unbalanced}}

\begin{proof}
By Lemma~\ref{lem:obj} and Theorem~\ref{thm:findcut}, we have that $S^{(t)}$ w.h.p. is either $c$-balanced or $\sum_{i \in S^{(t)}} d_i R_i \bullet X^{(t)} \geq \nfrac{2}{3} \cdot L(K_V) \bullet X^{(t)}.$ 
By Lemma~\ref{lem:jl} and as $\nfrac{1+\eps}{1-\eps} \leq \nfrac{4}{3},$ we have w.h.p.:
\begin{align*}
\Psi(P^{(t)}, S^{(t)}) = \sum_{i \in S^{(t)}} d_i R_i \bullet D^{-1} P_{2\tau}({\beta^{(t)}}) \geq \frac{1}{1+\eps} \cdot (\sum_{i \in S^{(t)}} d_i R_i \bullet X^{(t)}) \geq \frac{2}{3} \cdot L(K_V) \bullet X^{(t)} \geq \\
 \frac{2}{3} \cdot \frac{1-\eps}{1+\eps} \cdot L(K_V) \bullet D^{-1} P_{2\tau}({\beta^{(t)}}) = \frac{2}{3} \cdot \frac{1-\eps}{1+\eps} \cdot \Psi(P^{(t)},V) \geq \frac{1}{2} \cdot \Psi(P^{(t)}, V).
\end{align*}
\end{proof}

\subsubsection{Proof of Theorem~\ref{thm:reduction}}

\begin{proof}
By Lemma~\ref{lem:tracedecomp} and the Golden-Thompson inequality in Lemma~\ref{lem:gt}: 
$$
\Psi(P^{(t+1)}, V) = \Tr(e^{-2\tau C^{(t+1)}}) - 1 \leq \Tr\left(e^{-2\tau C^{(t)}} e^{-2\tau \Degin (\frac{\gamma}{T} \sum_{i \in S^{(t)}} L(S_i) ) \Degin}\right) - 1.
$$
We now apply Lemma~\ref{lem:mexptolin2} to the second term under trace.
To do this we notice that $\sum_{i \in S^{(t)}} L(S_i) \preceq 2 L(K_V) \preceq 2 D,$ so that
$$
\Degin \left(\frac{72\gamma}{T} \sum_{i \in S^{(t)}} L(S_i) \right) \Degin \preceq \frac{144\gamma}{T} I.
$$
Hence, we obtain 
$$
\Psi(P^{(t+1)}, V) \leq \Tr\left(e^{-2\tau C^{(t)}} \left(I - (1-e^{-288 \cdot \nfrac{\tau\gamma}{T}}) \cdot \frac{1}{2} \Degin ( \sum_{i \in S^{(t)}} L(S_i) ) \Degin\right) \right)-1.
$$
Applying the cyclic property of trace, we get
$$
\Psi(P^{(t+1)}, V) \leq  
\Psi(P^{(t)},V) - \frac{(1 -  e^{-288 \cdot \nfrac{ \tau \gamma}{T}})}{2} 
\sum_{i \in S^{(t)}} L(S_i)
\bullet D^{-1} P_{2\tau}({\beta^{(t)}}). 
$$
Next, we use Fact~\ref{fct:stardomination} to replace $L(S_i)$ by $R_i$ and notice that $288 \cdot \nfrac{ \tau \gamma}{T} = 2:$
$$
\Psi(P^{(t+1)}, V) \leq  
\Psi(P^{(t)},V) - \frac{(1 -  e^{-2})}{2} 
\sum_{i \in S^{(t)}} d_i R_i
\bullet D^{-1} P_{2\tau}({\beta^{(t)}}). 
$$
Then, we apply the definition of $\Psi(P^{(t)},S):$ 
$$
\Psi(P^{(t+1)}, V) \leq  \Psi(P^{(t)},V) - \nfrac{1}{3} \cdot \Psi(P^{(t)}, S).
$$
Finally, by Corollary~\ref{cor:unbalanced}, we know that w.h.p. $ \Psi(P^{(t)}, S^{(t)}) \geq \nfrac{1}{2} \cdot \Psi(P^{(t)},V)$ and the required result follows.
\end{proof}

\subsection{\SDP Interpretation}\label{sec:sdpinter}

OV designed an algorithm that outputs either a $\Omega(b)$-balanced cut of conductance $O(\sqrt{\gamma})$ or a certificate that no $b$-balanced cut of conductance $\gamma$ exists in time $\tilde{O}(\nfrac{m}{\gamma^2}).$
This algorithm uses the  \mwu of Arora and Kale~\cite{AK} to approximately solve an \SDP formulation of the \BS problem. The main technical contribution of their work is the routine {\sc FindCut} (implicit in their {\sc Oracle}), which takes the role of an approximate separation oracle for their \SDP.
In an iteration of their algorithm, OV use the \mwu update to produce a candidate \SDP-solution $Y^{(t)}.$ In one scenario, $Y^{(t)}$ does not have sufficiently low Laplacian objective value: 
\begin{equation}\label{eq:up1}
L \bullet Y^{(t)} \geq \Omega(\gamma) L(K_V) \bullet Y^{(t)}.
\end{equation}
In this case, the \mwu uses Equation~\ref{eq:up1} to produce a candidate solution $Y^{(t+1)}$ with lower objective value. 
Otherwise, {\sc FindCut} is run on the embedding corresponding to $Y^{(t)}.$ By Theorem~\ref{thm:findcut}, this yields either a cut of the required balance or a dual certificate that $Y^{(t)}$ is infeasible. This certificate has the form
\begin{equation}\label{eq:up2}
\gamma \cdot \sum_{i \in S^{(t)}} d_i R_i \bullet Y^{(t)} \geq \Omega(\gamma) L(K_V) \bullet Y^{(t)}
\end{equation}
and is used by the update to construct the next candidate $Y^{(t+1)}.$
The number of iterations necessary is determined by the width of the two possible updates described above. A simple calculation shows that the width of the update for Equation~\ref{eq:up1} is $\Theta(1),$ while for Equation~\ref{eq:up2}, it is only $O(\gamma).$
Hence, the overall width is $\Theta(1),$ implying that $O(\nfrac{\log n}{\gamma})$ iteration are necessary for the algorithm of OV to produce a dual certificate that the \SDP is infeasible and therefore no $b$-balanced cut of conductance $\gamma$ exists.

Our modification of the update is based on changing the starting candidate solutions from $Y^{(1)} \propto D^{-1}$ to $X^{(1)} \propto \Degin e^{-2\tau \Degin L \Degin} \Degin.$
In Lemma~\ref{lem:pot} and Lemma~\ref{lem:obj}, we show that this modification implies that all $X^{(t)}$ must now have 
$
L \bullet X^{(t)} \leq O(\gamma) \cdot L(K_V) \bullet X^{(t)}
$
or else we find a dual certificate that the \SDP is infeasible.
This additional guarantee effectively allows us to bypass the update of Equation~\ref{eq:up1} and only work with updates of the form given in Equation~\ref{eq:up2}. As a result, our width is now $O(\gamma)$ and we only require $O(\log n)$ iterations.

Another way to interpret our result is that all possible $\tau \approxeq \nfrac{\log n}{\gamma}$ updates of the form of Equation~\ref{eq:up1} in the algorithm of OV are regrouped into a single step, which is performed at the beginning of the algorithm.

\subsection{The {\sc FindCut} Subroutine}\label{sec:findcut}

Most of the material in this Section appears in \cite{OV} or in~\cite{Lthesis}. We reproduce it here in the language of this paper for completeness. The constants in these proofs are not optimized.

\subsubsection{Preliminaries}

\begin{fact} \label{fct:star} 
For a subset $S \subseteq V,$ 
$$
\sum_{i \in \bar{S}} d_i R_i \succeq \frac{\vol(S)}{2m}\left( L(K_V) - L(K_S)\right).
$$
\end{fact}
\begin{proof}
By Fact~\ref{fct:stardomination},
$$
\sum_{i \in \bar{S}} d_i R_i = \sum_{i \in \bar{S}} L(S_i) - \frac{\vol(\bar{S})}{2m} L(K_V).
$$
Moreover, by the definitions it is clear that
$$
\sum_{i \in \bar{S}} L(S_i) + \frac{\vol(S)}{2m} L(K_S) \succeq L(K_V).
$$
Combining these two equations, we obtain the required statement.
\end{proof}

The following is a variant of the sweep cut argument of Cheeger's inequality~\cite{Fan}, tailored to ensure that a constant fraction of the variance of the embedding is contained inside the output cut.
\begin{lemma} \label{lem:cheeger}
Let $x \in \mathbb{R}^n, x \geq 0,$ such that  $x^\top L x \leq \lambda$ and $\vol(\supp(x)) \leq \nfrac{2m}{2}.$ Relabel the vertices so that $x_1 \geq x_2 \geq \cdots \geq x_{z-1} > 0$ and $x_{z} = \cdots = x_n = 0.$ For $i \in [z-1],$ denote by $S_i \subseteq V,$ the sweep cut $\{1, 2, \ldots, i\}.$ 
Further, assume that $\sum_{i=1}^n d_i x_i^2 \leq 1,$ and, for some fixed $k \in [z-1],$ $\sum_{i=k}^{n} d_i x_i^2 \geq \sigma.$
Then, there is a sweep cut $S_h$ of $x$ such that $ z-1 \geq h \geq k$ and $\phi(S_h) \leq \nfrac{1}{\sigma} \cdot \sqrt{ 2 \lambda}.$
\end{lemma}

\noindent
We will also need the following simple fact.
\begin{fact} \label{fct:triangle}
Given $v,u, t \in \mathbb{R}^h,$ $\left(\norm{v-t} - \norm{u-t}\right)^2 \leq \norm{v - u}^2.$
\end{fact}

\subsubsection{Roundable Embeddings and Projections}

The following definition of {\it roundable embedding} captures the case in which a vector embedding of the vertices $V$ highlights a balanced cut of conductance close to $\alpha$ in $G.$
Intuitively, in a roundable embedding, a constant fraction of the total variance is spread over a large set $R$ of vertices. 
\begin{definition}[Roundable Embedding]\label{def:roundable}
Given an embedding $\{v_i\}_{i \in V}$ with Gram matrix $X,$ denote by $\Psi$ the total variance of the embedding:
$ \Psi \defeq L(K_V) \bullet X.$
Also, let 
$R = \{i \in V : \norm{v_i - v_\avg}^2 \leq 32 \cdot \nfrac{(1-b)}{b} \cdot \frac{\Psi}{2m}\}.$ 
For $\alpha > 0,$ we say that $\{v_i\}_{i \in V}$ is roundable  for $(G, b, \alpha)$ if:
\begin{itemize}
\item $L \bullet X \leq \alpha \Psi,$
\item $L(K_R) \bullet X \geq 
\frac{\Psi}{128}.$
\end{itemize} 
\end{definition}

\noindent
A roundable embedding can be converted into a balanced cut of conductance $O(\sqrt{\alpha})$ by using a standard projection rounding, which is a simple extension of an argument already appearing in~\cite{ARV} and~\cite{AK}. The rounding procedure {\sc ProjRound} is described in Figure~\ref{fig:rounding} for completeness. It is analyzed in~\cite{OV} and~\cite{Lthesis}, where the following theorem is proved.
\begin{theorem}[Rounding Roundable Embeddings]\cite{OV, Lthesis} \label{thm:stdround}
If $\{v_i \in \mathbb{R}^h \}_{i \in V}$ is roundable for $(G,b,\alpha)$, then {\sc ProjRound}$(\{v_i\}_{i \in V}, b)$ produces a $\Omega(b)$- balanced cut of conductance $O\left(\sqrt{\alpha}\right)$ with high probability in time $\tilde{O}(n h + m).$
\end{theorem}

\begin{figure*}[h]
  \begin{tabularx}{\textwidth}{|X|}
    \hline
  \begin{enumerate}[label=\arabic*.]
  \item {\bf {Input:}} An embedding $\{v_i \in \mathbb{R}^h \}_{i \in V} ,$ $b\in (0,\nfrac{1}{2}].$
  \item Let $c = \Omega(b) \leq \nfrac{b}{100}$ be a constant, fixed in the proof of Theorem~\ref{thm:stdround} in~\cite{OV}. 
\item For $t=1,2, \ldots, O(\log n)$:
\begin{enumerate}[label=\alph*.]
\item Pick a unit vector $u$ uniformly at random from $\mathbb{S}^{h-1}$ and let $x
\in \mathbb{R}^n$ with $x_i \defeq \sqrt{h} \cdot  { u^\top v_i}$.
\item Sort the vector $x.$ Assume w.l.og. that $x_1 \geq x_2 \geq \cdots \geq x_n$. Define $S_i \defeq \{j\in [n]:x_j \geq x_i\}$.
   
  \item  Let $S^{(t)} \defeq (S_i,\bar{S_i})$ which minimizes
  $ \phi(S_i)$ among sweep-cuts for which  $\vol(S_i) \in [c\cdot 2m, (1-c)\cdot 2m].$
\end{enumerate}
\item {\bf {Output:}} The cut $S^{(t)}$ of least conductance over all choices of $t.$ 
  \end{enumerate}\\
   \hline
    \end{tabularx}

  \caption{{\sc ProjRound}}
  \label{fig:rounding}
\end{figure*}

\subsubsection{Description of {\sc FindCut}}

In this subsection we describe the subroutine {\sc FindCut} and prove Theorem~\ref{thm:findcut}.

\begin{figure*}[htb]
  	\begin{tabularx}{\textwidth}{|X|}
    \hline
  	\begin{enumerate}[label=\arabic*.]
\item {\bf {Input:}} Instance graph $G,$ balance $b$, conductance value $\alpha$ and embedding $\{v_i\}_{i \in V},$ with Gram matrix $X.$ 

\item Let $r_i = \norm{v_i - v_\avg}$ for all $i \in V.$
Denote $\Psi \defeq L(K_V) \bullet X$ and define the set $R \defeq \{i \in V : r_i^2 \leq 32 \cdot \nfrac{(1-b)}{b} \cdot \nfrac{\Psi}{2m}\}.$
\item {\sc Case 1}: If $L \bullet X > \alpha \Psi,$ output {\sc FAIL} and terminate.
\item {\sc Case 2}: If $L(K_R) \bullet X \geq \nfrac{\Psi}{128},$ the embedding $\{v_i\}_{i \in V}$ is roundable for $(G, b, \alpha).$ Run {\sc ProjRound}, output the resulting cut and terminate.

\item \label{stp:cut} {\sc Case 3}: Relabel the vertices of $V$ such that $r_1 \geq r_2 \geq \cdots \geq r_n$ and let $S_i=\{1, \ldots,i\}$ be the $j^\text{th}$ sweep cut of $r.$ 
Let $z$ the smallest index such that $\vol(S_z) \geq \nfrac{b}{4} \cdot 2m.$ 	Output the most balanced sweep cut $C$ among $\{S_1, \ldots, S_{z-1}\},$ such that $\phi(C) \leq 40 \cdot \sqrt{\gamma}.$  
 \end{enumerate}\\
   \hline
    \end{tabularx}
  \caption{{\sc FindCut}}
  \label{fig:oracle}
\end{figure*}

\begin{theorem}[Theorem~\ref{thm:findcut} Restated]
Consider an embedding $\{v_i \in \R^d\}_{i \in V}$ with Gram matrix $X$ such that $L \bullet X^{(t)} \leq \alpha L(K_V) \bullet X^{(t)},$ for $\alpha > 0.$ On input $(G,b,\alpha, \{v_i\}_{i \in V}),$ {\sc FindCut} runs in time $\tilde{O}(md)$ and w.h.p. outputs a cut $C$ with $\phi(C) \leq O(\sqrt{\alpha}).$ Moreover, there is a constant $c = \Omega(b) \leq \nfrac{b}{100}$ such that either $C$ is $c$-balanced or 
$$
\sum_{i \in C} d_i R_i \bullet X \geq \nfrac{2}{3} \cdot L(K_V) \bullet X. 
$$
\end{theorem}

\begin{proof}
By Markov's inequality,  $\vol(\bar{R}) \leq \nfrac{b}{(32 \cdot(1-b))} \cdot 2m \leq \nfrac{b}{16} \cdot 2m \leq \nfrac{1}{32} \cdot 2m.$ 
By assumption, {\sc Case 1} cannot take place. If {\sc Case 2} holds, then the embedding is roundable: by Theorem~\ref{thm:stdround}, {\sc ProjCut} outputs an $\Omega(b)$-balanced cut $C$ with conductance $O(\sqrt{\alpha}).$ If this is not the case, we are in {\sc Case 3.}

We then have $L(K_R) \leq \nfrac{\Psi}{128}$ and, by Fact \ref{fct:star}:
\begin{align*}
\sum_{i \in \bar{R}} d_i R_i \bullet X =  
\sum_{i \in \bar{R}} d_i r_i^2
\geq \frac{\vol(R)}{2m} \cdot \left(1 - \frac{1}{128}\right) \cdot \Psi \geq \left(1 - \frac{1}{32}\right) \cdot \left(1 - \frac{1}{128}\right) \cdot \Psi \\
\geq \left(1 - \frac{5}{128}\right) \cdot \Psi.
\end{align*}
It must be the case that $\bar{R} = S_g$ for some $g \in [n],$ with $g \leq z$ as $\vol(S_g) \leq \vol(S_z).$
Let $k \leq z$ be the the vertex in $\overline{R}$ such that $\sum_{j=1}^{k} d_j r_j^2 \geq  \nfrac{3}{4} \cdot (1 - \nfrac{5}{128})$ and  $\sum_{j=k}^{g} d_j r_j^2 \geq \nfrac{1}{4} \cdot (1 - \nfrac{5}{128}).$ 
By the definition of $z,$ we have $k \leq g < z$ and $r_z^2 \leq \nfrac{4} {b} \cdot \nfrac{\Psi}{2m} \leq 8 \cdot \nfrac{(1-b)}{b} \cdot \nfrac{\Psi}{2m}.$ Hence, we have $r_z \leq \nfrac{1}{2} \cdot r_i,$ for all $i \geq g.$
Define the vector $x$ as $x_i \defeq (r_i - r_z)$ for $i \in S_z$ and $r_i \defeq 0$ for $i \notin S_z.$ 
Notice that:
\begin{align*}
x^\top L x = \sum_{\{i,j\} \in E} (x_i - x_j)^2 
\leq \sum_{\{i,j\} \in E} (r_i - r_j)^2 \\
\stackrel{\rm Fact \; \ref{fct:triangle}}{\leq} 
\sum_{\{i,j\} \in E} \norm{v_i - v_j}^2 \leq  \alpha \Psi.
\end{align*}
Also, $x \geq 0$ and $\vol(\supp(x)) \leq \nfrac{b}{4} \cdot 2m \leq \nfrac{2m}{2},$ by the definition of $z.$ Moreover,
$$
\sum_{i=1}^n d_i x_i^2 = \sum_{i=1}^z d_i (r_i - r_z)^2 \leq \sum_{i=1}^z d_i r_i^2 \leq  \Psi, 
$$ 
and
\begin{align*}
\sum_{i=k}^n d_i x_i^2 = \sum_{i=k}^z d_i (r_i -r_z)^2 \\
\geq  \sum_{i=k}^g d_i (r_i - \nfrac{1}{2} \cdot r_i)^2 
\\= \nfrac{1}{4} \cdot \sum_{i=k}^g   d_i r_i^2  
\\ 
\geq \nfrac{1}{16} \cdot (1 - \nfrac{5}{128}) \cdot \Psi\geq \nfrac{1}{20} \cdot \Psi.
\end{align*}
Hence we can now apply Lemma~\ref{lem:cheeger} to the vector $\nfrac{1}{\Psi} \cdot x.$
This shows that there exists a sweep cut $S_h$ with $ z > h \geq k,$ such that $\phi(S_h) \leq 40 \cdot \sqrt{\gamma} .$ It also shows that $C,$ as defined in Figure \ref{fig:oracle}, must exist. Moreover, it must be the case that $S_k \subseteq S_h \subseteq C.$ As $h \geq k,$ we have
$$
\sum_{i \in C} d_i R_i \bullet X = \sum_{i \in C} d_i r_i^2 \geq \sum_{i=1}^k d_i r_i^2 \geq   \frac{3}{4} \cdot \left(1-\frac{5}{128}\right) \cdot \Psi \geq \frac{2}{3} \cdot \Psi =
\frac{2}{3} \cdot L(K_V) \bullet X.
$$
Finally,  using the fact that $\{v_i\}_{i \in V}$ is embedded in $d$ dimensions, we can compute $L \bullet \tilde{X}$ in time $O(dm).$ 
Moreover, $L(K_V) \bullet X$ can be computed in time $O(nd)$ by using the decomposition  
$L(K_V) \bullet X = \sum_{i \in V} d_i \norm{v_i - v_{\avg}}^2.$
By the same argument, we can compute $L(K_R) \bullet X$ in time $O(nd).$
The sweep cut over $r$ takes time $\tilde{O}(m).$
And, by Theorem~\ref{thm:stdround}, {\sc ProjRound} runs in time $\tilde{O}(md).$
 Hence, the total running time is $\tilde{O}(md).$ 
\end{proof}

\section{Computing $\exp(-A)v$}
\label{sec:exp}
In this section, we describe procedures for approximating $\exp(-A)v$
up to an $\ell_2$ error of $\delta\norm{v}$, given a symmetric PSD
matrix $A$ and a vector $v$ (\emph{w.l.o.g.,} $\norm{v}=1$). In
particular, we give the required procedures and proofs for
Theorems~\ref{thm:expRational}, Theorem~\ref{thm:expPoly} and
Theorem~\ref{thm:expRational2}. For this section, we will assume the
upper bound from Theorem~\ref{thm:exp-poly-approx:restated} (which is
a more precise version of Theorem~\ref{thm:exp-poly-approx}),
regarding polynomials approximating $e^{-x}$. Discussion about this
theorem and the proofs are included in Section~\ref{sec:poly}. We
restate the basic definitions used in this section for completeness.
\subsubsection*{\it Definitions.} We will always work with square $n
\times n$ matrices over $\mathbb{R}.$ For a matrix $M,$ abusing
notation, we will denote its exponential by $\exp(-M)$ which is
defined as $\sum_{i\geq 0} \frac{(-1)^i}{i!}M^{i}.$ $\norm{M}\defeq
\sup _{\|x\|=1} \norm{Mx}$ denotes the spectral norm of $M.$ $M$ is
said to be \emph{Symmetric and Diagonally Dominant} (SDD) if, $M_{ij}
= M_{ji},$ for all $i,j$ and $M_{ii} \ge \sum_j |M_{ij}|,$ for all
$i$.  $M$ is called \emph{Upper Hessenberg} if, $(M)_{ij} = 0$ for $i
> j+1.$ $M$ is called \emph{tridiagonal} if $M_{ij} = 0$ for $i > j+1$
and for $j > i+1.$ Let $\Lambda(M)$ denote the spectrum of a matrix
$M$ and let $\lambda_1(M)$ and $\lambda_n(M)$ denote the largest and the
smallest eigenvalues of $M$ respectively. For a matrix $M,$ let $m_M$
denote the number of non-zero entries in $M.$ Further, let $t_M$
denote the time required to multiply the matrix $M$ with a given
vector $w.$ In general $t_M$ depends on how $M$ is given as an input
and can be $\Theta(n^2)$. However, it is possible to exploit the
special structure of $M$ if given as an input appropriately: It is
possible to just multiply the non-zero entries of $M,$ giving $t_M =
O(m_M).$ Also, if $M$ is a rank one matrix $ww^\top,$ where $w$ is
known, we can multiply with $M$ in $O(n)$ time.  For any positive
integer $k,$ let $\Sigma_k$ denote the set of all polynomials with
degree at most $k$. Given a degree $k$ polynomial $p \defeq \sum_{i=0}^k a_i
\cdot x^i,$ the $\ell_1$ norm of $p,$ denoted as $\norm{p}_1$ is
defined as $\norm{p}_1 = \sum_{i \ge 0}^k |a_i|.$
\subsubsection*{\it Algorithms for Theorem~\ref{thm:expRational},
  \ref{thm:expRational-psd} and \ref{thm:expRational2}.}
Theorem~\ref{thm:expRational}, \ref{thm:expRational-psd} and
\ref{thm:expRational2} are based on a common algorithm we describe,
called {\expRational} (see Figure~\ref{fig:expRational}), which
requires a procedure $\Invert_A$ with the following guarantee: given a
vector $y,$ a positive integer $k$ and $\varepsilon_1 > 0,$
$\Invert_A(y,k,\varepsilon_1)$ returns a vector $u_1$ such that,
$\norm{(I+\nfrac{A}{k})^{-1}y - u_1} \le \varepsilon_1 \norm{y}.$ The
algorithms for the two theorems differ only in their implementation of
$\Invert_A.$ We prove the following theorem about {\expRational}.
\begin{theorem}[Running Time of {\expRational} given $\Invert_A$]
\label{thm:expRational:params}
Given a symmetric p.s.d. matrix $A \succeq 0$, a vector $v$ with
$\norm{v}=1,$ an error parameter $0< \delta \le 1$ and oracle access to
$\Invert_A,$ for parameters $k \defeq O(\log \nfrac{1}{\delta})$ and
$\epsilon_1 \defeq \exp(-\Theta(k \log k + \log (1+\norm{A}))),$ {\expRational}
computes a vector $u$ such that $\norm{\exp(-A)v-u} \le \delta$, in
time $O(T^{\text{inv}}_{A,k,\varepsilon_1}\cdot k + n\cdot k^2 +
k^3),$ where $T^{\text{inv}}_{A,k,\varepsilon_1}$ is the time
required by $\Invert_A(\cdot,k,\varepsilon_1).$
\end{theorem}
\noindent The proof of this theorem appears in Section
\ref{sec:error}. Theorem~\ref{thm:expRational} will follow from the
above theorem by using the Spielman-Teng SDD solver to implement the
$\Invert_A$ procedure (See Section~\ref{sec:expRational:proof}). For
Theorem~\ref{thm:expRational2}, we combine the SDD solver with the
Sherman-Morrison formula (for matrix inverse with rank 1 updates) to
implement the $\Invert_A$ procedure (See
Section~\ref{sec:expRational2:proof}).

\subsubsection*{\it Algorithm for Theorem~\ref{thm:expPoly}.} The procedure and
proof for Theorem~\ref{thm:expPoly} is based on the well-known Lanczos
method.  We give a description of the Lanczos method (\emph{e.g.}
see~\cite{Saad}) in Figure~\ref{fig:lanczos} and give a proof of a
well known theorem about the method that permits us to extend
polynomial approximations for a function $f$ over reals to
approximating $f$ over matrices (Theorem~\ref{thm:lanczos}).
Combining our result on polynomials approximating $e^{-x}$ from the
upper bound in Theorem \ref{thm:exp-poly-approx:restated} with the theorem
about the Lanczos method, we give a proof of the following theorem
that immediately implies Theorem~\ref{thm:expPoly}.
\begin{theorem}[Running Time Using {\lanczos}]
\label{thm:expPoly:params}
Given a symmetric p.s.d. matrix $A$, a vector $v$ with $\norm{v}=1$ and a
parameter $0< \delta\le 1$, for  
\[k \defeq O\left(\sqrt{\max\{\log^2 \nfrac{1}{\delta},
    (\lambda_1(A)-\lambda_n(A)) \cdot \log \nfrac{1}{\delta}\}}\cdot
  \left(\log \nicefrac{1}{\delta}\right) \cdot \log \log
  \nicefrac{1}{\delta}\right),\]
 and $f(x) = e^{-x},$ the procedure
{\lanczos} computes a vector $u$ such that $\norm{\exp(-A)v-u} \le
\norm{\exp(-A)}\delta.$ The time
taken by {\lanczos} is $O\left((n+t_A)k +k^2\right)$.
\end{theorem}

\begin{remark}
  Note the $k^3$ term in the running time for
  Theorem~\ref{thm:expRational:params} and the $k^2$ term in the
  running time for Theorem~\ref{thm:expPoly:params}. This is the time
  required for computing the eigendecomposition of a $(k+1)\times
  (k+1)$ symmetric matrix. While this process requires $O(k^3)$ time
  in general, as in Theorem~\ref{thm:expRational:params}; in case of
  Theorem~\ref{thm:expPoly:params}, the matrix is tridiagonal and
  hence the time required is $O(k^2)$ (see~\cite{eigendecomp}).
\end{remark}

\subsubsection*{\it Organization.} We first describe the Lanczos method
and prove some of its properties in Section~\ref{sec:lanczos}. Then,
we give descriptions of the {\lanczos} and the {\expRational}
procedures in Section~\ref{sec:exp-procedures}. Assuming
Theorem~\ref{thm:expRational:params}, we give proofs of
Theorem~\ref{thm:expRational} and Theorem~\ref{thm:expRational2} in
Section~\ref{sec:expRational:proof} and
Section~\ref{sec:expRational2:proof} respectively by implementing the
respective $\Invert_A$ procedures. Finally, we give the error analysis
for {\expRational} and a proof for
Theorem~\ref{thm:expRational:params} in Section~\ref{sec:error}.

\subsection{Lanczos Method -- From Scalars to Matrices}
\label{sec:lanczos}
Suppose one is given a symmetric PSD matrix $B,$ and a function
$f:\mathbb{R} \mapsto \mathbb{R}.$ Then one can define $f(B)$ as
follows: Let $u_1,\ldots,u_n$ be eigenvectors of $B$ with eigenvalues
$\lambda_1,\ldots,\lambda_n.$ Define $f(B) \defeq \sum_{i}
f(\lambda_i)u_iu_i^\top.$ Given a vector $v,$ we wish to compute
$f(B)v.$ Since exact computation of $f(B)$ usually requires
diagonalization of $B,$ which is costly, we seek an approximation to
$f(B)v$. 

For a given positive integer $k,$ the Lanczos method looks for an
approximation to $f(B)v$ of the form $p(B)v,$ where $p$ is a
polynomial of degree $k.$ Note that for any polynomial $p$ of degree
at most $k,$ the vector $p(B)v$ is a linear combination of the vectors
$\{v,Bv,\ldots,B^kv\}$.  The span of these vectors is referred to as
the {\it Krylov Subspace} and is defined below.
\begin{definition}[Krylov Subspace]
  Given a matrix $B$ and a vector $v$, the Krylov subspace of order
  $k$, denoted by $\calK(B,v,k)$, is defined as the subspace that is
  spanned by the vectors $\{v,Bv,\ldots,B^kv\}$.
\end{definition}

\noindent Note that any vector in $\calK(B,v,k)$ has to be of the form
$p(B)v$, where $p$ is some degree $k$ polynomial.  The Lanczos method
starts by generating an orthonormal basis for $\calK(B,v,k)$. Let
$v_0,\ldots,v_k$ be any orthonormal basis for $\calK(B,v,k),$ and let
$V_k$ be the $n \times (k+1)$ matrix with $\{v_i\}_{i=0}^k$ as its
columns. Thus, $V_k^\top V_k = I_k$ and $V_kV_k^\top $ denotes the
projection onto the subspace. Also, let $T_k$ be the operator $B$ in
the basis $\{v_i\}_{i =0}^k,$ restricted to this subspace,
\emph{i.e.}, $T_k \defeq V_k^\top B V_k.$ Since, all the vectors
$v,Bv,\ldots,B^kv$ are in the subspace, any of these vectors (or a
linear combination of them) can be obtained by applying $T_k$ to $v$
(after a change of basis), instead of $B$. The following
lemma proves this formally.
\begin{lemma}[Exact Computation with Polynomials. See \emph{e.g.}~\cite{Saad}]
\label{lem:exact-power}
Let $V_k$ be the orthonormal basis, and $T_k$ be the operator $B$
restricted to $\calK(B,v,k)$ where $\norm{v}=1$, \emph{i.e.}, $T_k
= V_k^\top B V_k$. Let $p$ be a polynomial of degree at most $k$. Then,
\[p(B)v = V_k p(T_k) V_k^\top  v.\]
\end{lemma}
\begin{proof}
  Recall that $V_kV_k^\top$ is the orthogonal projection onto the
  subspace $\calK(B,v,k)$. By linearity, it suffices to prove this
  when $p$ is $x^t$ for $t \le k$. This is true for $t=0$ since $V_k
  V_k^\top v = v $. For any $j \le k$, $B^j v $ lies in
  $\calK(B,v ,k),$ thus, $\forall\ j\le k,\ V_k V_k^\top B^j v  =
  B^j v .$ Hence,
\begin{align*}
B^t v  & = (V_kV_k^\top )B(V_kV_k^\top )B\cdots B(V_kV_k^\top )v  \\
& = V_k(V_k^\top BV_k)(V_k^\top BV_k)\cdots (V_k^\top BV_k)V_k^\top v  = V_kT_k^t V_k^\top v 
\end{align*}
\end{proof}

The following lemma shows that $V_kf(T_k)V_k^\top v $ approximates
$f(B)v $ as well as the \emph{best} degree $k$ polynomial that
uniformly approximates $f$. The proof is based on the observation that
if we express $f$ as a sum of \emph{any} degree $k$ polynomial and an
\emph{error} function, the above lemma shows that the polynomial part
is exactly computed in this approximation. 
\begin{lemma}[Approximation by Best Polynomial (Lemma 4.1, \cite{Saad})]
\label{lem:exact-poly}
Let $V_k$ be the orthonormal basis, and $T_k$ be the operator $B$
restricted to $\calK(B,v ,k)$ where $\norm{v }=1$, \emph{i.e.}, $T_k
= V_k^\top B V_k$. Let $f\from \rea \to \rea$ be any
function such that $f(B)$ and $f(T_k)$ are well-defined. Then,
\[\norm{f(B)v  - V_kf(T_k)V_k^\top v } \le \min_{p_k \in \Sigma_k}
\left(\max_{\lambda \in \Lambda(B)} |f
  (\lambda)-p_k(\lambda)|+\max_{\lambda \in \Lambda(T_k)} |f
  (\lambda)-p_k(\lambda)|\right)\ .\]
\end{lemma}
\begin{proof}
Let $p_k$ be any degree $k$ polynomial. Let $r_k \defeq f-p_k$. Then,
\begin{align*}
\norm{f(B)v  - V_kf(T_k) V_k^\top v } & \le \norm{p_k(B)v  - V_kp_k(T_k) V_k^\top v  } + \norm{r_k(B)v  - V_kr_k(T_k) V_k^\top v } \\
& \le 0 + \norm{r_k(B)} + \norm{V_kr_k(T_k) V_k^\top }  \qquad \qquad
\qquad \text{(Using Lemma~\ref{lem:exact-power})}\\
& = \max_{\lambda \in \Lambda(B)} |r_k(\lambda)|+\max_{\lambda \in \Lambda(T_k)} |r_k(\lambda)|.
\end{align*}
Minimizing over $p_k$ gives us our lemma.
\end{proof}

\noindent Observe that in order to compute this approximation, we do
not need to know the polynomial explicitly. It suffices to prove that
there exists a degree $k$ polynomial that uniformly approximates $f$
well on an interval containing the spectrum of $B$ and $T_k$ (For
exact computation, $\Lambda(T_k) \subseteq \Lambda(B)$.) Moreover, if
$k \ll n$, the computation has been reduced to a much smaller
matrix. We now show that an orthonormal basis for the Krylov
Subspace, $V_k,$ can be computed quickly and then describe the {\lanczos}
procedure.

\subsubsection{Efficiently Computing a Basis for the Krylov Subspace}
In this section, we show that if we construct the basis $\{v_i\}_{i=0}^k$ in a
particular way, the matrix $T_k$ has extra structure. In particular,
if $B$ is symmetric, we show that $T_k$ must be
\emph{tridiagonal}. This will help us
speed up the construction of the basis.

Suppose we compute the orthonormal basis $\{v_i\}_{i=0}^k$
iteratively, starting from $v_0=v$: For $i=0,\ldots,k$, we compute
$Bv_{i}$ and remove the components along the vectors
$\{v_0,\ldots,v_{i}\}$ to obtain a new vector that is orthogonal to
the previous vectors. This vector, scaled to norm 1, is defined to be
$v_{i+1}.$ These vectors, by construction, satisfy that for all $i \le
k,$ $\Span\{v_0,\ldots,v_i\} = \Span\{v,Bv,\ldots,B^kv\}.$ Note that
$(T_k)_{ij} = v_i^\top B v_j.$

If we construct the basis iteratively as above, $Bv_j \in
\Span\{v_0,\ldots,v_{j+1}\}$ by construction, and if $i > j+1,$ $v_i$
is orthogonal to this subspace and hence $v_i^\top (Bv_j)=0$. Thus,
$T_k$ is \emph{Upper Hessenberg}, \emph{i.e.}, $(T_k)_{ij} = 0$ for $i
> j+1$.

Moreover, if $B$ is symmetric, $v_j^\top (Bv_i) = v_i^\top (Bv_j),$
and hence $T_k$ is symmetric and tridiagonal. This means that at most
three coefficients are non-zero in each row. Thus, while constructing
the basis, at step $i+1$, it needs to orthonormalize $Bv_i$ only
w.r.t.  $v_{i-1}$ and $v_i$. This fact is used for efficient
computation of $T_k.$ The algorithm {\lanczos} appears in Figure
\ref{fig:lanczos} and the following meta-theorem summarizes the main
result regarding this method.

\begin{figure*}
\begin{tabularx}{\textwidth}{|X|}
\hline
\vspace{0mm}
{\bf Input}: A symmetric matrix $B \succeq 0$, a vector $v$ such that
$\norm{v}=1,$ a positive integer $k,$ and a function $f \from \rea \to \rea$. \\
{\bf Output}: A vector $u$ that is an approximation to $f(B)v$.
\begin{enumerate}[label=\arabic*.]
      \item Initialize $v_0 \defeq v.$
      \item For $i = 0$ to $k-1,$\hfill {\footnotesize (Construct an orthonormal basis to Krylov subspace of order $k$}
        \begin{enumerate}[label=\alph*.]
      \item If $i = 0$, compute $w_0 \defeq Bv_0$. Else, compute $w_i
= Bv_i - \beta_i v_{i-1}.$ \label{item:expPoly:2} \hfill {\footnotesize
  (Orthogonalize w.r.t. $v_{i-1}$)}
     \item Define $\alpha_i \defeq v_i^\top w_i$ and $w_i^\prime
\defeq w_i - \alpha_iv_i\ ^*$. \hfill {\footnotesize (Orthogonalize
  w.r.t. $v_i$)}
     \item Define $\beta_{i+1} \defeq \norm{w_i^\prime}$ and $v_{i+1}
\defeq w_i^\prime/\beta_{i+1}.$ \hfill {\footnotesize (Scaling it to norm 1)}
        \end{enumerate}
      \item Let $V_k$ be the $n\times (k+1)$ matrix whose columns are
        $v_0,\ldots,v_k$ respectively.

      \item Let $T_k$ be the $(k+1)\times(k+1)$ matrix such that for
        all $i$,
        $(T_k)_{ii} = v_i^\top B v_i = \alpha_i, (T_k)_{i,i+1} =
        (T_k)_{i+1,i} = v_{i+1}^\top B v_{i} =
        \beta_{i+1}$ and all other entries are 0. \hfill
        {\footnotesize (Compute $T_k \defeq V_k^\top B V_k $)}
     \item Compute ${\mathcal{B}} \defeq f\left(T_k\right)$
        exactly via eigendecomposition. Output the vector
        $V_k\mathcal{B}V_k^\top v$.
\vspace{0mm}
{\small \item[*] If $w_i^\prime = 0$, compute the approximation with the matrices $T_{i-1}$ and $V_{i-1},$ instead of $T_k$ and $V_k$. The error bounds still hold. }
\vspace{-3mm}
\end{enumerate}
\\ 
\hline
\end{tabularx}
  \caption{The {\lanczos} algorithm for approximating $f(B)v$}
  \label{fig:lanczos}
\end{figure*}

\begin{theorem}[{\lanczos} Theorem]
\label{thm:lanczos}
Given a symmetric p.s.d. matrix $B$, a vector $v$ with $\norm{v}=1,$ a
function $f$ and a positive integer parameter $k$ as inputs, the procedure
{\lanczos} computes a vector $u$ such that,
\[\norm{f(B)v-u} \le 2\cdot \min_{p_k \in \Sigma_k} \max_{\lambda \in
  \Lambda(B)} |f (\lambda)-p_k(\lambda)|\ .\]
Here $\Sigma_k$ denotes the set of all degree $k$ polynomials and
$\Lambda(B)$ denotes the spectrum of $B$.  The time taken by {\lanczos}
is $O\left((n+t_B)k +k^2\right).$
\end{theorem}
\begin{proof}
  The algorithm {\lanczos} implements the Lanczos method we've
  discussed here. The guarantee on $u$ follows from
  Lemma~\ref{lem:exact-poly} and the fact that $\Lambda(T_k) \subseteq
  \Lambda(B)$. We use the fact that $(T_k)_{ij} = v_i^\top B v_j$ and
  that $T_k$ must be \emph{tridiagonal} to reduce our work to just
  computing $O(k)$ entries in $T_k.$ The total running time is
  dominated by $k$ multiplications of $B$ with a vector, $O(k)$
  dot-products and the eigendecomposition of the tridiagonal matrix
  $T_k$ to compute $f(T_k)$ (which can be done in $O(k^2)$ time
  \cite{eigendecomp}), giving a total running time of $O\left((n+t_B)k
    +k^2\right).$
\end{proof}

\subsection{Procedures for Approximating $\exp(-A)v$.}
\label{sec:exp-procedures}
Having introduced the Lanczos method, we describe the algorithms we use
for approximating the matrix exponential.
\subsubsection{Using {\lanczos} for Approximating $\exp(-A)v$ -- Proof of
  Theorem ~\ref{thm:expPoly}}
Theorem~\ref{thm:expPoly} follows from
Theorem~\ref{thm:expPoly:params}, which is proved by combining
Theorem~\ref{thm:lanczos} about the approximation guarantee of the
{\lanczos} algorithm and Theorem~\ref{thm:exp-poly-approx:restated} (a
more precise version of Theorem~\ref{thm:exp-poly-approx}) about
polynomials approximating $e^{-x}.$ We now give a proof of Theorem~\ref{thm:expPoly:params}.
\begin{proof}
  We are given a matrix $A,$ a unit vector $v$ and an error parameter
  $\delta$. Let $p_{\lambda_n(A),\lambda_1(A),\nfrac{\delta}{2}}(x)$
  be the polynomial given by
  Theorem~\ref{thm:exp-poly-approx:restated} and let $k$ be its
  degree. We know from the theorem that
  $p_{\lambda_n(A),\lambda_1(A),\nfrac{\delta}{2}}(x)$ satisfies
  $\sup_{x \in [\lambda_n(A),\lambda_1(A)]}
  |e^{-x}-{p}_{\lambda_n(A),\lambda_1(A),\nfrac{\delta}{2}}(x)| \le
  \nfrac{\delta}{2}\cdot e^{-\lambda_n(A)}$, and that its degree is,
\[k \defeq O\left(\sqrt{\max\{\log^2 \nfrac{1}{\delta},
    (\lambda_1(A)-\lambda_n(A)) \cdot \log \nfrac{1}{\delta}\}}\cdot
  \left(\log \nicefrac{1}{\delta}\right) \cdot \log \log
  \nicefrac{1}{\delta}\right).\] 
Now, we run the {\lanczos} procedure with the matrix $A,$ the
vector $v,$ function $f(x) = e^{-x}$ and parameter $k$ as inputs, and
output the vector $u$ returned by the procedure. 
In order to prove the error guarantee, we use
Theorem~\ref{thm:lanczos} and bound the error using the polynomial
$p_{\lambda_n(A),\lambda_1(A),\nfrac{\delta}{2}}.$ Let $r(x) \defeq
\exp(-x)-p_{\lambda_n(A),\lambda_1(A),\nfrac{\delta}{2}}(x)$. We get,
\begin{align*}
\norm{\exp(-A)v-u} & \stackrel{Thm.~\ref{thm:lanczos}}{\le} 2
\max_{\lambda \in \Lambda(A)} |r_k(\lambda)| \stackrel{\Lambda(A)
  \subseteq [\lambda_n(A),\lambda_1(A)]}{\le} 2
\max_{\lambda \in [\lambda_n(A),\lambda_1(A)]} |r_k(\lambda)| \stackrel{Thm.~\ref{thm:exp-poly-approx:restated}}{\le}
\delta\cdot e^{-\lambda_n(A)}
\end{align*}
By Theorem~\ref{thm:lanczos}, the total running time is $O((n+t_A)k+k^2).$ 
\end{proof}

\subsubsection{The {\expRational} Algorithm.}
Now we move on to applying the Lanczos method in a way that was
suggested as a heuristic by Eshof and Hochbruck~\cite{EH}.  The
starting point here, is the following result by Saff, Schonhage and
Varga \cite{SSV}, that shows that simple rational functions provide
uniform approximations to $e^{-x}$ over $[0,\infty)$ where the error
term decays exponentially with the degree. Asymptotically, this result
is best possible, see \cite{newman}. 
\begin{theorem}[Rational Approximation \cite{SSV}]
\label{thm:rational-approximations}
There exists constants $c_1 \ge 1$ and $k_0$ such that, for any integer $k
\ge k_0$, there exists a polynomial $P_k(x)$ of degree $k-1$ such that,
\[\sup_{x \in [0,\infty)} \left\lvert \exp(-x) -
  \frac{P_k(x)}{(1+\nicefrac{x}{k})^k}\right\rvert \le c_1k\cdot 2^{-k}\ . \]
\end{theorem}
\noindent Note that the rational function given by the above lemma can be
written as a polynomial in $(1+\nicefrac{x}{k})^{-1}$. The following
corollary makes this formal. 
\begin{corollary}[Polynomial in $(1+\nfrac{x}{k})^{-1}$]
\label{cor:pk-star}
There exists constants $c_1 \ge 1$ and $k_0$ such that, for any integer $k
\ge k_0$, there exists a polynomial $p_k^\star(x)$ of degree $k$ such
that $p_k^\star(0)=0,$ and,
\begin{align}
\sup_{t \in (0,1]} \left\lvert e^{-\nfrac{k}{t}+k} - p^\star_k(t)
\right\rvert = \sup_{x \in [0,\infty)} \left\lvert e^{-x} -
  p_k^\star\left((1+\nfrac{x}{k})^{-1}\right)\right\rvert \le
c_1k\cdot 2^{-k}\ .
\end{align}
\end{corollary}
\begin{proof}
  Define $p^\star_{k}$ as $p_k^\star(t) \defeq t^k \cdot
  P_k\left(\nfrac{k}{t}-k\right),$ where $P_k$ is the polynomial from
  Theorem~\ref{thm:rational-approximations}. Note that since $P_k$ is
  a polynomial of degree $k-1$, $p^\star_k$ is a polynomial of degree
  $k$ with the constant term being zero, \emph{i.e.}, $p_k^\star(0) =
  0$.  Also, for any $k \ge k_0$,
\begin{align*}
  \sup_{t \in (0,1]} \left\lvert e^{-\nfrac{k}{t}+k} - p^\star_k(t)
  \right\rvert = \sup_{x \in [0,\infty)} \left\lvert e^{-x} -
    p_k^\star\left((1+\nfrac{x}{k})^{-1}\right)\right\rvert = \sup_{x
    \in [0,\infty)} \left\lvert e^{-x} -
    \frac{P_k(x)}{(1+\nicefrac{x}{k})^k}\right\rvert \le c_1k\cdot
  2^{-k}\ .
\end{align*}
\end{proof}

The corollary above inspires the application of the Lanczos method to
obtain the {\expRational} algorithm that appears in Figure
\ref{fig:expRational}.  We would like to work with the function $f(x)
= e^{k(1-\nfrac{1}{x})}$ and the matrix $B
\defeq (I+\nfrac{A}{k})^{-1}$ for some positive integer $k$ and use
the Lanczos method to compute approximation to $\exp(-A)v$ in the
Krylov subspace $\calK(B,v,k)$, for small $k$.  This is equivalent to
looking for uniform approximations to $\exp(-y)$ that are degree $k$
polynomials in $(1+\nfrac{y}{k})^{-1}.$

Unfortunately, we can't afford to exactly compute the vector
$(I+\nfrac{A}{k})^{-1}y$ for a given vector $y$. Instead, we will
resort to a fast but error-prone solver, \emph{e.g.} the Conjugate
Gradient method and the Spielman-Teng SDD solver
(Theorem~\ref{thm:Spielman-Teng}). Since the computation is now
approximate, the results for Lanczos method no longer apply. Dealing
with the error poses a significant challenge as the Lanczos method is
iterative and the error can propagate quite rapidly. A significant new
and technical part of the paper is devoted to carrying out the error
analysis in this setting. The details appear in Section
\ref{sec:error}.

Moreover, due to inexact computation, we can no longer assume $B$ is
symmetric. Hence, we perform complete orthonormalization while
computing the basis $\{v_i\}_{i=0}^k.$ We also define the symmetric
matrix $\widehat{T}_k \defeq \nfrac{1}{2}\cdot(T_k^\top + T_k)$ and
compute our approximation using this matrix. The complete procedure
{\expRational}, with the exception of specifying the choice of
parameters, is described in Figure \ref{fig:expRational}. We give a
proof of Theorem~\ref{thm:expRational:params} in
Section~\ref{sec:error}.
\begin{figure*}
\begin{tabularx}{\textwidth}{|X|}
\hline
\smallskip
{\bf Input}: A Matrix $A \succeq 0$, a vector $v$ such that $\norm{v}=1,$ and an approximation parameter $\epsilon$. \\
{\bf Output}: A vector $u$ such that $\norm{\exp(-A)v - u} \le \epsilon$.\\
{\bf Parameters:} Let $k \defeq O(\log \nfrac{1}{\eps})$ and
$\epsilon_1 \defeq \exp(-\Theta(k \log k + \log (1+\norm{A}))).$ 
\begin{enumerate}[label=\arabic*.]
      \item Initialize $v_0 \defeq v.$
      \item For $i = 0$ to $k-1,$ \hfill {\footnotesize  (Construct an orthonormal basis to Krylov subspace of order $k$
 )}
        \begin{enumerate}[label=\alph*.]
        \item \label{item:expRational:2} Call the procedure
          $\Invert_A(v_i,k,\varepsilon_1).$ The procedure returns a
          vector $w_i,$ such that, $\norm{(I+\nfrac{A}{k})^{-1}v_i - w_i} \le
          \varepsilon_1 \norm{v_i}.$ \hfill {\footnotesize (Approximate $(I+ \nfrac{A}{k})^{-1}v_i$)}
        \item For $j=0,\ldots,i$,
          \begin{enumerate}[label=\roman*.]
          \item Let $\alpha_{j,i} \defeq v_j^{\top}w_i$. \hfill
            {\footnotesize (Compute projection onto $w_i$)}
          \end{enumerate}
        \item Define $w_i^\prime \defeq w_i - \sum_{j=0}^i
          \alpha_{j,i}v_j$. \hfill
            {\footnotesize (Orthogonalize w.r.t. $v_j$ for  $j \le i$)}
       \item Let $\alpha_{i+1,i} \defeq
          \norm{w_i^\prime}\ ^*$ and $v_{i+1} \defeq w_i^\prime/\alpha_{i+1,i}$. \hfill
            {\footnotesize (Scaling it to norm $1$)}
        \item For $j=i+2,\ldots,k$, 
          \begin{enumerate}[label=\roman*.]
          \item Let $\alpha_{j,i} \defeq 0$.
          \end{enumerate}
        \end{enumerate}
      \item Let $V_k$ be the $n\times (k+1)$ matrix whose columns are
        $v_0,\ldots,v_k$ respectively. \hfill
                  \item Let $T_k$ be the $(k+1)\times(k+1)$ matrix
        $(\alpha_{i,j})_{i,j \in \{0,\ldots,k\}}$ and 
        $\widehat{T}_k \defeq \nfrac{1}{2}(T_k^\top +T_k)$. \hfill
            {\footnotesize (Symmetrize $T_k$)}
     \item Compute ${\mathcal{B}} \defeq \exp\left(k\cdot (I-\widehat{T}_k^{-1})\right)$ exactly and output the vector $V_k\mathcal{B}e_1$.
\vspace{0mm}
{\small \item[*] If $w^\prime_i = 0$, compute the approximation the matrices $T_{i-1}$
  and $V_{i-1},$ instead of $T_k$ and $V_k$. The error bounds still hold. }
\end{enumerate}
\\
\hline
\end{tabularx}
  \caption{The {\expRational} algorithm for approximating $\exp(-A)v$}
  \label{fig:expRational}
\end{figure*}
\subsection{Exponentiating PSD Matrices -- Proofs of
  Theorem~\ref{thm:expRational} and~\ref{thm:expRational2}}
In this section, we give a proof of Theorem~\ref{thm:expRational} and
Theorem~\ref{thm:expRational-psd}, assuming
Theorem~\ref{thm:expRational:params}. Our algorithms for these
theorems are based on the combining the {\expRational} algorithm with
appropriate $\Invert_A$ procedures.
\subsubsection{SDD Matrices -- Proof of Theorems~\ref{thm:expRational}}
\label{sec:expRational:proof}
For Theorem~\ref{thm:expRational} about exponentiating SDD matrices,
we implement the $\Invert_A$ procedure using the Spielman-Teng SDD
solver~\cite{ST3}. Here, we state an improvement on the Spielman-Teng
result by Koutis, Miller and Peng \cite{KMP}.
\begin{theorem}[SDD Solver \cite{KMP}]
\label{thm:Spielman-Teng}
Given a system of linear equations $Mx=b$, where the matrix $M$ is
SDD, and an error parameter $\epsilon > 0$, it is possible to obtain a
vector $u$ that is an approximate solution to the system, in the sense
that
\[\|u-M^{-1}b\|_M \le \epsilon \|M^{-1}b\|_M\ .\]
The time required for this computation is $\tilde{O}\left(m_M \log n
  \log \nfrac{1}{\epsilon}\right)$, where $M$ is an $n\times n$ matrix. (The tilde hides $\log \log n$ factors.)
\end{theorem} 
We restate Theorem~\ref{thm:expRational} for completeness.
\begin{theorem}[Theorem~\ref{thm:expRational} Restated]
Given an $n \times n$ symmetric matrix $A$ which is SDD, a vector $v$ and a
parameter $\delta \le 1$, there is an algorithm that can compute a 
vector $u$ such that $\norm{\exp(-A)v-u} \le \delta\norm{v}$ in 
  time $\tilde{O}((m_A+n)\log (2+\norm{A})).$ The tilde hides $\poly(\log n)$ and $\poly(\log
\nfrac{1}{\delta})$ factors.
\end{theorem}
\begin{proof}
  We use the {\expRational} procedure to approximate the
  exponential. We only need to describe how to implement the
  $\Invert_A$ procedure for an SDD matrix $A$.  Recall that the
  procedure $\Invert_A$, given a vector $y,$ a positive integer $k$
  and real parameter $\varepsilon_1 >0,$ is supposed to return a
  vector $u_1$ such that $\norm{(I+\nfrac{A}{k})^{-1}y-u_1} \le
  \varepsilon_1\norm{y},$ in time
  $T^{\text{inv}}_{A,k,\varepsilon_1}.$ Also, observe that this is
  equivalent to approximately solving the linear system $(I+
  \nfrac{A}{k})z = y$ for the vector $z.$ 

  If the matrix $A$ is SDD, $(I+\nfrac{A}{k})$ is also SDD, and hence,
  we can use the Spielman-Teng SDD solver to implement $\Invert_A$. We
  use Theorem~\ref{thm:Spielman-Teng} with inputs $(I+\nfrac{A}{k}),$
  the vector $y$ and error parameter $\varepsilon_1.$ It returns a
  vector $u_1$ such that,
\[\|(I+\nfrac{A}{k})^{-1}y - u_1\|_{(I+\nfrac{A}{k})} \le
\epsilon_1\|(I+\nfrac{A}{k})^{-1}y\|_{(I+\nfrac{A}{k})}\ .\] 
This implies that,
\begin{align*}
\|(I+\nfrac{A}{k})^{-1}y - u_1\|^2 & = ((I+\nfrac{A}{k})^{-1}y - u_1)^\top ((I+\nfrac{A}{k})^{-1}y - u_1) \\
& \le ((I+\nfrac{A}{k})^{-1}y - u_1)^\top (I+\nfrac{A}{k}) ((I+\nfrac{A}{k})^{-1}y - u_1) \\
& \le  \norm{(I+\nfrac{A}{k})^{-1}y - u_1}^2_{(I+\nfrac{A}{k})} \le \epsilon_1^2 \cdot \|(I+\nfrac{A}{k})^{-1}y\|^2_{(I+\nfrac{A}{k})} \\
& = \epsilon_1^2 \cdot y^\top (I+\nfrac{A}{k})^{-1} y \le \epsilon_1^2 \cdot y^\top y\ ,
\end{align*}
which gives us $\norm{(I+\nfrac{A}{k})^{-1}y-u_1} \le
\varepsilon_1\norm{y}$, as required for $\Invert_A$. Thus,
Theorem~\ref{thm:expRational:params} implies that the procedure
{\expRational} computes a vector $u$ approximating $e^{-A}v$, as
desired.

The time required for the computation of $u_1$ is
$T^{\text{inv}}_{A,k,\varepsilon_1} = \tilde{O}\left((m_A+n) \log n
  \log \nfrac{1}{\epsilon_1}\right),$ and hence from
Theorem~\ref{thm:expRational:params}, the total running time is
$\tilde{O}\left((m_A+n) \log n (\log \nfrac{1}{\delta} + \log
  (1+\norm{A})) \log \nfrac{1}{\delta} + (\log
  \nfrac{1}{\delta})^3\right)$, where the tilde hides polynomial
factors in $\log \log n$ and $\log \log \nicefrac{1}{\delta}$.
\end{proof}

\subsubsection{General PSD Matrices -- Proof of
  Theorem~\ref{thm:expRational-psd}}
\label{sec:expRational-psd:proof}
For Theorem~\ref{thm:expRational-psd} about exponentiating general PSD
matrices, we implement the $\Invert_A$ procedure using the Conjugate
Gradient method.  We use the following theorem.
\begin{theorem}[Conjugate Gradient Method. See~\cite{cg}] 
Given a system of linear equations $Mx=b$ and an error parameter
$\epsilon > 0$, it is possible to obtain a vector $u$ that is an
approximate solution to the system, in the sense that
\[\|u-M^{-1}b\|_M \le \epsilon \|M^{-1}b\|_M.\] The time required for this
computation is $O\left(t_M \sqrt{\kappa(M)} \log \nfrac{1}{\epsilon}
\right),$ -- where $\kappa(M)$ denotes the condition
number of $M$.
\end{theorem}
We restate Theorem~\ref{thm:expRational-psd} for completeness.
\begin{theorem}[Theorem~\ref{thm:expRational-psd} Restated]
\label{thm:expRational-psd-restated}
Given an $n \times n$ symmetric PSD matrix $A$, a vector $v$ and a
parameter $\delta \le 1$, there is an algorithm that can compute a
vector $u$ such that $\norm{\exp(-A)v-u} \le \delta\norm{v}$ in time
$\tilde{O}\left((t_A+n) \sqrt{1+ \norm{A}} \log (2+\norm{A})\right).$
Here the tilde hides $\poly(\log n)$ and $\poly(\log
\nfrac{1}{\delta})$ factors.
\end{theorem}
  
\begin{proof}
  We use the {\expRational} procedure to approximate the
  exponential. We run the Conjugate Gradient method with the on input
  $(I+\nfrac{A}{k}),$ the vector $y$ and error parameter
  $\varepsilon_1.$ The method returns a vector $u_1$ with the same
  guarantee as the SDD solver. As in Theorem~\ref{thm:expRational},
  this implies $\norm{(I+\nfrac{A}{k})^{-1}y-u_1} \le
  \varepsilon_1\norm{y}$, as required for $\Invert_A$.  Thus,
  Theorem~\ref{thm:expRational:params} implies that the procedure
  {\expRational} computes a vector $u$ approximating $e^{-A}v$, as
  desired.

We can compute $u_1$ in time
$T^{\text{inv}}_{A,k,\varepsilon_1}  = O\left(t_A
\sqrt{\frac{1+\nfrac{1}{k} \cdot \lambda_1(A)}{1+\nfrac{1}{k}\cdot
\lambda_n(A)}} \log \nfrac{1}{\varepsilon_1} \right) = O\left(t_A
\sqrt{1+\norm{A} } \log \nfrac{1}{\varepsilon_1} \right),$  and hence from
Theorem~\ref{thm:expRational:params}, the total running time is
 $$\tilde{O}\left(t_A \sqrt{1+\norm{A}}(\log
\nfrac{1}{\delta} + \log (1+\norm{A})) \log \nfrac{1}{\delta} +
(\log \nfrac{1}{\delta})^2\right),$$ where the tilde hides polynomial
factors in $\log \log n$ and $\log \log \nicefrac{1}{\delta}$.
\end{proof}

\subsection{Beyond SDD - Proof of Theorem \ref{thm:expRational2}}
\label{sec:expRational2:proof}
In this section, we give a proof of Theorem~\ref{thm:expRational2},
which we restate below.
\begin{theorem}[Theorem~\ref{thm:expRational2} Restated]
  Given an $n \times n$ symmetric matrix $A=\Pi H M H \Pi$ where $M$
  is SDD, $H$ is a diagonal matrix with strictly positive entries and
  $\Pi$ is a rank $(n-1)$ projection matrix $= 1-ww^\top$ ($w$ is
  explicitly known and $\norm{w}=1$), a vector $v$ and a parameter
  $\delta \le 1$, there is an algorithm that can compute a vector $u$
  such that $\norm{\exp(-A)v-u} \le \delta\norm{v}$ in time
  $\tilde{O}((m_M+n)\log (2+\norm{H M H})).$ The tilde hides
  $\poly(\log n)$ and $\poly(\log \nfrac{1}{\delta})$ factors.
\end{theorem}
\begin{proof}
  In order to prove this, we will use the {\expRational}
  procedure. For $A = \Pi HMH \Pi,$
  Lemma~\ref{lem:invert-proj-diag-sdd} given below implements the
  required $\Invert_A$ procedure. A proof of this lemma is given later
  in this section.
\begin{lemma}[$\Invert_A$ Procedure for Theorem~\ref{thm:expRational2}]
\label{lem:invert-proj-diag-sdd}
Given a positive integer $k,$ vector $y,$ an error parameter
$\varepsilon_1,$ a rank $(n-1)$ projection matrix $\Pi = I - ww^\top$
(where $\norm{w} = 1$ and $w$ is explicitly known), a diagonal matrix
$H$ with strictly positive entries, and an invertible SDD matrix $M$
with $m_M$ non-zero entries; we can compute a vector $u$ such that
$\norm{(I + \nfrac{1}{k}\cdot \Pi H M H\Pi)^{-1}y - u} \le
\varepsilon_1\norm{(I+\nfrac{1}{k}\cdot \Pi HMH\Pi)^{-1}y},$ in time
$\tilde{O}((m_M+n) \log n \log \frac{1+\norm{HMH}}{\varepsilon_1}).$ (The
tilde hides $\poly(\log \log n)$ factors.)
\end{lemma}

Assuming this lemma, we prove our theorem by combining this lemma with
Theorem~\ref{thm:expRational:params} about the {\expRational}
procedure, we get that we can compute the desired vector $u$
approximating $e^{-A}v$ in total time
$$\tilde{O}\left((m_M +n) \log
  n (\log \nfrac{1}{\delta} + \log (1+\norm{HMH})) \log
  \nfrac{1}{\delta} + (\log \nfrac{1}{\delta})^3\right),$$ where the
tilde hides polynomial factors in $\log \log n$ and $\log \log
\nicefrac{1}{\delta}.$
\end{proof}

In order to prove Lemma~\ref{lem:invert-proj-diag-sdd}, we need to
show how to approximate the inverse of a matrix of the form $HMH,$
where $H$ is diagonal and $M$ is SDD. The following lemma achieves this.
\begin{lemma}
\label{lem:invert-diag-sdd}
Given a vector $y,$ an error parameter $\varepsilon_1,$ a diagonal
matrix $H$ with strictly positive entries, and an invertible SDD
matrix $M$ with $m_m$ non-zero entries; we can compute a vector $u$
such that $\norm{(HMH)^{-1}y - u}_{HMH} \le
\varepsilon_1\norm{(HMH)^{-1}y}_{HMH},$ in time $\tilde{O}((m_M+n)
\log n \log \nfrac{1}{\varepsilon_1}). $ (The tilde hides factors of
$\log \log n.$)
\end{lemma}
\begin{proof}
  Observe that $(HMH)^{-1}y = H^{-1}M^{-1}H^{-1}y.$ Use the SDD solver (Theorem~\ref{thm:Spielman-Teng})
  with inputs $M,$ vector $H^{-1}y$ and parameter $\varepsilon_1$ to
  obtain a vector $u_1$ such that,
  \[\norm{M^{-1}(H^{-1}y) - u_1}_{M} \le
  \varepsilon_1\norm{M^{-1}H^{-1}y}_{M} .\] Return the vector $u
  \defeq H^{-1}u_1.$ We can bound the error in the output vector $u$
  as follows,
\begin{align*}
  \norm{(HMH)^{-1}y - u}^2_{HMH} & = \norm{H^{-1}M^{-1}H^{-1}y -
    H^{-1}u_1}^2_{HMH} \\
  & = (H^{-1}M^{-1}H^{-1}y - H^{-1}u_1)^\top (HMH)
  (H^{-1}M^{-1}H^{-1}y
  - H^{-1}u_1) \\
  & = (M^{-1}H^{-1}y - u_1)^\top M (M^{-1}H^{-1}y  - u_1) \\
  & = \norm{M^{-1}(H^{-1}y) - u_1}^2_{M}\\
  & \le \varepsilon_1^2\norm{M^{-1}H^{-1}y}^2_{M} = \varepsilon_1^2
  (M^{-1}H^{-1}y)^\top M (M^{-1}H^{-1}y) \\
  & = \varepsilon_1^2 (H^{-1}M^{-1}H^{-1}y)^\top (HMH)
  (H^{-1}M^{-1}H^{-1}y) \\
  & = \varepsilon_1^2 \norm{H^{-1}M^{-1}H^{-1}y}^2_{HMH} = \varepsilon_1^2
  \norm{(HMH)^{-1}y}^2_{HMH}
\end{align*}

Thus, $\norm{(HMH)^{-1}y-u}_{HMH} \le \varepsilon_1
\norm{(HMH)^{-1}y}_{HMH}.$ Since $H$ is diagonal, multiplication by
$H^{-1}$ requires $O(n)$ time. Hence, the total time is dominated by
the SDD solver, giving a total running time of $\tilde{O}((m_M+n) \log n
\log \nfrac{1}{\varepsilon_1}).$
\end{proof}
\noindent Now, we prove Lemma~\ref{lem:invert-proj-diag-sdd}. 
\begin{lemma}[Lemma~\ref{lem:invert-proj-diag-sdd} Restated]
Given a positive integer $k,$ vector $y,$ an error parameter
$\varepsilon_1,$ a rank $(n-1)$ projection matrix $\Pi = I - ww^\top$
(where $\norm{w} = 1$ and $w$ is explicitly known), a diagonal matrix
$H$ with strictly positive entries, and an invertible SDD matrix $M$
with $m_M$ non-zero entries; we can compute a vector $u$ such that
$\norm{(I + \nfrac{1}{k}\cdot \Pi H M H\Pi)^{-1}y - u} \le
\varepsilon_1\norm{(I+\nfrac{1}{k}\cdot \Pi HMH\Pi)^{-1}y},$ in time
$\tilde{O}((m_M+n) \log n \log \frac{1+\norm{HMH}}{\varepsilon_1}).$ (The
tilde hides $\poly(\log \log n)$ factors.)
\end{lemma}
\begin{proof}
  We sketch the proof idea first. Using the fact that $w$ is an
  eigenvector of our matrix, we will split $y$ into two components --
  one along $w$ and one orthogonal. Along $w,$ we can easily compute
  the component of the required vector. Among the orthogonal
  component, we will write our matrix as the sum of $I +
  \nfrac{1}{k}\cdot HMH$ and a rank one matrix, and use the
  Sherman-Morrison formula to express its inverse. Note that we can
  use Lemma~\ref{lem:invert-proj-diag-sdd} to compute the inverse of
  $I + \nfrac{1}{k}\cdot HMH.$ The procedure is described in
  Figure~\ref{fig:invert_A} and the proof for the error analysis is
  given below.

  Let $M_1 \defeq \nfrac{1}{k}\cdot HMH.$ Then, $I + \nfrac{1}{k}\cdot
  \Pi HMH \Pi  = I + \Pi M_1 \Pi.$ Without loss of generality,
  we will assume that $\norm{y}=1$.  Note that $I + \Pi M_1 \Pi \succ
  0,$ and hence is invertible.  Let $z \defeq y - (w^\top y)w$. Thus,
  $w^\top z = 0.$ Since $w$ is an eigenvector of $(I+\Pi M_1\Pi)$ with
  eigenvalue 1, we get,
\begin{align}
(I+\Pi M_1 \Pi)^{-1}y = (I+\Pi M_1 \Pi)^{-1}z + (w^\top  y)w.
\label{eq:projection-inverse1}
\end{align}
Let's say $t\defeq (I+\Pi M_1 \Pi)^{-1}z$. Then, $t + \Pi M_1 \Pi t = z$. Left-multiplying by $w^\top $, we get, $w^\top t = w^\top z = 0.$ Thus, $\Pi t = t,$ and hence $(I + \Pi M_1)t = z$, or equivalently, $t = (I + \Pi M_1)^{-1} z.$
\begin{align}
(I+\Pi M_1 \Pi)^{-1}z & = 
(I + \Pi M_1)^{-1}z = (I + M_1 - ww^\top M_1)^{-1}z = (I + M_1 - w(M_1w)^\top )^{-1}z \nonumber \\
& = \left((I + M_1)^{-1} - \frac{(I+M_1)^{-1}ww^\top
    M_1(I+M_1)^{-1}}{1 + w^\top M_1(I+M_1)^{-1}w}\right)z \nonumber\\
& \qquad \qquad \qquad \qquad \textrm{(Sherman-Morrison formula)} \nonumber\\
& = (I + M_1)^{-1}z - \frac{w^\top M_1(I+M_1)^{-1}z}{1 + w^\top M_1(I+M_1)^{-1}w}(I+M_1)^{-1}w
\label{eq:projection-inverse2}
\end{align}
Since we can write $I + M_1 = I + HMH = H(H^{-2} + M)H,$ we can use
Lemma~\ref{lem:invert-diag-sdd} to estimate $(I+M_1)^{-1}z$ and
$(I+M_1)^{-1}w$. Using Equation~\eqref{eq:projection-inverse2}, the
procedure for estimating $(I+\Pi M_1 \Pi)^{-1}x$ is described in
Figure~\ref{fig:invert_A}.
\begin{figure*}
\begin{tabularx}{\textwidth}{|X|}
\hline
\vspace{-3mm}
\begin{enumerate}[label=\arabic*.]
\item Compute $z\defeq y - (w^\top y)w.$
\item Estimate $(I+M_1)^{-1}z$ with error parameter $\frac{\varepsilon_1}{6(1+\norm{M_1})}.$ Denote the vector returned by $\beta_1$.
\item Estimate $(I+M_1)^{-1}w$ with error parameter $\frac{\varepsilon_1}{6(1+\norm{M_1})}.$ Denote the vector returned by $\beta_2$.
\item Compute
\begin{equation}
u_1 \defeq \beta_1 - \frac{w^\top M_1\beta_1}{1+w^\top
  M_1\beta_2}\beta_2 + (w^\top y)w.
\label{eq:exp-projection-estimate}
\end{equation}
Return $u_1$.
\vspace{-2mm}
\end{enumerate}
\\
\hline
\end{tabularx}
  \caption{The $\Invert_A$ procedure for Theorem~\ref{thm:expRational2}}
  \label{fig:invert_A}
\end{figure*}

We need to upper bound the error in the above estimation
procedure. From the assumption, we know that $\beta_1 = (I+M_1)^{-1}z
- e_1,$ where $\norm{e_1}_{(I+M_1)} \le
\frac{\varepsilon_1}{6(1+\norm{M_1})} \norm{(I+M_1)^{-1}z}_{(I+M_1)}$,
and $\beta_2 = (I+M_1)^{-1}x - e_2,$ where $,\norm{e_2}_{(I+M_1)} \le
\frac{\varepsilon_1}{6(1+\norm{M_1})}\norm{(I+M_1)^{-1}w}_{(I+M_1)}.$
Combining Equations~\eqref{eq:projection-inverse1} and
\eqref{eq:projection-inverse2} and subtracting
Equation~\eqref{eq:exp-projection-estimate}, we can write the error
as,
\begin{align*}
& (I+\Pi M_1 \Pi)^{-1}y - u_1 = (I + M_1)^{-1}z - \beta_1 -\frac{w^\top M_1(I+M_1)^{-1}z}{1 + w^\top M_1(I+M_1)^{-1}w}(I+M_1)^{-1}w + \frac{w^\top M_1\beta_1}{1+w^\top  M_1\beta_2}\beta_2 \\
& \qquad  = e_1 -\frac{w^\top M_1(I+M_1)^{-1}z}{1 + w^\top M_1(I+M_1)^{-1}w}(I+M_1)^{-1}w + \frac{w^\top M_1[(I+M_1)^{-1}z - e_1]}{1 + w^\top M_1[(I+M_1)^{-1}w - e_2]}[(I+M_1)^{-1}w - e_2] \\
& \qquad = e_1 +\frac{w^\top M_1(I+M_1)^{-1}z \cdot w^\top M_1e_2}{(1 + w^\top M_1(I+M_1)^{-1}w) (1 + w^\top M_1[(I+M_1)^{-1}w - e_2])}(I+M_1)^{-1}w \\
& \qquad  \qquad
- \frac{w^\top M_1e_1}{1 + w^\top M_1[(I+M_1)^{-1}w - e_2]}[(I+M_1)^{-1}w - e_2]
- \frac{w^\top M_1(I+M_1)^{-1}z}{1 + w^\top M_1[(I+M_1)^{-1}w - e_2]}e_2
\end{align*}
Let us first bound the scalar terms. Note that $\norm{z} \le \norm{y} =1.$
\[|w^\top M_1(I+M_1)^{-1}z| \le \norm{w}\norm{M_1(I+M_1)^{-1}z} \le \norm{M_1(I+M_1)^{-1}} \le 1\ ,\]
\begin{align*}
w^\top M_1e_1 \le \norm{w}\norm{M_1}\norm{e_1} \le
\norm{M_1}\norm{e_1}_{(I+M_1)} & \le \nfrac{\varepsilon_1}{6}\cdot
\norm{(I+M_1)^{-1}z}_{(I+M_1)} \\
& =
\nfrac{\varepsilon_1}{6}\cdot\norm{(I+M_1)^{-\nfrac{1}{2} }z} \le
\nfrac{\varepsilon_1}{6}.
\end{align*}
Similarly, $w^\top M_1e_2 \le \nfrac{\varepsilon_1}{6}\cdot.$ Also,
$M_1(I+M_1)^{-1} \succeq 0$ and hence $w^\top M_1(I+M_1)^{-1}w \ge 0$. Thus,
\begin{align*}
\norm{(I+\Pi M_1 \Pi)^{-1}y - u_1} & \le \norm{(I+\Pi M_1 \Pi)^{-1}y -
  u_1}_{(I+M_1)}  \\
& \le \norm{e_1}_{(I+M_1)} + \frac{1\cdot \nfrac{\varepsilon_1}{6}}{1 \cdot (1-\nfrac{\varepsilon_1}{6})}\norm{(I+M_1)^{-1}w}_{(I+M_1)} \\
& \quad + \frac{\nfrac{\varepsilon_1}{6}}{
  (1-\nfrac{\varepsilon_1}{6})}(\norm{(I+M_1)^{-1}w}_{(I+M_1)} +
\norm{e_2}_{(I+M_1)}) +
\frac{1}{1-\nfrac{\varepsilon_1}{6}}\norm{e_2}_{(I+M_1)}\\
& \le \frac{\varepsilon_1}{6(1+\norm{M_1})} \norm{(I+M_1)^{-1}z}_{(I+M_1)} + \frac{4\cdot
  \nfrac{\varepsilon_1}{6}}{1-\nfrac{\varepsilon_1}{6}}\norm{(I+M_1)^{-1}w}_{(I+M_1)}\\
& \le  \nfrac{\varepsilon_1}{6} \norm{(I+M_1)^{-\nfrac{1}{2}}z} +
\nfrac{4\varepsilon_1}{5} \norm{(I+M_1)^{-\nfrac{1}{2}}w}\le \varepsilon_1
\end{align*}
Other than the estimation of $(I+M_1)^{-1}z$ and $(I+M_1)^{-1}w$, 
we need to compute a constant number of dot products and a constant
number of matrix-vector products with the matrix $M_1.$ Multiplying a
vector with $M_1 = \nfrac{1}{k}\cdot HMH$ takes time $O(m_M+n)$, giving a total
time of $\tilde{O}((m_M+n) \log n
\log \frac{1+\norm{HMH}}{\varepsilon_1})$
\end{proof}

\subsection{Error Analysis for {\expRational}}\label{sec:error}
In this section, we give the proof of Theorem~\ref{thm:expRational:params},
except for the proof of a few lemmas, which have been presented in the
Section~\ref{appendix:proofs:error} for better readability.

\subsubsection*{\it Proof Overview.}
At a very high-level, the proof follows the outline of the proof for
Lanczos method. We first show that assuming the error in computing the
inverse is small, $\widehat{T}_k$ can be used to approximate small
powers of $B= (I+\nfrac{A}{k})^{-1}$ when restricted to the Krylov
subspace, {\em i.e.} for all $i \le k,$ $\|B^i v-V_k \widehat{T}^i_k
V_k^\top v\| \lessapprox \varepsilon_2,$ for some small
$\varepsilon_2.$. This implies that we can bound the error in
approximating $p((I+\nfrac{A}{k})^{-1})$ using $p(\widehat{T}_k)$, by
$\varepsilon_2\norm{p}_1,$ where $p$ is a polynomial of degree at most
$k.$ This is the most technical part of the error analysis because we
need to capture the propagation of error through the various
iterations of the algorithm. We overcome this difficulty by expressing
the final error as a sum of $k$ terms, with the $i^\text{th}$ term
expressing how much error is introduced in the final candidate vector
because of the error in the inverse computation during the
$i^\text{th}$ iteration. Unfortunately, the only way we know of
bounding each of these terms is by {\em tour de force}. A part of this
proof is to show that the spectrum of $\widehat{T}_k$ cannot shift far
from the spectrum of $B.$

To bound the error in the candidate vector output by the algorithm,
{\em i.e.} $\|f(B)v-V_kf(\widehat{T}_k)V_k^\top v\|,$ we start by
expressing $e^{-x}$ as the sum of a degree $k$-polynomial $p_k$ in
$(1+\nfrac{x}{k})^{-1}$ and a remainder function $r_k.$ We use the
analysis from the previous paragraph to upper bound the error in the
polynomial part by $\varepsilon_2 \norm{p}_1.$ We bound the
contribution of the remainder term to the error by bounding
$\norm{r_k(B)}$ and $\|{r_k(\widehat{T}_k)}\|.$ This step uses the
fact that eigenvalues of $r_k(\widehat{T}_k)$ are
$\{r_k(\lambda_i)\}_i,$ where $\{\lambda_i\}_i$ are eigenvalues of
$\widehat{T_k}.$ To complete the error analysis, we use the
polynomials $p_k^\star$ from Corollary~\ref{cor:pk-star} and bound its
$\ell_1$ norm. Even though we do not know $p_k^\star$ explicitly, we
can bound $\norm{p_k^\star}_1$ indirectly by writing it as an
interpolation polynomial and using that the values it assumes in
$[0,1]$ have to be small in magnitude.

\begin{proof}
For notational convenience, define $B \defeq
(I+\nfrac{A}{k})^{-1}.$ Since the computation of $Bv_i$ is not exact
in each iteration, the eigenvalues of $\widehat{T}_k$ need not be
eigenvalues of $B$. Also, Lemma \ref{lem:exact-power} no longer holds,
\emph{i.e.}, we can't guarantee that $V_k\widehat{T}_k^te_1$ is identical
to $B^tv_0.$ However, we can prove the following lemma that proves
bounds on the spectrum of $\widehat{T}_k$ and also bounds the norm of the
difference between the vectors $V_k\widehat{T}_k^te_1$ and $B^tv_0.$
This is the most important and technically challenging part of the proof.
\begin{lemma}[Approximate Computation with $\widehat{T}_k$. Proof in Sec.~\ref{appendix:proofs:error}]
\label{lem:t-hat}
The coefficient matrix $\widehat{T}_k$ generated satisfies the following:
\begin{enumerate}[label=\arabic*.]
\item The eigenvalues of $\widehat{T}_k$ lie in  $\left[
    \left(1+\frac{\lambda_1(A)}{k}\right)^{-1} -
    \epsilon_1\sqrt{k+1},\left(1+\frac{\lambda_n(A)}{k}\right)^{-1}
    + \epsilon_1\sqrt{k+1} \right].$
\item For any $t \le k$, if $\epsilon_1 \le \epsilon_2 /
  (8(k+1)^{\nicefrac{5}{2}})$ and $\epsilon_2 \le 1$, we have, 
$\norm{B^t v_0 - V_k\widehat{T}_k^te_1} \le \epsilon_2\ .$
\end{enumerate}
\end{lemma}
Here is an idea of the proof of the above lemma: Since, during every
iteration of the algorithm, the computation of $Bv_i$ is approximate,
we will express $BV_k$ in terms of $T_k$ and an error matrix $E$. This
will allow us to express $\widehat{T}_k$ in terms of $T_k$ and a
different error matrix. The first part of the lemma will follow
immediately from the guarantee of the $\Invert_A$ procedure.

For the Second part, we first express $BV_k - V_k \widehat{T}_k$ in
terms of the error matrices defined above. Using this, we can write
the telescoping sum $B^tV_k - V_k \widehat{T}^t_k = \sum_{j=1}^t
B^{t-j}(BV_k - V_k\widehat{T}_k)\widehat{T}^{j-1}_k.$ We use triangle
inequality and a \emph{tour de force} calculation to bound each
term. A complete proof is included in Section~\ref{appendix:proofs:error}.

\noindent 
As a simple corollary, we can bound the error in the computation of
the polynomial, in terms of the $\ell_1$ norm of the polynomial being
computed.
\begin{corollary}[Approximate Polynomial Computation. Proof in
  Sec.~\ref{appendix:proofs:error}]
\label{cor:approx-poly-err-l1}
  For any polynomial $p$ of degree at most $k$, if $\epsilon_1 \le
  \epsilon_2 / (2(k+1)^{\nicefrac{3}{2}})$ and $\epsilon_2 \le 1$,
\[\norm{p(B) v_0 - V_kp(\widehat{T}_k)e_1} \le \epsilon_2 \norm{p}_1.\]
\end{corollary}

\noindent 
Using this corollary, we can prove an analogue of Lemma \ref{lem:exact-poly}, giving
error bounds on the procedure in terms of degree $k$ polynomial
approximations. The proof is very similar and is based on writing $f$
as a sum of a degree $k$ polynomial and an \emph{error function}.
\begin{lemma}[Polynomial Approximation for {\expRational}. Proof in Sec.~\ref{appendix:proofs:error}]
\label{lem:min-poly-expv1}
Let $V_k$ be the ortho-normal basis and $\widehat{T}_k$ be the matrix of
coefficients generated by {\expRational}. Let $f$ be any function such that
$f(B)$ and $f(T_k)$ are defined. Define $r_k(x) \defeq f(x) - p(x).$ Then,
\begin{align}
\label{eq:final-error}
\norm{f(B)v_0 - V_kf(\widehat{T}_k)e_1} \le \min_{p \in \Sigma_k}
\left(\epsilon_2 \norm{p}_1 + \max_{\lambda \in \Lambda(B)}
  |r_k(\lambda)| +  \max_{\lambda \in \Lambda(\widehat{T}_k)}
  |r_k(\lambda)| \right).
\end{align}
\end{lemma}

\noindent In order to control the second error term in the above
lemma, we need to bounds the eigenvalues of $\widehat{T}_k$, which is
provided by Lemma~\ref{lem:t-hat}.
 
For our application, $f(t) = f_{k}(t) \defeq
\exp\left(k\cdot\left(1-\nfrac{1}{t}\right)\right)$ so that
$f_{k}((1+\nfrac{x}{k})^{-1}) = \exp(-x)$. This function is
discontinuous at $t=0$. Under exact computation of the inverse, the
eigenvalues of $\widehat{T}_k$ would be the same as the eigenvalues of $B$
and hence would lie in $(0,1]$. Unfortunately, due to the errors, the
eigenvalues of $\widehat{T}_k$ could be outside the interval. Since $f$ is
discontinuous at 0, and goes to infinity for small negative values, in
order to get a reasonable approximation to $f$, we will ensure that
the eigenvalues of $\widehat{T}_k$ are strictly positive, \emph{i.e.},
$\epsilon_1\sqrt{k+1} < (1+\nfrac{1}{k}\cdot \lambda_1(A))^{-1}$.

We will use the polynomials $p_k^\star$ from
Corollary~\ref{cor:pk-star} in Lemma~\ref{lem:min-poly-expv1} to bound
the final error. We will require the following lemma to bound the
$\ell_1$-norm of $p_k^\star.$
\begin{lemma}[$\ell_1$-norm Bound. Proof in Sec.~\ref{appendix:proofs:error}]
\label{lem:L1-norm}
Given a polynomial $p$ of degree $k$ such that $p(0) = 0$ and 
\[\sup_{t \in (0,1]} \left\lvert e^{-\nfrac{k}{t}+k} - p(t) \right\rvert = \sup_{x \in [0,\infty)} \left\lvert e^{-x} -
  p\left((1+\nfrac{x}{k})^{-1}\right)\right\rvert \le 1\ ,\] we must
have $\norm{p}_1 \le (2k)^{k+1}.$
\end{lemma}
This lemma is proven by expressing $p$ as the interpolation polynomial
on the values attained by $p$ at the $k+1$ points
$0,\nfrac{1}{k},\ldots,\nfrac{k}{k},$ which allows us to express the
coefficients in terms of these values. We can bound these values, and
hence, the coefficients, since we know that $p$ isn't too far from the
exponential function. A complete proof is included in
Section~\ref{appendix:proofs:error}.

Corollary~\ref{cor:pk-star} shows that
$p_k^\star(t)$ is a good uniform approximation to
$e^{-\nfrac{k}{t}+k}$ over the interval $(0,1]$. Since $\Lambda(B) \subseteq
(0,1],$ this will help us help us bound the second error term in
Equation~\eqref{eq:final-error}. Since $\widehat{T}_k$ can have
eigenvalues larger that 1, we need to bound the error in approximating
$f_{k}(t)$ by $p_{k}^{\star}(t)$ over an interval $(0,\beta],$ where
$\beta \ge 1$. The following lemma, gives us the required error
bound. This proof for this lemma bounds the error over $[1,\beta]$ by
applying triangle inequality and bounding the change in $f_k$ and $p$
over $[1,\beta]$ separately.
\begin{lemma}[Approximation on Extended
  Interval. Proof in Sec.~\ref{appendix:proofs:error}]
\label{lem:approx-extended-interval}
For any $\beta \ge 1$, any degree $k$ polynomial $p$ satisfies,
\[\sup_{t \in (0,\beta]} |p(t) - f_{k}(t)| \le \norm{p}_1 \cdot (\beta^k-1) + (f_{k}(\beta)-f_{k}(1))+ \sup_{t \in (0,1]}
|p(t)-f_{k}(t)|\ . \]
\end{lemma}

\noindent We bound the final error using the polynomial
$p_{k}^\star$ in Equation \eqref{eq:final-error}. We will use the
above lemma for $\beta \defeq 1+\epsilon_1\sqrt{k+1}$ and assume that $\epsilon_1\sqrt{k+1} < (1+\nfrac{1}{k}\cdot \lambda_1(A))^{-1}.$ 
\begin{align*}
\norm{f(B)v_0 - V_kf(\widehat{T}_k)e_1} & \le \epsilon_2\norm{p_{k}^\star}_1
+ \max_{\lambda \in \Lambda(B)} |r_k(\lambda)| +  \max_{\lambda \in \Lambda(\widehat{T}_k)} |r_k(\lambda)| \\
& \le \epsilon_2\norm{p_{k}^\star}_1 + \sup_{\lambda \in
  (0,1]} |(f_{k} - p_{k}^\star)(\lambda)|+ \sup_{\lambda \in
  (0,\beta]} |(f_{k} - p_{k}^\star)(\lambda)|. \\
& \qquad \qquad \text{(Since $\Lambda(B) \subseteq (0,1]$ and
  $\Lambda(\widehat{T}_k) \subseteq (0,\beta]$ )} \\
& \le \epsilon_2\norm{p_{k}^\star}_1 + \sup_{t \in (0,1]}
|p_{k}^\star(t)-f_{k}(t)| + \\
& \qquad \qquad \norm{p_{k}^\star}_1 \cdot (\beta^k-1) +
(f_{k}(\beta)-f_{k}(1))+ \sup_{t \in (0,1]} |p_{k}^\star(t)-f_{k}(t)|
\\
& = \norm{p_{k}^\star}_1 \cdot (\epsilon_2 + \beta^k-1) +
(\exp\left(\nfrac{k(\beta-1)}{\beta}\right) -1)+ 2\sup_{t \in (0,1]}
|p_{k}^\star(t)-f_{k}(t)|\ .
\end{align*}
Given $\delta < 1$, we plug in the following parameters,
\[k\defeq\max\{k_0,\log_2 \nicefrac{8c_1}{\delta} + 2
\log_2 \log_2 \nicefrac{8c_1}{\delta}\} = O\left(\log
  \nicefrac{1}{\delta}\right)\ ,\]
\[\epsilon_1 \defeq
\nfrac{\delta}{32}\cdot(k+1)^{-\nicefrac{5}{2}}\cdot(1+\nfrac{1}{k}\cdot
\lambda_1(A))^{-1}\cdot (2k)^{-k-1},\ \beta \defeq
1+\epsilon_1\sqrt{k+1},\ \epsilon_2 \defeq
8(k+1)^{\nicefrac{5}{2}}\epsilon_1 \ ,\]
where $k_0,c_1$ are the constants given by
Corollary~\ref{cor:pk-star}. Note that these parameters satisfy the
condition $\epsilon_1\sqrt{k+1} <
(1+\nfrac{1}{k}\cdot\lambda_1(A))^{-1}$. Corollary~\ref{cor:pk-star}
implies that $p_k^\star(0)=0$ and
\begin{align}
\sup_{t \in (0,1]}
|p_{k}^\star(t)-f_{k}(t)| & \le
\frac{\delta}{8}\cdot \frac{\log_2 \nicefrac{8c_1}{\delta} + 2 \log_2 \log_2
\nicefrac{8c_1}{\delta}}{(\log_2 \nicefrac{8c_1}{\delta})^2} \nonumber\\ & \le
\frac{\delta}{8} \cdot \frac{1}{\log_2 \nicefrac{8c_1}{\delta}} \left(1+ 2\cdot \frac{\log_2 \log_2
\nicefrac{8c_1}{\delta}}{\log_2 \nicefrac{8c_1}{\delta}} \right) \le
\frac{\delta}{8} \cdot \frac{1}{3} \cdot 3 \le \frac{\delta}{8}\ ,
\label{eq:rational-error2}
\end{align}
where the last inequality uses $\delta \le 1 \le c_1$ and $\log_2 x \le x,
\forall x \ge 0$. Thus, we can use Lemma~\ref{lem:L1-norm} to
conclude that $\norm{p_k^\star}_1 \le (2k)^{k+1}.$

\noindent
We can simplify the following expressions,
\[\exp\left(\nfrac{k(\beta-1)}{\beta}\right)-1 \le
\exp\left(k\epsilon_1\cdot\sqrt{k+1}\right) -1 \le
\exp(\nfrac{\epsilon_2}{8}) -1 \le (1+\nfrac{\epsilon_2}{4})-1 =
\nfrac{\epsilon_2}{4}\ ,\]
\[\beta^k-1 = (1+\epsilon_1 \sqrt{k+1})^k -1 \le \exp(k\cdot
\epsilon_1\sqrt{k+1}) -1 \le \nfrac{\epsilon_2}{4} .\]
Thus the total error $\norm{u-\exp(-A)v} = \norm{f(B)v_0 -
  V_kf(\widehat{T}_k)e_1} \le (2k)^{k+1}\cdot 2\epsilon_2 + \epsilon_2 +
\nfrac{\delta}{4} \le \delta$.

\subsubsection*{\it Running Time.} The running time for the procedure is
dominated by $k$ calls to the $\Invert_A$ procedure with parameters $k$
and $\epsilon_1$, computation of at most $k^2$ dot-products and the
exponentiation of $\widehat{T}_k$. The exponentiation of
$\widehat{T}_k$ can be done in time $O(k^3)$~\cite{eigendecomp}. Thus the total
running time is $O(T^{\text{inv}}_{A,k,\varepsilon_1}\cdot k + n\cdot
k^2 + k^3).$
\noindent This completes the proof of the Theorem \ref{thm:expRational:params}.
\end{proof}

\subsubsection{Remaining Proofs}
In this section, we give the remaining proofs in Section~\ref{sec:error}.
\label{appendix:proofs:error}
\begin{lemma}[Lemma \ref{lem:t-hat} Restated]
\label{lem:t-hat:restated}
The coefficient matrix $\widehat{T}_k$ generated satisfies the following:
\begin{enumerate}[label=\arabic*.]
\item The eigenvalues of $\widehat{T}_k$ lie in the interval $\left[
    \left(1+\frac{\lambda_1(A)}{k}\right)^{-1} -
    \epsilon_1\sqrt{k+1},\left(1+\frac{\lambda_n(A)}{k}\right)^{-1}
    + \epsilon_1\sqrt{k+1} \right].$
\item For any $t \le k$, if $\epsilon_1 \le \epsilon_2 /
  (8(k+1)^{\nicefrac{5}{2}})$ and $\epsilon_2 \le 1$, we have, 
$\norm{B^t v_0 - V_k\widehat{T}_k^te_1} \le \epsilon_2\ .$
\end{enumerate}
\end{lemma}
\begin{proof}
Given a vector $y,$ a positive integer $k$ and real parameter
$\varepsilon_1 >0,$ $\Invert_A(y,k,\varepsilon_1)$ returns a vector
$u_1$ such that $\norm{By-u_1} \le
\varepsilon_1\norm{y},$ in time $T^{\text{inv}}_{A,k,\varepsilon_1}.$ Thus, for
each $i$, the vector $w_i$ satisfies
$\|Bv_i - w_i\| \le \epsilon_1\norm{v_i} = \epsilon_1 .$
\noindent Also define $u_i$ as $u_i \defeq Bv_i-w_i$.
Thus, we get, $\norm{u_i} \le \epsilon_1$. Let $E$ be the $n\times
(k+1)$ matrix with its columns being $u_0,\ldots,u_k$. We can write
the following recurrence,
\begin{align}
\label{eq:modified-recurrence}
BV_k & = V_kT_k + E + \alpha_{k,k+1}v_{k+1}e^\top _{k+1}\ ,
\end{align}
where each column of $E$ has $\ell_2$ norm at most $\epsilon_1$. Note
that we continue to do complete orthonormalization, so $V_k^\top V_k =
I_k$. Thus, $T_k$ is not tridiagonal, but rather \emph{Upper
  Hessenberg}, \emph{i.e.}, $(T_k)_{ij} = 0$ whenever $i > j+1$.

\noindent Multiplying both sides of Equation \eqref{eq:modified-recurrence} by $V_k^\top $, we get $T_k = V_k^\top BV_k - V_k^\top E$. This implies,
\begin{align}
\label{eq:Tk-hat}
\widehat{T}_k & =  V_k^\top BV_k  - \nfrac{1}{2}\cdot(V_k^\top E + E^\top  V_k) \\
&  = V_k^\top (V_kT_k + E + \alpha_{k,k+1}v_{k+1}e^\top _{k+1})  - \nfrac{1}{2}\cdot(V_k^\top E + E^\top  V_k) \qquad \text{(Using \eqref{eq:modified-recurrence}})\nonumber \\
& = T_k + \nfrac{1}{2}\cdot(V_k^\top E - E^\top V_k).
\label{eq:Tk-hat2}
\end{align}

\noindent Define $E_1 \defeq \nfrac{1}{2}\cdot(V_k^\top E + E^\top  V_k)$. Thus, using Equation \eqref{eq:Tk-hat}, $\widehat{T}_k = V^\top _kBV_k - E_1$. Let us first bound the norm of $E_1$.
\begin{align*}
\norm{E_1}\le \nfrac{1}{2}\cdot(\norm{V_k^\top E}+\norm{E^\top V_k}) \le \nfrac{1}{2}\cdot(\norm{E}+\norm{E^\top }) \le \norm{E}_F \le \epsilon_1 \sqrt{k+1} \ .
\end{align*}

\noindent Since $\widehat{T}_k = V_k^\top BV_k - E_1$. 
 We have,
\begin{align*}
\lambda_{\max}(\widehat{T}_k) & \le \lambda_{1}(B) + \norm{E_1} \le (1+\nfrac{1}{k}\cdot \lambda_{n}(A))^{-1} + \epsilon_1\sqrt{k+1}\ ,\\
\lambda_{\min}(\widehat{T}_k) & \ge \lambda_{n}(B) - \norm{E_1}  \ge (1+\nfrac{1}{k}\cdot \lambda_{1}(A))^{-1} - \epsilon_1\sqrt{k+1}\ .
\end{align*}
(We use $\lambda_{\max}$ and $\lambda_{\min}$ for the largest and
smallest eigenvalues of $\widehat{T}_k$ respectively in order to avoid confusion
since $\widehat{T}_k$ is a $(k+1)\times (k+1)$ matrix and not an $n
\times n$ matrix.)

\noindent First, let us compute $BV_k - V_k\widehat{T}_k$.
\begin{align}
BV_k - V_k\widehat{T}_k \stackrel{\eqref{eq:modified-recurrence},\eqref{eq:Tk-hat2}}{=}~ & V_kT_k + E + \alpha_{k,k+1}v_{k+1}e^\top _{k+1} - V_k\left(T_k + \nfrac{1}{2}\cdot(V_k^\top E - E^\top V_k)\right) \nonumber\\ 
\label{eq:approx-poly1}
~=~~ & \left(I-\nfrac{1}{2}\cdot V_kV_k^\top \right)E + \nfrac{1}{2}\cdot V_kE^\top V_k + \alpha_{k,k+1}v_{k+1}e^\top _{k+1}\ .
\end{align}

\noindent Now,
\begin{align}
  \norm{B^tv_0 - V_k\widehat{T}^t_ke_1} & = \norm{\sum_{j=1}^t B^{t-j}(BV_k - V_k\widehat{T}_k)\widehat{T}^{j-1}_ke_1} \qquad \qquad \qquad \text{(Telescoping sum)}\nonumber\\
  & \stackrel{\eqref{eq:approx-poly1}}{=} \norm{\sum_{j=1}^t B^{t-j}\left(\left(I-\nfrac{1}{2}\cdot V_kV_k^\top \right)E + \nfrac{1}{2}\cdot V_kE^\top V_k + \alpha_{k,k+1}v_{k+1}e^\top _{k+1}\right)\widehat{T}^{j-1}_ke_1} \nonumber\\
  &
  \stackrel{\Delta-\text{ineq.}}{\le} \sum_{j=1}^t\norm{ B^{t-j}\left(\left(I-\nfrac{1}{2}\cdot V_kV_k^\top \right)E + \nfrac{1}{2}\cdot V_kE^\top V_k\right)\widehat{T}^{j-1}_ke_1}  \nonumber\\
  & \label{eq:approx-poly-err1} \qquad + \norm{\sum_{j=1}^t
    B^{t-j}\left(\alpha_{k,k+1}v_{k+1}e^\top
      _{k+1}\right)\widehat{T}^{j-1}_ke_1}\ .
\end{align}

\noindent We can bound the first term in Equation \eqref{eq:approx-poly-err1} as follows.
\begin{align}
  & \sum_{j=1}^t\norm{ B^{t-j}\left(\left(I-\nfrac{1}{2}\cdot
        V_kV_k^\top \right)E + \nfrac{1}{2}\cdot V_kE^\top
      V_k\right)\widehat{T}^{j-1}_ke_1} \le
  \sum_{j=1}^t\norm{\left(I-\nfrac{1}{2}\cdot V_kV_k^\top \right)E + \nfrac{1}{2}\cdot V_kE^\top V_k}\norm{\widehat{T}_k}^{j-1} \nonumber\\
  & \qquad \qquad \qquad \qquad \qquad \qquad \qquad \qquad \qquad \qquad  \qquad\qquad \text{(Using $\norm{B} \le 1$)} \nonumber\\
  & \qquad \qquad \qquad \qquad \qquad \qquad  \le \left(\sum_{j=1}^t (1+\epsilon_1\sqrt{k+1})^{j-1} \right) \norm{\left(I-\nfrac{1}{2}\cdot V_kV_k^\top \right)E + \nfrac{1}{2}\cdot V_kE^\top V_k} \nonumber\\
  & \qquad \qquad \qquad \qquad\qquad \qquad \qquad \qquad \qquad \qquad  \qquad\qquad \text{(Using $\norm{\widehat{T}_k} \le 1+\epsilon_1 \sqrt{k+1}$)} \nonumber\\
  & \qquad \qquad \qquad \qquad \qquad \qquad  \le t(1+\epsilon_1\sqrt{k+1})^{t-1} \left(\norm{\left(I-\nfrac{1}{2}\cdot V_kV_k^\top \right)}\norm{E} + \nfrac{1}{2}\cdot \norm{V_k}\norm{E^\top }\norm{V_k}\right) \nonumber\\
  & \qquad \qquad \qquad \qquad \qquad \qquad \le
  2t\epsilon_1\sqrt{k+1}(1+\epsilon_1\sqrt{k+1})^{t-1}.
\label{eq:approx-poly-err2}
\end{align}

\noindent The second term in Equation \eqref{eq:approx-poly-err1} can be bounded as follows.
\begin{align}
  & \norm{\sum_{j=1}^t B^{t-j}\left(\alpha_{k,k+1}v_{k+1}e^\top _{k+1}\right)\widehat{T}^{j-1}_ke_1} \le |\alpha_{k,k+1}|\sum_{j=1}^t \norm{B}^{t-j}\norm{v_{k+1}e^\top _{k+1}\widehat{T}^{j-1}_ke_1} \nonumber\\
  & \qquad \qquad \le (1+\epsilon_1)\sum_{j=1}^t \norm{B}^{t-j}\norm{v_{k+1}e^\top _{k+1}\left(T_k + \nfrac{1}{2}\cdot (V_k^\top E - E^\top V_k) \right)^{j-1}e_1} \nonumber\\
  &  \qquad \qquad\qquad \qquad \text{(Using $|\alpha_{k,k+1}| \le \norm{w_k} \le \norm{Bv_k} + \epsilon_1 \le 1+\epsilon_1$ and \eqref{eq:Tk-hat2})} \nonumber\\
  & \qquad \qquad\le (1+\epsilon_1)\sum_{j=1}^t \norm{B}^{t-j}\norm{v_{k+1}e^\top _{k+1}\left(\left(T_k + \nfrac{1}{2}\cdot (V_k^\top E - E^\top V_k) \right)^{j-1} - T_k^{j-1}\right)e_1} \nonumber\\
  & \qquad \qquad\qquad \qquad \text{(Using $e^\top _{k+1}T_k^re_1 = 0$ for $r < k$ as $T_k$ is Upper Hessenberg)} \nonumber\\
  &  \qquad \qquad\le (1+\epsilon_1)\sum_{j=1}^t \norm{v_{k+1}e^\top _{k+1}}\left(\sum_{i=1}^{j-1} \binom{j-1}{i} \norm{\nfrac{1}{2}\cdot (V_k^\top E - E^\top V_k)}^{i}\norm{T_k}^{j-1-i}\right) \nonumber\\
  & \qquad \qquad\qquad \qquad \text{(Using sub-multiplicity of $\norm{\cdot}$ and $\norm{B} \le 1$)} \nonumber\\
  & \qquad \qquad\le (1+\epsilon_1)\sum_{j=1}^t \left(\sum_{i=1}^{j-1} \binom{j-1}{i} (\epsilon_1\sqrt{k+1})^i (1+\epsilon_1\sqrt{k+1})^{j-1-i}\right) \nonumber\\
  & \qquad \qquad \qquad \qquad\text{(Using $\norm{v_{k+1}e^\top _{k+1}} \le 1,\norm{T_k} \le (1+\epsilon_1\sqrt{k+1})$} \nonumber\\
  & \qquad \qquad \qquad \qquad \qquad \qquad \qquad \text{ and $\norm{\nfrac{1}{2}\cdot (V_k^\top E - E^\top V_k)} \le \epsilon_1\sqrt{k+1}$)} \nonumber\\
  & \qquad \qquad\le (1+\epsilon_1)\sum_{j=1}^t ((1+2\epsilon_1\sqrt{k+1})^{j-1} -1) \nonumber\\
  & \qquad \qquad\le t (1+\epsilon_1)
  ((1+2\epsilon_1\sqrt{k+1})^{t-1} -1). \label{eq:approx-poly-err3}
\end{align}

\noindent Combining Equations \eqref{eq:approx-poly-err1},\eqref{eq:approx-poly-err2} and \eqref{eq:approx-poly-err3}, we get,
\begin{align*}
\norm{B^tv_0 - V_k\widehat{T}^t_ke_1} & \le 2t\epsilon_1\sqrt{k+1}(1+\epsilon_1\sqrt{k+1})^{t-1} +  t (1+\epsilon_1) ((1+2\epsilon_1\sqrt{k+1})^{t-1} -1)  \\
& \le 2t\epsilon_1\sqrt{k+1}e^{\epsilon_1(t-1)\sqrt{k+1}}+ te^{\epsilon_1}(e^{2 \epsilon_1(t-1)\sqrt{k+1}} - 1) \qquad \text{(Using $1+x \le e^x)$} \\
& \le 2t\epsilon_1\sqrt{k+1}e^{\epsilon_1(t-1)\sqrt{k+1}}+ t(e^{2 \epsilon_1t\sqrt{k+1}} - 1) \\
& \le \nfrac{\epsilon_2}{4}\cdot e^{\nicefrac{\epsilon_2}{8(k+1)}} + k(e^{\nicefrac{\epsilon_2}{4(k+1)}}-1) \qquad \text{(Using $\epsilon_1 \le \nfrac{\epsilon_2}{8(k+1)^{\nicefrac{5}{2}}}$ and $t \le k$)} \\
& \le \nfrac{\epsilon_2}{4}\cdot e^{\nicefrac{\epsilon_2}{8(k+1)}} + k\nfrac{\epsilon_2}{4(k+1)}\cdot\left(1+\nfrac{\epsilon_2}{4(k+1)}\right) \qquad \text{(Using $e^x \le 1+ x + x^2$ for $0 \le x \le 1$)} \\
& \le \nfrac{\epsilon_2}{2} + \nfrac{\epsilon_2}{2} \le \epsilon_2\qquad \text{(Using $\epsilon_2 \le 1$ and $k \ge 0$)}.
\end{align*}

\noindent This proves the lemma. 
\end{proof}

\begin{corollary}[Corollary~\ref{cor:approx-poly-err-l1} Restated]
  For any polynomial $p$ of degree at most $k$, if $\epsilon_1 \le
  \epsilon_2 / (2(k+1)^{\nicefrac{3}{2}})$ and $\epsilon_2 \le 1$,
\[\norm{p(B) v_0 - V_kp(\widehat{T}_k)e_1} \le \epsilon_2 \norm{p}_1.\]
\end{corollary}
\begin{proof}
Suppose $p(x)$ is the polynomial $\sum_{t = 0}^k a_t\cdot x^t,$ 
\begin{align*}
\norm{p(B) v_0 - V_kp(\widehat{T}_k)e_1} & = \norm{\sum_{t = 0}^k
  a_t\cdot B^t v_0 - V_k\sum_{t = 0}^k a_t\cdot \widehat{T}^t_k e_1} \\
& \le \sum_{t = 0}^k |a_t|\cdot \norm{B^t v_0 - V_k \widehat{T}_k^t
  e_1}  \le \epsilon_2 \sum_{t = 0}^k |a_t| = \epsilon_2 \norm{p}_1\ ,
\end{align*}
where the last inequality follows from the previous lemma as
$\epsilon_2,\epsilon_1$ satisfy the required conditions.
\end{proof}

\begin{lemma}[Lemma \ref{lem:min-poly-expv1} Restated]
\label{lem:min-poly-expv1:restated}
Let $V_k$ be the orthonormal basis and $\widehat{T}_k$ be the matrix of coefficients generated by the above procedure. Let $f$ be any function such that $f(B)$ and $f(T_k)$ are defined. Then,
\begin{align}
\norm{f(B)v_0 - V_kf(\widehat{T}_k)e_1} \le \min_{p \in \Sigma_k}
\left(\epsilon_2 \norm{p}_1 + \max_{\lambda \in \Lambda(B)}
  |r_k(\lambda)| +  \max_{\lambda \in \Lambda(\widehat{T}_k)}
  |r_k(\lambda)| \right).
\end{align}
\end{lemma}
\begin{proof}
Let $p$ be any degree $k$ polynomial. Let $r_k \defeq f-p$. We 
express $f$ as $p+r_k$ and use the previous lemma to bound the error
in approximating $p(B)v_0$ by $V_kf(\widehat{T}_k)e_1.$ 
\begin{align*}
\norm{f(B)v_0 - V_kf(\widehat{T}_k)e_1} & \le \norm{p(B)V_ke_1 - V_kp(\widehat{T}_k)e_1} + \norm{V_kr_k(B)e_1 - V_kr_k(\widehat{T}_k)e_1}\\
& \le \epsilon_2 \norm{p}_1 + \norm{V_kr_k(B)e_1} + \norm{V_kr_k(\widehat{T}_k)e_1}\\
& \le \epsilon_2 \norm{p}_1 + \norm{r_k(B)} + \norm{r_k(\widehat{T}_k)} \\
& \le \epsilon_2 \norm{p}_1 + \max_{\lambda \in \Lambda(B)} |r_k(\lambda)| +  \max_{\lambda \in \Lambda(\widehat{T}_k)} |r_k(\lambda)|.
\end{align*}
Minimizing over $p$ gives us our lemma.
\end{proof}

\begin{lemma}[Lemma~\ref{lem:L1-norm} Restated]
Given a polynomial $p$ of degree $k$ such that $p(0) = 0$ and 
\[\sup_{t \in (0,1]} \left\lvert e^{-\nfrac{k}{t}+k} - p(t) \right\rvert = \sup_{x \in [0,\infty)} \left\lvert e^{-x} -
  p\left((1+\nfrac{x}{k})^{-1}\right)\right\rvert \le 1\ ,\]
we must have 
$\norm{p}_1 \le (2k)^{k+1}.$
\end{lemma}
\begin{proof}
We know that $p(0) = 0.$ Interpolating at the $k+1$ points $t
= 0,\nfrac{1}{k},\nfrac{2}{k},\ldots,1,$ we can use Lagrange's interpolation
formula to give,
\[p(x) \equiv \sum_{i=1}^k \frac{\prod_{0 \le j \le k, j\neq i} (x-\nfrac{j}{k})}{\prod_{0 \le j \le k, j\neq i} (\nfrac{i}{k}-\nfrac{j}{k})}p(\nfrac{i}{k}) .\]
The above identity is easily verified by evaluating the expression at
the interpolation points and noting that it is a degree $k$ polynomial
agreeing with $p$ at $k+1$ points. Thus, if we were to write $p(x) =
\sum_{l=1}^k a_l \cdot x^l$ (note that $a_0 = 0$), we can express the coefficients $a_l$ as
follows.
\[a_l = \sum_{i=1}^k \frac{{\displaystyle \sum_{\stackrel{0 \le j_1 <
        \ldots < j_{k-l} \le k }{j_1,\ldots,j_{k-l} \neq i}}
    (-1)^{k-l} \nfrac{j_1}{k}\cdot \ldots \cdot \nfrac{j_{k-l}}{k}}}
{\prod_{0 \le j \le k, j\neq i}
  (\nfrac{i}{k}-\nfrac{j}{k})}p(\nfrac{i}{k}).\] Applying triangle
inequality, and noting that $p(t)$ is a 1-uniform approximation to
$e^{-\nfrac{k}{t}+k}$ for $t \in (0,1],$ we get,
\[|a_l| \le \sum_{i=1}^k \frac{ \binom{k}{k-l}}
{(\nfrac{1}{k})^k} |p(\nfrac{i}{k})| \le \sum_{i=1}^k \frac{ \binom{k}{k-l}}
{(\nfrac{1}{k})^k} \left(e^{-\frac{k(k-i)}{i}} +
  \nfrac{\delta}{2}\right) \le 2\cdot k^{k+1}\binom{k}{k-l}\ .\]
Thus, we can bound the $\ell_1$ norm of $p$ as follows,
\[\norm{p}_1 = \sum_{l=1}^k |a_l| \le \sum_{l=1}^k 2\cdot
k^{k+1}\binom{k}{k-l} \le (2k)^{k+1}\]
\end{proof}

\begin{lemma}[Lemma \ref{lem:approx-extended-interval} Restated]
\label{lem:approx-extended-interval:restated}
For any $\beta \ge 1$, any degree $k$ polynomial $p$ satisfies,
\[\sup_{t \in (0,\beta]} |p(t) - f_{k}(t)| \le \norm{p}_1 \cdot (\beta^k-1) + (f_{k}(\beta)-f_{k}(1))+ \sup_{t \in (0,1]}
|p(t)-f_{k}(t)| \]
\end{lemma}
\begin{proof}
  Given a degree $k$ polynomial $p$ that approximates $f_k$ over
  $(0,1],$ we wish to bound the approximation error over $(0,\beta]$
  for $\beta \ge 1$. We will split the error bound over $(0,1]$ and
  $[1,\beta]$. Since we know that $f_k(1)-p(1)$ is small, we will
  bound the error over $[1,\beta]$ by applying triangle inequality and
  bounding the change in $f_k$ and $p$ over $[1,\beta]$ separately.

  Let $\beta > 0$.  First, let us calculate $\sup_{t \in [1,\beta]}
  |p(t) - f_{k}(t)|.$
\begin{align*}
\sup_{t \in [1,\beta]} |p(t) - f_{k}(t)| & \le \sup_{t
  \in [1,\beta]} (|p(t) - p(1)| +
|p(1)-f_{k}(1)| + |f_{k}(1)-f_{k}(t)|) \\
& \le \sup_{t \in [1,\beta]} (\norm{p}_1 \cdot \max_{0
\le i \le k}|t^i - 1^i| + |p(1)-f_{k}(1)| +
|f_{k}(1)-f_{k}(t)|) \\
& \le \sup_{t \in [1,\beta]} (\norm{p}_1 \cdot \max_{0
\le i \le k}|t^i - 1^i|) + |p(1)-f_k(1)| + \sup_{t \in
[1,\beta]} |f_k(1)-f_k(t)| \\
& \le \norm{p}_1 \cdot (\beta^k-1) +
|p(1)-f_{k}(1)| + \sup_{t \in [1,\beta]}
|f_{k}(1)-f_{k}(t)| \\
& \qquad \qquad \qquad \text{(Since $t \ge 1$ and $t^k$ is increasing for $t \ge 0$)} \\
& \le \norm{p}_1 \cdot (\beta^k-1) + 
|p(1)-f_{k}(1)| + (f_{k}(\beta)-f_{k}(1))\\
& \qquad \qquad \qquad \text{(Since $f_{k}(t)$ is an increasing
  function for $t \ge 0$).} 
\end{align*}
Now, we can bound the error over the whole interval as follows.
\begin{align*}
\sup_{t \in (0,\beta]} |p(t) - f_{k}(t)| & =
\max\{\sup_{t \in (0,1]} |p(t) - f_{k}(t)|,\sup_{t
    \in [1,\beta]} |p(t) - f_{k}(t)|\} \\
& \le \max\{\sup_{t \in (0,1]} |p(t) - f_{k}(t)|,
  \\
& \qquad \norm{p}_1 \cdot (\beta^k-1) + 
|p(1)-f_{k}(1)| + (f_{k}(\beta)-f_{k}(1))\} \\
& = \norm{p}_1 \cdot (\beta^k-1) +
(f_{k}(\beta)-f_{k}(1))+ \sup_{t \in (0,1]}
|p(t)-f_{k}(t)|\ .
\end{align*}
\end{proof}

\section{Uniform Approximations to $e^{-x}$}
\label{sec:poly}
In this section, we discuss uniform approximations to $e^{-x}$ and
prove give a proof of Theorem~\ref{thm:exp-poly-approx} that shows the
existence of polynomials that approximate $e^{-x}$ uniformly over the
interval $[a,b],$ whose degree grows as $\sqrt{b-a}$ and also gives a
lower bound stating that this dependence is necessary.  We restate a
more precise version of the theorem here for completeness.
\begin{theorem}[Uniform Approximation to $e^{-x}$]
\label{thm:exp-poly-approx:restated} $\\$
\begin{itemize}
\item{\bf Upper Bound.} For every $0 \le a < b,$ and a given error parameter $0<\delta \le 1$,
there exists a polynomial ${p}_{a,b,\delta}$ that satisfies,
\[\sup_{x \in [a,b]} |e^{-x}-{p}_{a,b,\delta}(x)| \le
\delta\cdot e^{-a},\] and has degree $O\left(\sqrt{\max\{\log^2
\nfrac{1}{\delta},(b-a) \cdot \log \nfrac{1}{\delta} \}}\cdot
\left(\log \nicefrac{1}{\delta}\right) \cdot \log \log
\nicefrac{1}{\delta} \right)$. 
\item{\bf Lower Bound.} For every $0 \leq a < b $ such that $a + \log_e 4 \le b,$ and $\delta
\in (0,\nfrac{1}{8}],$ any polynomial $p(x)$ that approximates
$e^{-x}$ uniformly over the interval $[a,b]$ up to an error of
$\delta\cdot e^{-a},$ must have degree at least
$\frac{1}{2}\cdot\sqrt{b-a}\ .$
\end{itemize}
\end{theorem}

\subsubsection*{\it Organization}
We first discuss a few preliminaries (Section~\ref{sec:poly:prelims})
and discuss relevant results that were already known and compare our
result to the existing lower bounds
(Section~\ref{sec:poly:previous-results}). Finally, we give a proof of
the upper bound in Theorem~\ref{thm:exp-poly-approx} in
Section~\ref{sec:exp-poly-approx-proof} and of the lower bound in
Section~\ref{sec:sqrt-lower-bound-proof}. Readers familiar with
standard results in approximation theory can skip directly to the
proofs in Section~\ref{sec:exp-poly-approx-proof}
and~\ref{sec:sqrt-lower-bound-proof}.

\subsection{Preliminaries}
\label{sec:poly:prelims}
Given an interval $[a,b],$ we are looking for low-degree polynomials
(or rational functions) that approximate the function $e^{-x}$ in the
$\sup$ norm over the interval.  
\begin{definition}[$\delta$-Approximation]
A function $g$ is called a $\delta$-approximation to a
function $f$ over an interval $\calI,$ if, $\sup_{x \in \calI}
|f(x)-g(x)| \le \delta$. 
\end{definition}
\noindent Such approximations are known as \emph{uniform
  approximations} in approximation theory and have been studied quite
extensively. We will consider both finite and infinite intervals
$\calI$.

\noindent For any positive integer $k,$ let $\Sigma_k$ denote the set
of all degree $k$ polynomials. We also need to define the $\ell_1$
norm of a polynomial.
\begin{definition}[$\ell_1$ Norm of a Polynomial]
  Given a degree $k$ polynomial $p \defeq \sum_{i=0}^k a_i \cdot x^i,$
  the $\ell_1$ norm of $p,$ denoted as $\norm{p}_1$ is defined as
  $\norm{p}_1 = \sum_{i \ge 0}^k |a_i|.$
\end{definition}

\subsection{Known Approximation Results and Discussion.}
\label{sec:poly:previous-results}
Approximating the exponential is a classic question in Approximation
Theory, see \emph{e.g.}~\cite{cheney-book}. We ask the following
question:

\vspace{1mm}
{\bf Question:} {\it Given $\delta \le 1$ and $a < b,$ what is the smallest
degree of a polynomial that is an $\delta\cdot e^{-a}$-approximation
to $e^{-x}$ over the interval $[a,b]$?}
\vspace{1mm}

This qustion has been studied in the following form: Given $\lambda,$
what is the best low degree polynomial (or rational function)
approximation to $e^{\lambda x}$ over $[-1,1]$? In a sense, these
questions are equivalent, as is shown by the following lemma, proved
using a linear shift of variables. A proof is included in
Section~\ref{appendix:proofs:poly}.
\begin{lemma}[Linear Variable Shift for Approximation]
\label{lem:approx-linear-shift}
For any non-negative integer $k$ and $l$ and real numbers $b > a,$
\[ \min_{p_k \in \Sigma_k, q_l \in \Sigma_l}\sup_{t \in [a,b]}
\left|e^{-t} - \frac{p_k(t)}{q_l(t)} \right| = e^{-\frac{b+a}{2}}
\cdot \min_{p_k \in \Sigma_k, q_l \in \Sigma_l}\sup_{x \in [-1,1]}
\left|e^{\frac{(b-a)}{2} x} - \frac{p_k(x)}{q_l(x)} \right|
\]
\end{lemma}

Using the above lemma, we can translate the known results to our
setting. As a starting point, we could approximate $e^{-x}$ by
truncating the Taylor series expansion of the exponential. We state
the approximation achieved in the following lemma. A proof is included
in Section~\ref{appendix:proofs:poly}.
\begin{lemma}[Taylor Approximation]
\label{lem:taylor}
The degree $k$ polynomial obtained by truncating Taylor's expansion of
$e^{-t}$ around the point $\nfrac{b+a}{2}$ is a uniform approximation
to $e^{-t}$ on the interval $[a,b]$ up to an error of
\[e^{-\frac{b+a}{2}} \cdot \sum_{i=k+1}^\infty
\frac{1}{i!}\left(\frac{b-a}{2}\right)^i,\] 
which is smaller than $\delta\cdot e^{-\frac{b+a}{2}}$ for $k \ge
\max\{\frac{e^2(b-a)}{2},\log \nfrac{1}{\delta}\}$
\end{lemma}

\noindent A lower bound is known in the case where the size of the
interval is fixed, \emph{i.e.}, $b-a = O(1).$
\begin{proposition}[Lower Bound for Polynomials over Fixed Interval,
  \cite{poly-lb,saff-rational}]
For any $a,b \in \rea$ such that $b-a$ is \emph{fixed}, as $k$ goes to
infinity, the best approximation achieved by a degree $k$ polynomial
has error
\[(1+o(1)) \frac{1}{(k+1)!}\left(\frac{b-a}{2}\right)^{k+1}\cdot
e^{-\frac{b+a}{2}}.\]
\end{proposition}

In essence, this theorem states that if the \emph{size of the interval
  is fixed}, the polynomials obtained by truncating the Taylor series
expansion achieve asymptotically the least error possible and hence,
the best asymptotic degree for achieving a $\delta\cdot
e^{-\frac{b+a}{2}}$-approximation. In addition,
Saff~\cite{saff-rational} also shows that if, instead of polynomials,
we allow rational functions where the degree of the denominator is a
constant, the degree required for achieving a $\delta\cdot
e^{-\frac{b+a}{2}}$-approximation changes at most by a constant.

These results indicate that tight bounds on the answer to our question
should be already known. In fact, at first thought, the optimality of
the Taylor series polynomials seems to be in contradiction with our
results. However, note the two important differences:
\begin{enumerate}[label=\arabic*.]
\item The error in our theorem is $e^{-a}\cdot \delta,$ whereas, the
  Taylor series approximation involves error $e^{-\frac{b+a}{2}}\cdot
  \delta,$ which is \emph{smaller}, and hence requires larger degree.
\item Moreover, the lower bound applies only when the length of the
  interval $(b-a)$ is constant, in which case, our theorem says that
  the required degree is $\poly(\log \nfrac{1}{\delta}),$ which is
  $\Omega(\log \nfrac{1}{\delta}),$ in accordance with the lower
  bound.
\end{enumerate}

If the length of the interval $[a,b]$ grows unbounded (as is the case
for our applications to the Balanced Separator problem in the previous
sections), the main advantage of using polynomials from
Theorem~\ref{thm:exp-poly-approx:restated} is the improvement in the
degree from linear in $(b-a)$ to $\sqrt{b-a}.$

\subsection{Proof of Upper Bound in Theorem~\ref{thm:exp-poly-approx}}
\label{sec:exp-poly-approx-proof}
In this section, we use Theorem~\ref{thm:rational-approximations} by
Saff, Schonhage and Varga \cite{SSV}, rather, more specifically,
Corollary~\ref{cor:pk-star} to give a proof of the upper bound result
in Theorem~\ref{thm:exp-poly-approx}. We restate
Corollary~\ref{cor:pk-star} for completeness.
\begin{corollary}[Corollary~\ref{cor:pk-star} Restated, \cite{SSV}]
\label{cor:pk-star:restated}
There exists constants $c_1 \ge 1$ and $k_0$ such that, for any integer $k
\ge k_0$, there exists a polynomial $p_k^\star(x)$ of degree $k$ such
that $p_k^\star(0)=0,$ and,
\begin{align}
\label{eq:pk-approx}
\sup_{t \in (0,1]} \left\lvert e^{-\nfrac{k}{t}+k} - p^\star_k(t)
\right\rvert = \sup_{x \in [0,\infty)} \left\lvert e^{-x} -
  p_k^\star\left((1+\nfrac{x}{k})^{-1}\right)\right\rvert \le
c_1k\cdot 2^{-k}\ .
\end{align}
\end{corollary}

Our approach is to compose the polynomial $p_k^\star$ given by
Corollary~\ref{cor:pk-star:restated} with polynomials approximating
$(1+\nfrac{x}{k})^{-1}$ , to construct polynomials approximating
$e^{-x}$. We first show the existence of polynomials approximating
$x^{-1},$ and from these polynomials, we will derive approximations to
$(1+\nfrac{x}{k})^{-1}.$

Our goal is to find a polynomial $q$ of degree $k,$ that minimizes
$\sup_{x \in [a,b]} |q(x)-\nfrac{1}{x}|.$ We slightly modify this
optimization to minimizing $\sup_{x \in [a,b]} |x\cdot q(x)-1|.$ Note
that $x\cdot q(x) - 1$ is a polynomial of degree $k+1$ which evaluates
to $-1$ at $x=0,$ and conversely every polynomial that evaluates to
$-1$ at $0$ can be written as $x\cdot q(x)-1$ for some $q$. So, this
is equivalent to minimizing $\sup_{x \in [a,b]} |q_1(x)|$, for a
degree $k+1$ polynomial $q_1$ such that $q_1(0)=-1$. By scaling and
multiplying by $-1$, this is equivalent to finding a polynomial $q_2,$
that maximizes $q_2(0),$ subject to $\sup_{x \in [a,b]} |q_2(x)| \le
1$. If we shift and scale the interval $[a,b]$ to $[-1,1]$, the
optimal solution to this problem is known to be given by the well
known Chebyshev polynomials. We put all these ideas together to prove
the following lemma. A complete proof is included in
Section~\ref{appendix:proofs:poly}.
\begin{lemma}[Approximating $x^{-1}$]
\label{lem:Chebyshev1}
For every $\epsilon > 0$, $b > a > 0$, there exists a polynomial
$q_{a,b,\epsilon}(x)$ of degree $\ceil{\sqrt{\frac{b}{a}}\log
  \frac{2}{\epsilon}}$ such that $\sup_{x \in [a,b]} |x\cdot
q_{a,b,\epsilon}(x) - 1| \le \epsilon.$
\end{lemma}

\noindent As a simple corollary, we can approximate
$(1+\nfrac{x}{k})^{-1},$ or rather generally, $(1+\nu x)^{-1}$ for
some $\nu > 0,$ by polynomials. A proof is included in
Section~\ref{appendix:proofs:poly}.
\begin{corollary}[Approximating $(1+\nu x)^{-1}$]
  For every $\nu >0,\epsilon > 0$ and $b > a \ge 0$, there exists a
  polynomial $q^\star_{\nu,a,b,\epsilon}(x)$ of degree
  $\ceil{\sqrt{\frac{1+\nu b}{1+\nu a}}\log \frac{2}{\epsilon}}$ such
  that
$\sup_{x \in [a,b]} |(1+\nu x)\cdot q^\star_{\nu,a,b,\epsilon}(x) - 1| \le
\epsilon.$
\label{cor:Chebyshev1-shift}
\end{corollary}
The above corollary implies that the expression $(1+\nu x)\cdot
q^\star$ is within $1\pm \varepsilon$ on $[a,b]$. If $\varepsilon$ is
small, for a small positive integer $t$, $[(1+\nu x)\cdot q^\star]^t$
should be at most $1\pm O(t\varepsilon).$ The following lemma, proved
using the binomial theorem proves this formally. A proof is included
in Section~\ref{appendix:proofs:poly}.
\begin{lemma}[Approximating $(1+\nu x)^{-t}$]
\label{lem:Chebyshev-power}
For all real $\epsilon > 0$, $b > a \ge 0$ and positive integer $t$; if $t\epsilon \le 1$, then,
\[\sup_{x \in [a,b]} |((1+\nu x) \cdot q^\star_{\nu,a,b,\epsilon}(x))^t - 1| \le
2t\epsilon,\]
where $q^\star_{\nu,a,b,\epsilon}$ is the polynomial given by
Corollary~\ref{cor:Chebyshev1-shift}.
\end{lemma}

\noindent Since $q^\star$ is an approximation to
$(1+\nfrac{x}{k})^{-1},$ in order to bound the error for the
composition $p_k^\star(q^\star),$ we need to bound how the value of
the polynomial $p^\star_k$ changes on small perturbations in the
input. We will use the following crude bound in terms of the $\ell_1$
norm of the polynomial.
\begin{lemma}[Error in Polynomial]
\label{lem:poly-error-L1}
For any polynomial $p$ of degree $k,$ and any $x,y \in \rea,$
$|p(x)-p(y)| \le \norm{p}_1 \cdot \max_{0 \le i \le k} |x^i - y^i|.$
\end{lemma}

\noindent In order to utilize the above lemma, we will need a bound on the $\ell_1$
norm of $p_k^\star,$ which is provided by Lemma~\ref{lem:L1-norm},
that bounds the $\ell_1$ norm of any polynomial in
$(1+\nfrac{x}{k})^{-1}$ that approximates the exponential function and
has no constant term. We restate the lemma here for completeness.
\begin{lemma}[$\ell_1$-norm Bound. Lemma~\ref{lem:L1-norm} Restated]
\label{lem:L1-norm:restated}
Given a polynomial $p$ of degree $k$ such that $p(0) = 0$ and 
\[\sup_{t \in (0,1]} \left\lvert e^{-\nfrac{k}{t}+k} - p(t) \right\rvert = \sup_{x \in [0,\infty)} \left\lvert e^{-x} -
  p\left((1+\nfrac{x}{k})^{-1}\right)\right\rvert \le 1\ ,\] we must
have $\norm{p}_1 \le (2k)^{k+1}.$
\end{lemma}

We can now analyze the error in approximating $e^{-x}$ by the
polynomial $p_k^\star(q^\star)$ and give a proof for
Theorem~\ref{thm:exp-poly-approx}.
\begin{proof}
Given $\delta \le 1$, let $k=\max\{k_0,\log_2 \nicefrac{4c_1}{\delta}
+ 2 \log_2 \log_2 \nicefrac{4c_1}{\delta}\} = O\left(\log
\nicefrac{1}{\delta}\right),$ where $k_0,c_1$ are the constants given
by Corollary~\ref{cor:pk-star:restated}. Moreover, $p_k^\star$
is the degree $k$ polynomial given by Corollary~\ref{cor:pk-star:restated},
which gives, $p_k^\star(0)=0,$ and,
\begin{align}
\sup_{x \in [0,\infty)} \left|e^{-x} -
  p^\star_k\left((1+\nicefrac{x}{k})^{-1}\right) \right| & \le
\frac{\delta}{4}\cdot \frac{\log_2 \nicefrac{4c_1}{\delta} + 2 \log_2 \log_2
\nicefrac{4c_1}{\delta}}{(\log_2 \nicefrac{4c_1}{\delta})^2} \nonumber\\ & \le
\frac{\delta}{4} \cdot \frac{1}{\log_2 \nicefrac{4c_1}{\delta}} \left(1+ 2\cdot \frac{\log_2 \log_2
\nicefrac{4c_1}{\delta}}{\log_2 \nicefrac{4c_1}{\delta}} \right) \le
\frac{\delta}{4} \cdot \frac{1}{2} \cdot 3 \le \frac{\delta}{2}\ ,
\label{eq:rational-error1}
\end{align}
where the last inequality uses $\delta \le 1 \le c_1$ and $\log_2 x \le x,
\forall x \ge 0$. Thus, we can use Lemma~\ref{lem:L1-norm:restated} to conclude
that $\norm{p_k^\star}_1 \le (2k)^{k+1}.$

Let $\nu \defeq \nicefrac{1}{k}$. Define $\varepsilon$ as
$\varepsilon \defeq \frac{\delta}{2(2k)^{k+2}}.$
Let $p_{a,b,\delta}(x) \defeq e^{-a}\cdot p^\star_{k} \left(
q^\star_{\nu,0,b-a,\epsilon}(x-a)\right),$ where
$q^\star_{\nu,0,b-a,\epsilon}$ is the polynomial of degree
$\ceil{\sqrt{1+\nu(b-a)}\log \frac{2}{\epsilon}}$ given by
Corollary~\ref{cor:Chebyshev1-shift}. Observe that $p_{a,b,\delta}(x)$ is a
polynomial of degree that is the product of the degrees of $p_k^\star$
and $q^\star_{\nu,0,b-a,\epsilon}$, \emph{i.e.},
$k\ceil{\sqrt{1+\nu(b-a)}\log \frac{2}{\epsilon}}$. Also note that
$k\epsilon < 1$ and hence we can use
Lemma~\ref{lem:Chebyshev-power}. We show that $p_{a,b,\delta}$
is a uniform $\delta$-approximation to $e^{-x}$ on the interval $[a,b].$
\begin{align*}
& \sup_{x \in [a,b]} |e^{-x}-p_{a,b,\delta}(x)| =
e^{-a} \cdot \sup_{x \in [a,b]}
\left|e^{-(x-a)}-e^{a}\cdot p_{a,b,\delta}(x)\right|
= e^{-a} \cdot \sup_{y \in [0,b-a]}
\left|e^{-y}-e^{a}\cdot p_{a,b,\delta}(y+a)\right| \\
&  \qquad  \stackrel{\text{by def}}{=} e^{-a} \cdot \sup_{y \in [0,b-a]}
\left|e^{-y}- p^\star_{k}
\left(
  q^\star_{\nu,0,b-a,\epsilon}(y)\right)\right| \\
& \qquad \stackrel{\Delta-\text{ineq.}}{\leq} e^{-a} \cdot \sup_{y \in
  [0,b-a]} \left(\left|e^{-y}-p^\star_{k}
\left((1+\nu y)^{-1}\right)\right| + \left|p^\star_{k}
\left((1+\nu y)^{-1}\right) - p^\star_{k}
\left(q^\star_{\nu,0,b-a,\epsilon}(y)\right)\right| \right)\\
& \qquad \le e^{-a} \cdot \sup_{y \in
  [0,b-a]} \left|e^{-y}-p^\star_{k}
\left((1+\nu y)^{-1}\right)\right| 
+  e^{-a} \cdot \sup_{y \in
  [0,b-a]} \left|p^\star_{k}
\left((1+\nu y)^{-1}\right) - p^\star_{k}
\left(q^\star_{\nu,0,b-a,\epsilon}(y)\right)\right| \\
& \qquad \stackrel{Lem.~\ref{lem:poly-error-L1}}{\le} e^{-a} \cdot \sup_{y \in
  [0,\infty)} \left|e^{-y}-p^\star_{k}
\left((1+\nu y)^{-1}\right)\right| \\
& \qquad \qquad  +  e^{-a} \cdot
\norm{p^\star_{k}}_1\cdot \max_{0\le i \le k}\sup_{y \in
  [0,b-a]} \left|
(1+\nu y)^{-i} - 
\left(q^\star_{\nu,0,b-a,\epsilon}(y)\right)^i\right| \\
& \qquad \stackrel{Eq.~\eqref{eq:rational-error1}}{\le} e^{-a} \cdot \frac{\delta}{2}  +  e^{-a} \cdot
\norm{p^\star_{k}}_1\cdot \max_{0\le i \le k} \sup_{y \in
  [0,b-a]} (1+\nu y)^{-i} \left|
1 - \left((1+\nu y) \cdot
  q^\star_{\nu,0,b-a,\epsilon}(y)\right)^i \right| \\
& \qquad \stackrel{Lem.~\ref{lem:Chebyshev-power},\ref{lem:L1-norm:restated}}{\le} e^{-a} \cdot \frac{\delta}{2}+
2k\epsilon\cdot  e^{-a} \cdot (2k)^{k+1} \le \delta\cdot e^{-a}
\end{align*}

The degree of the polynomial $p_{a,b,\delta}$ is 
\begin{align*}
k\ceil{\sqrt{1+\nu(b-a)}\log \frac{2}{\epsilon}} & =
O\left(\sqrt{k^2 + k
      (b-a)}\cdot \left(k\log k + \log \nfrac{1}{\delta}\right)\right)\\
 & = O\left(\sqrt{\max\{\log^2 \nfrac{1}{\delta},\log \nfrac{1}{\delta}
      (b-a)\}}\cdot \left(\log
    \nicefrac{1}{\delta}\right) \cdot \log \log \nicefrac{1}{\delta}
\right)
\end{align*}
\end{proof}

\subsection{Proof of Lower Bound in Theorem~\ref{thm:exp-poly-approx}}
\label{sec:sqrt-lower-bound-proof}
In this section, we will use the following well known theorem of
Markov from approximation theory to give a proof of the lower bound
result in Theorem~\ref{thm:exp-poly-approx}.
\begin{theorem}[Markov, See \cite{cheney-book}]
\label{thm:markov}
Let $p \from \rea \to \rea$ be a univariate polynomial of degree $d$
such that any real number $a_1 \le x \le a_2,$ satisfies $b_1 \le p(x)
\le b_2.$ Then, for all $a_1 \le x \le a_2,$ the derivative of $p$
satisfies $|p^\prime(x)| \le d^2\cdot\frac{b_2 - b_1}{a_2 - a_1}.$
\end{theorem}

The idea is to first use uniform approximation bound to bound the
value of the polynomial within the interval of approximation. Next, we
use the approximation bound and the \emph{Mean Value theorem} to show that there
must exist a point $t$ in the interval where $|p^\prime(t)|$ is
large. We plug both these bounds into Markov's theorem to deduce our
lower bound.
\begin{proof}
  Suppose $p$ is a degree $k$ polynomial that is a uniform
  approximation to $e^{-x}$ over the interval $[a,b]$ up to an error
  of $\delta\cdot e^{-a}$. For any $x \in [a,b],$ this bounds the
  values $p$ can take at $x.$ Since $p$ is a uniform approximation to
  $e^{-x}$ over $[a,b]$ up to an error of $\delta\cdot e^{-a},$ we
  know that for all $x \in [a,b],$ $e^{-x} - \delta\cdot e^{-a} \le
  p(x) \le e^{-x} + \delta\cdot e^{-a}.$ Thus, $\max_{x \in [a,b]}
  p(x) \le e^{-a} +\delta\cdot e^{-a}$ and $\min_{x \in [a,b]} p(x)
  \ge e^{-b} - \delta\cdot e^{-a}.$

  Assume that $\delta \le \nfrac{1}{8},$ and $b \ge a + \log_e 4 \ge a
  + \log_e \nfrac{2}{(1-4\delta).}$ Applying the \emph{Mean Value
    theorem} on the interval $[a,a+\log_e \nfrac{2}{(1-4\delta)}],$ we
  know that there exists $t \in [a,a+\log_e \nfrac{2}{(1-4\delta)} ],$
  such that,
\begin{align*}
  |p^\prime(t)| = \left\lvert \frac{p(a+\log_e \nfrac{2}{(1-4\delta)}) -
      p(a)}{\log_e \nfrac{2}{(1-4\delta)}} \right \rvert &\ge
  \frac{(e^{-a} -\delta\cdot e^{-a}) - (e^{-a -\log_e
      \nfrac{2}{1-4\delta}} + \delta\cdot e^{-a}) }{\log_e
    \nfrac{2}{(1-4\delta)}} \\
  & \ge e^{-a} \frac{1-2\delta - \frac{(1-4\delta)}{2}}{\log_e
      \nfrac{2}{(1-4\delta)}} =  e^{-a} \frac{1}{2\log_e
      \nfrac{2}{(1-4\delta)}}
\end{align*}
We plug this in Markov's theorem (Theorem~\ref{thm:markov}) stated above to deduce,
\[e^{-a} \frac{1}{2\log_e
      \nfrac{2}{(1-4\delta)}} \le k^2 \frac{(e^{-a}+\delta\cdot
      e^{-a}) - (e^{-b}-\delta\cdot e^{-a})}{b-a} \le k^2\cdot e^{-a}
    \cdot \frac{1+2\delta}{b-a}\ .\]
Rearranging, we get,
\[k \ge \sqrt{\frac{b-a}{2\cdot(1+2\delta)\cdot\log_e
    \nfrac{2}{(1-4\delta)}} } \ge \sqrt{\frac{b-a}{2\cdot \nfrac{5}{4}\cdot\log_e
    4} }  \ge \frac{1}{2}\cdot\sqrt{b-a},\]
where the second inequality uses $\delta \le \nfrac{1}{8}.$
\end{proof}

\subsection{Remaining Proofs}
\label{appendix:proofs:poly}
\begin{lemma}[Lemma \ref{lem:approx-linear-shift} Restated]
\label{lem:approx-linear-shift:restated}
For any non-negative integer $k$ and $l$ and real numbers $b > a,$
\[ \min_{p_k \in \Sigma_k, q_l \in \Sigma_l}\sup_{t \in [a,b]}
\left|e^{-t} - \frac{p_k(t)}{q_l(t)} \right| = e^{-\frac{b+a}{2}}
\cdot \min_{p_k \in \Sigma_k, q_l \in \Sigma_l}\sup_{x \in [-1,1]}
\left|e^{\frac{(b-a)}{2} x} - \frac{p_k(x)}{q_l(x)} \right|
\]
\end{lemma}
\begin{proof}
Using the substitution $t \defeq \frac{(b+a)}{2}-\frac{(b-a)}{2}x,$
\begin{align*}
  \min_{p_k \in \Sigma_k, q_l \in \Sigma_l}\sup_{t \in [a,b]} \left|e^{-t} - \frac{p_k(t)}{q_l(t)} \right| & =
  \min_{p_k \in \Sigma_k,q_l \in \Sigma_l}\sup_{x \in [1,-1]}
  \left|e^{-\frac{(b+a)}{2}+\frac{(b-a)}{2}x} -
    \frac{p_k\left(\nfrac{(b+a)}{2}-\nfrac{(b-a)}{2}x\right)}{q_l\left(\nfrac{(b+a)}{2}-\nfrac{(b-a)}{2}x\right)} \right| \\
  & = e^{-\frac{b+a}{2}} \cdot \min_{p^\prime_k \in \Sigma_k,
    q^\prime_l \in \Sigma_l}\sup_{x \in
    [-1,1]} \left|e^{\frac{(b-a)}{2} x} - \frac{p^\prime_k(x)}{q_l^\prime(x)} \right|
\end{align*}
\end{proof}

\begin{lemma}[Lemma \ref{lem:taylor} Restated]
\label{lem:taylor:restated}
The degree $k$ polynomial obtained by truncating Taylor's expansion of
$e^{-t}$ around the point $\nfrac{b+a}{2}$ is a uniform approximation
to $e^{-t}$ on the interval $[a,b]$ up to an error of
\[e^{-\frac{b+a}{2}} \cdot \sum_{i=k+1}^\infty
\frac{1}{i!}\left(\frac{(b-a)}{2}\right)^i,\]
which is smaller than $\delta$ for $k \ge \max\{\frac{e^2(b-a)}{2},\log \nfrac{1}{\delta}\}$
\end{lemma}
\begin{proof}
Let $q_k(t)$ be the degree $k$ Taylor approximation of the function
$e^{-t}$ around the point $\nfrac{(b+a)}{2},$ \emph{i.e.}, $q_k(t)
\defeq e^{-\frac{(b+a)}{2}}\sum_{i=0}^k \frac{1}{i!} \left(t-\frac{(b+a)}{2}\right)^i.$ 
\begin{align*}
\sup_{t \in [a,b]} |e^{-t} - q_k(t)| & = \sup_{t \in [a,b]}
e^{-\frac{b+a}{2}} \cdot \left|\sum_{i=k+1}^\infty
\frac{1}{i!} \left(t-\frac{(b+a)}{2}\right)^i\right| =
e^{-\frac{b+a}{2}} \cdot \sum_{i=k+1}^\infty
\frac{1}{i!}\left(\frac{(b-a)}{2}\right)^i
\end{align*}
Using the inequality $i! > \left(\frac{i}{e}\right)^i,$ for all $i,$
and assuming $k \ge \frac{e^2(b-a)}{2},$ we get,
\[ \sum_{i=k+1}^\infty
\frac{1}{i!}\left(\frac{(b-a)}{2}\right)^i \le \sum_{i=k+1}^\infty
\left(\frac{e(b-a)}{2i}\right)^i \le \sum_{i=k+1}^\infty
e^{-i}= \frac{1}{e-1} e^{-k}\ ,\]
which is smaller than $\delta$ for $k \ge \log \nfrac{1}{\delta}$.
\end{proof}

\begin{lemma}[Lemma~\ref{lem:Chebyshev1} Restated]
\label{lem:Chebyshev1:restated}
For every $\epsilon > 0$, $b > a > 0$, there exists a polynomial $q_{a,b,\epsilon}(x)$ of degree
$\ceil{\sqrt{\frac{b}{a}}\log \frac{2}{\epsilon}}$ such that 
\[\sup_{x \in [a,b]} |x\cdot q_{a,b,\epsilon}(x) - 1| \le \epsilon.\]
\end{lemma}
\begin{proof}
If $T_{k+1}(x)$ denotes the degree $k+1$ Chebyshev polynomial, consider the
function,
\[ q_{a,b,\epsilon}(x) \defeq \frac{1}{x}\left(1 -
  \frac{T_{k+1}\left(\frac{b+a-2x}{b-a}\right)}{T_{k+1}\left(\frac{b+a}{b-a}\right)}
\right). \]

First, we need to prove that the above expression is a
polynomial. Clearly $1 -
  \frac{T_{k+1}\left(\frac{b+a-2x}{b-a}\right)}{T_{k+1}\left(\frac{b+a}{b-a}\right)}$
    is a polynomial and evaluates to 0 at $x=0$. Thus, it must have
    $x$ as a factor. Thus $q_{a,b,\epsilon}$ is a polynomial of degree $k
    $. Let $\kappa = \nicefrac{b}{a}$ and note that $\kappa > 1$. Thus,
\begin{align*}
\sup_{x \in [a,b]} |x\cdot q_{a,b,\epsilon}(x) - 1| &  = \sup_{x \in [a,b]} \left|
  \frac{T_{k+1}\left(\frac{b+a-2x}{b-a}\right)}{T_{k+1}\left(\frac{b+a}{b-a}\right)}
\right| \\
& \le T_{k+1}\left(\frac{b+a}{b-a}\right)^{-1} \qquad \qquad
\text{(Since $|T_{k+1}(y)| \le 1$ for $|y| \le 1$)} \\
& =
2\left( \left(\frac{\sqrt{\kappa}+1}{\sqrt{\kappa}-1}\right)^{k+1}
+ \left(\frac{\sqrt{\kappa}-1}{\sqrt{\kappa}+1}\right)^{k+1}
\right)^{-1} \qquad \text{(By def.)}\\
& \le 
2\left(\frac{\sqrt{\kappa}+1}{\sqrt{\kappa}-1}\right)^{-k-1} \qquad\qquad
\text{(Each term is positive since $\sqrt{\kappa} > 1$)}\\
& =
2\left(\frac{1 - \nfrac{1}{\sqrt{\kappa}}}{1 +
    \nfrac{1}{\sqrt{\kappa}}}\right)^{k+1} \\
& \le 2\cdot \left(1 - \nfrac{1}{\sqrt{\kappa}}\right)^{k+1} \le
2\cdot 
e^{-\nfrac{(k+1)}{\sqrt{\kappa}}} \le \epsilon,
\end{align*}
for $k = \ceil{\sqrt{\kappa}\log \frac{2}{\epsilon}}$. The
  first inequality follows from the fact that $|T_{k+1}(x)| \le 1$ for all
  $|x| \le 1$.
\end{proof}

\begin{corollary}[Corollary~\ref{cor:Chebyshev1-shift} Restated]
For every $\nu >0,\epsilon > 0$ and $b > a \ge 0$, there exists a polynomial $q^\star_{\nu,a,b,\epsilon}(x)$ of degree
$\ceil{\sqrt{\frac{1+\nu b}{1+\nu a}}\log \frac{2}{\epsilon}}$ such that 
\[\sup_{x \in [a,b]} |(1+\nu x)\cdot q^\star_{\nu,a,b,\epsilon}(x) - 1| \le
\epsilon.\]
\label{cor:Chebyshev1-shift:restated}
\end{corollary}
\begin{proof}
Consider the polynomial $q^\star_{\nu,a,b,\epsilon}(x) \defeq q_{1+\nu a,1+\nu
  b,\epsilon}\left(1+\nu x\right),$ where $q_{1+\nu a,1+\nu
  b,\epsilon}$ is given by the previous lemma.
\begin{align*}
\sup_{x \in [a,b]}\left|(1+\nu x)\cdot q^\star_{\nu,a,b,\epsilon}(x) -
  1\right| & =
\sup_{x \in [a,b]}\left|(1+\nu x)\cdot q_{1+\nu a,1+\nu
  b,\epsilon}\left(1+\nu x\right) - 1\right| \\
& \stackrel{t \defeq 1+\nu x}{=} \sup_{t \in [1+\nu
a,1 +\nu b]}\left|t\cdot q_{1+\nu a,1+\nu
  b,\epsilon}\left(t\right) - 1\right| \stackrel{Lem.~\ref{lem:Chebyshev1}}{\le} \epsilon.
\end{align*}
Since $1+\nu x$ is a linear transformation, the degree of
$q^\star_{\nu,a,b,\epsilon}$ is the same as that of $q_{1+\nu a,1+\nu
  b,\epsilon}$ , which is, $\ceil{\sqrt{\frac{1+\nu b}{1+\nu
      a}}\log \frac{2}{\epsilon}}.$
\end{proof}

\begin{lemma}[Lemma~\ref{lem:Chebyshev-power} Restated]
\label{lem:Chebyshev-power:restated}
For all real $\epsilon > 0$, $b > a \ge 0$ and positive integer $t$; if $t\epsilon \le 1$, then,
\[\sup_{x \in [a,b]} |((1+\nu x) \cdot q^\star_{\nu,a,b,\epsilon}(x))^t - 1| \le
2t\epsilon,\]
where $q^\star_{\nu,a,b,\epsilon}$ is the polynomial given by Corollary~\ref{cor:Chebyshev1-shift}.
\end{lemma}
\begin{proof}
We write the expression $(1+\nu x) \cdot
q^\star_{\nu,a,b,\epsilon}(x)$ as 1 plus an error term and then use
the Binomial Theorem to expand the $t^\text{th}$ power.
\begin{align*}
\sup_{x \in [a,b]} |((1+\nu x)\cdot q^\star_{\nu,a,b,\epsilon}(x))^t - 1|  & =
\sup_{x \in [a,b]} \left| \left(1 - \left[1
      - (1+\nu x)\cdot q^\star_{\nu,a,b,\epsilon}(x)\right]\right)^t - 1\right| \\
& = \sup_{x \in [a,b]} \left| \sum_{i=1}^t \binom{t}{i}\left(1- (1+\nu
    x) \cdot q^\star_{\nu,a,b,\epsilon}(x)\right)^{i}\right| \\
& \le \sup_{x \in [a,b]} \sum_{i=1}^t \binom{t}{i}\left|1 - (1+\nu x)
  \cdot  q^\star_{\nu,a,b,\epsilon}(x)\right|^{i} \\
& \le \sum_{i=1}^t \binom{t}{i} \sup_{x \in [a,b]} \left|1 - (1+\nu x)
  \cdot  q^\star_{\nu,a,b,\epsilon}(x)\right|^{i} \\
& \stackrel{Cor.~\ref{cor:Chebyshev1-shift}}{\le} \sum_{i=1}^t \binom{t}{i}\epsilon^{i} =
(1+\epsilon)^t-1\\
& \le \exp(t\epsilon)-1 \le 1+t\epsilon + (t\epsilon)^2 -1 \le 2t\epsilon,
\end{align*}
where the second last inequality uses $e^x \le 1+x+x^2$ for $x \in [0,1].$
\footnote{For $x \in [0,1]$, $e^x = \sum_{i\ge 0} \frac{x^i}{i!} = 1 + x
  + x^2\left(\frac{1}{2!} + \frac{x}{3!} +\ldots \right) \le 1 + x
  + x^2\left(\frac{1}{2} + \frac{x}{2^2} + \frac{x}{2^3}+\ldots
  \right) \le 1+ x + x^2$}
\end{proof}

\begin{lemma}[Lemma~\ref{lem:poly-error-L1} Restated]
\label{lem:poly-error-L1:restated}
For any polynomial $p$ of degree $k,$ and any $x,y \in \rea$
\[|p(x)-p(y)| \le \norm{p}_1 \cdot \max_{0 \le i \le k} |x^i - y^i|\]
\end{lemma}
\begin{proof}
Suppose $p(t)$ is the polynomial $\sum_{i=0}^k a_i\cdot  t^i,$ where
$a_i \in \rea.$ Then,
\begin{align*}
|p(x)-p(y)| & = \left|\sum_{i=0}^k a_i \cdot  x^i - \sum_{i=0}^k a_i \cdot y^i
\right|  \le \sum_{i=0}^k |a_i| |x^i - y^i| \\
& \le \left(\sum_{i=0}^k |a_i|\right) \max_{0 \le i \le k} |x^i - y^i|
= \norm{p}_1 \cdot \max_{0 \le i \le k} |x^i - y^i|
\end{align*}
\end{proof}

\newpage

\bibliographystyle{plain}
\bibliography{stoc}

\end{document}